\documentclass[11pt]{article}

\usepackage[sectionbib]{natbib}
\usepackage{algorithm}
\usepackage{algorithmic}
\usepackage{booktabs}
\usepackage{lineno}
\usepackage{graphicx}
\usepackage{color,xcolor}
\usepackage{subfigure} 
\usepackage{geometry}
\usepackage{multirow}

\usepackage[colorlinks=true,linkcolor=blue,citecolor=blue]{hyperref}%

\usepackage{amsmath}
\usepackage{amssymb}
\usepackage{amsthm}

\usepackage[utf8]{inputenc} 
\usepackage{hyperref}       
\usepackage{url}            
\usepackage{booktabs}       
\usepackage{amsfonts}       
\usepackage{nicefrac}       
\usepackage{microtype}      

\usepackage{natbib}

\usepackage{array,epsfig,rotating}
\usepackage{tabularx}
\usepackage{subfigure}
\usepackage{graphicx}
\usepackage{verbatim}
\usepackage{arydshln}
\usepackage{authblk}

\newtheorem{theorem}{Theorem}
\newtheorem{lemma}{Lemma}

\newtheorem{corollary}{Corollary}

\newtheorem{remark}{Remark}

\newtheorem{assumption}{Assumption}

\allowdisplaybreaks

\def\AA{\mathbf A}

\def\hat{\widehat}
\def\tilde{\widetilde}

\newcommand{\ttt}{\boldsymbol \theta}

\newcommand{\BB}{\mathbf B}

\newcommand{\XX}{\mathbf X}
\newcommand{\MM}{\mathbf M}
\newcommand{\xx}{\mathbf x}
\newcommand{\ZZ}{\mathbf Z}

\newcommand{\zz}{\mathbf z}

\newcommand{\UU}{\mathbf U}
\newcommand{\OO}{\mathbf O}
\newcommand{\VV}{\mathbf V}
\newcommand{\WW}{\mathbf W}
\newcommand{\II}{\mathbf I}

\newcommand{\uu}{\mathbf u}
\newcommand{\vv}{\mathbf v}
\newcommand{\ww}{\mathbf w}
\newcommand{\bb}{\mathbf b}

\newcommand{\ME}{\mathbb E}
\newcommand{\MP}{\mathbb P}
\newcommand{\MU}{\mathbb U}

\newcommand\independent{\protect\mathpalette{\protect\independenT}{\perp}}
\def\independenT#1#2{\mathrel{\rlap{$#1#2$}\mkern2mu{#1#2}}}

\def\spacingset#1{\renewcommand{\baselinestretch}%
{#1}\small\normalsize} \spacingset{1}
\spacingset{1}

\usepackage{float}

\begin{document}



\title{Efficient Generalization and Transportation
}
\author{Zhenghao Zeng$^1$, Edward H.\ Kennedy$^1$, Lisa M.\ Bodnar$^2$, Ashley I.\ Naimi$^3$}
\affil{$^1$Department of Statistics \& Data Science, Carnegie Mellon University \\
$^2$Department of Epidemiology, University of Pittsburgh \\
$^3$Department of Epidemiology, Emory University}

  \date{}

\maketitle

\begin{abstract}
When estimating causal effects, it is important to assess external validity, i.e., determine how useful a given study is to inform a practical question for a specific target population. One challenge is that the covariate distribution in the population underlying a study may be different from that in the target population. If some covariates are effect modifiers, the average treatment effect (ATE) may not generalize to the target population. To tackle this problem, we propose new methods to generalize or transport the ATE from a source population to a target population, in the case where the source and target populations have different sets of covariates. When the ATE in the target population is identified, we propose new doubly robust estimators and establish their rates of convergence and limiting distributions. Under regularity conditions, the doubly robust estimators provably achieve the efficiency bound and are locally asymptotic minimax optimal.  A sensitivity analysis is provided when the identification assumptions fail. Simulation studies show the advantages of the proposed doubly robust estimator over simple plug-in estimators. Importantly, we also provide minimax lower bounds and higher-order estimators of the target functionals. The proposed methods are applied in transporting causal effects of dietary intake on adverse pregnancy outcomes from an observational study to \textit{the whole U.S. pregnant female population}.

\end{abstract}


\def\spacingset#1{\renewcommand{\baselinestretch}%
{#1}\small\normalsize} \spacingset{1}
\spacingset{1}

\section{Introduction} \label{Intro}

In causal inference, researchers typically aim at estimating causal effects in a particular ``source'' population based on randomized trials or observational studies in that population.  population is the sole population of interest, random samples from the source are representative, and standard techniques to estimate average treatment effects (ATEs) can be applied to obtain reliable estimates of effects in the whole population. However, in many studies (e.g. randomized trials in clinical medicine, policy analysis) samples are drawn from a different population, unrepresentative of the target of interest \citep{kennedy2015literature, bell2016estimates, allcott2015site}. In other words, the source and target populations can be different, and the ATE we obtain from the sample only applies to the source population -- and may not generalize to the target population directly. Failing to consider this lack of representation may yield unreliable conclusions that can even be harmful, especially in medicine and policy evaluations \citep{westreich2019target,chen2021ethical}.

One example of unrepresentative samples is when the distributions of some covariates (e.g. age and BMI) in the target population differ from those in the source population. When some of these covariates are effect modifiers (e.g. age and BMI may modify the effects of some medicine), the ATE in the target population can be quite different from that in the source population. There exists a rich literature in bridging findings from the source population to the target population \citep{pearl2011transportability, cole2010generalizing, tipton2013improving, rudolph2017robust, buchanan2018generalizing, dahabreh2019generalizing, dahabreh2020extending, chattopadhyay2022one}. Most of these papers adopt the idea that we first estimate the probability for a subject to be in the source population and use it as a reweighting term in the proposed estimator. For instance, \cite{dahabreh2019generalizing} and \cite{dahabreh2020extending} proposed three types of estimators based on outcome modeling, inverse probability weighting and doubly robust-style augmented inverse probability weighting. \cite{rudolph2017robust} derived efficient influence functions and robust estimators for three transported estimands: ATE, intent-to-treat ATE and complier ATE. \cite{chattopadhyay2022one} developed a one-step weighting estimator where the weights are learned from a convex optimization problem to simultaneously model the inverse propensity score and outcome regression functions. However, the theoretical properties of the doubly robust estimators have not been well understood. The form of their second-order bias and the conditions required to ensure their asymptotic normality are unclear. Another limitation is that the previous work assumes the sets of covariates in the source population and the target population are the same. 
In practice there often exist some covariates that are available in the source, but not in the target. Simply ignoring them is not an efficient way to exploit the samples. For instance, including more covariates in the source may enable us to be more confident in the conditional exchangeability assumption in the source population.

In this work, we fill these two gaps by developing efficiency theory and establishing theoretical guarantees for the doubly robust estimators. In particular, we generalize previous results to accommodate covariate mismatch in the source and target populations. After formulating the generalization and transportation functional of interest, we examine the assumptions required to identify these functionals. When these functionals are identified, we derive the first-order influence functions and establish the asymptotic normality of a doubly robust estimator (under additional conditions). Simulations show how the proposed doubly robust estimator has a smaller estimation error compared to a simple plug-in estimator. To provide a complete story for generalization and transportation of causal effects, we then consider the minimax lower bound and higher-order estimation when the source population and the target population share the same set of covariates. We derive the minimax lower bound and propose a new higher-order estimator based on second-order influence functions, which attains the minimax lower bound in a broader regime than the doubly robust estimator. To the best of our knowledge, the minimax lower bound and high-order estimation have not been studied under generalization and transportation settings. Finally we apply our proposed methods to transporting the causal effects of fruit and vegetable intake on adverse pregnancy outcomes from an observational study to \textit{the whole U.S. pregnant female population} as an illustration.

The paper is organized as follows. In Section \ref{sec:Preliminaries} we introduce the setup and notation. In Section \ref{sec:Identification} we discuss the identification assumptions to identify the ATE in the target population. Efficiency theory and doubly robust methods are provided in Section \ref{sec:Efficiency-theory}. We then summarize the minimax lower bounds of the target functionals in Section \ref{sec:Minimax} and derive higher-order estimators in Section \ref{sec:Quadratic}. Simulation studies are presented in Section \ref{sec:Simulation} to explore finite-sample properties of our methods. In Section \ref{sec:Real-data} we provide a data analysis to illustrate the proposed methods, transporting effects of dietary intake on pregnancy outcomes. Finally we conclude with a discussion in Section \ref{sec:Discussion}. All the proofs and additional details on real data analysis are presented in the supplementary materials.

\section{Preliminaries}\label{sec:Preliminaries}

In this section we first introduce the setup in generalization and transportation setting. We then formalize the treatment effects of interests as statistical functionals with the potential outcome framework \citep{splawa1990application, rubin1974estimating}. Finally we introduce some notation that will be useful in presenting our results.

\subsection{Data Structure}

In the generalization and transportation setting there are typically two populations, i.e., a source population and a target population. The source population is usually the underlying population of a randomized trial or observational study and is defined by enrollment processes and inclusion
or exclusion criteria of the study. We assume we observe $n_1$ source samples 
$$\mathcal{D}_1 = \{ \ZZ_i = (\XX_i,A_i,Y_i, S_i = 1), 1\leq i \leq n_1 \}$$ 
from the source population, where $\XX$ is a $d$-dimensional vector containing all covariates, $A$ is the treament assignment and $Y$ is the outcome. We denote $S$ as an indicator such that $S=1$ if a subject is in the source population and $S=0$ otherwise. We are interested in estimating the ATE in a target population without direct access to treatment and outcome information on target samples. The observational unit in the target population is $\ZZ = (\VV, S=0)$ and the target dataset $\mathcal{D}_2$ consists of $n_2$ such realizations, i.e., 
$$\mathcal{D}_2  = \{\ZZ_i = (\VV_i, S_i=0), n_1+1 \leq i \leq n_1 +n_2=n \} , $$
where, importantly, $\VV \subseteq \XX$ represents partial covariates in $\XX$, which may be a strict subset. Here we slightly abuse the notation and denote both the covariates vector and set of covariates as $\VV$ (in the target population) and $\XX$ (in the source population). In practice, there may be variables present in the target population but not in the source population. Since the information about the treatment effect in the source population is captured by $\XX$, covariates that are not included in $\XX$ may not help us identify and estimate the ATE in the target population. Therefore, we do not use these variables to generalize or transport the causal effects and assume $\VV \subseteq \XX$. 

\subsection{Effects of Interest}

Before defining our estimands of interest, we first discuss differences in the target population definition in the generalization and transportation setting. In the generalization setting, the source population is a subset of the target population. After we collect treatment and outcome information and estimate the ATE in the source population, we hope to \emph{generalize} the causal effects to the whole target population. The most natural design is a trial nested within a cohort of eligible individuals \citep{dahabreh2019generalizing}. In such designs, researchers collect covariate information for all individuals, but only collect treatment and outcome information from a subset of them. This setup arises in comprehensive cohort studies \citep{olschewski1985comprehensive}, where only a few patients consent to randomization in a clinical trial, as well as clinical trials embedded in healthcare systems where all individuals' information is collected routinely but only some of them are included in a trial \citep{fiore2016integrating}. 

In contrast, in the transportation setting, the source population is (at least partly) external to the target population \citep{cole2010generalizing}. The source population is different from the target population and source samples may not be representative of the target population. The goal is to \emph{transport} the causal effects from the source population to the target population. This setup arises widely in public policy research, where randomized trial or observational study is conducted in selected samples while the target samples are from administrative databases or surveys and can be very different from the samples enrolled in the study (our data example in Section \ref{sec:Real-data} falls into this category, as will be discussed in more detail shortly).

Now we define the generalization and transportation effects of interest. We use the random variable $Y^a$ to denote the potential (counterfactual) outcome we would have observed
had a subject received treatment $A = a$, which may be contrary to the fact $Y$. For simplicity, we consider binary treatment in this work, where $A=1$ means treatment and $A=0$ means control. Then the ATE in the target population is 
\begin{equation}\label{generalization-effect}
    \psi := \mathbb{E}[Y^1 - Y^0]
\end{equation}
in the generalization case and 
\begin{equation}\label{transportation-effect}
    \theta := \mathbb{E}[Y^1 - Y^0|S=0]
\end{equation}
in the transportation case. Note that in the generalization case the source population is part of the target population and hence we take an expectation over the whole population. However, in the transportation case the source population may not be representative of the target population, and we only take an expectation in the target population (i.e. conditioning on $S=0$).

Under standard causal assumptions (consistency, positivit \citep{rosenbaum1983central}, exchangeability \citep{hernan2023causal}), the ATE in the source population $\mathbb{E}[Y^1 - Y^0|S=1]$ can be identified and efficiently estimated. Mathematically, we do not have $\mathbb{E}[Y^1 - Y^0|S=0] = \mathbb{E}[Y^1 - Y^0|S=1] = \mathbb{E}[Y^1 - Y^0]$ in general, especially when the effect modifiers have different distributions in the source and the target population. Hence we need additional assumptions and novel methodology to estimate the treatment effect in the target population. We define the mean potential outcomes in the target population as
\begin{equation}\label{theta_a}
    \psi_a = \mathbb{E}[Y^a], \quad \theta_a = \mathbb{E}[Y^a|S=0]
\end{equation}
so $\psi_a$ and $\theta_a$ correspond to the generalization and transportation cases, respectively.
In this paper we will focus on the identification and estimation of $\psi_a$ and $\theta_a$. After estimating $\psi_a$ or $\theta_a$, the ATE in the target population can be estimated by taking the difference. Other estimands, such as the risk ratio or odds ratio, can be estimated in a similar manner.

\subsection{Nuisance Functions \& Other Notation}

To present our results concisely we need the following notation on commonly used nuisance functions. 
\begin{itemize}
    \item The propensity score in the source population is $\pi_a(\xx) = \MP(A = a | \XX = \xx, S = 1)$. If the source dataset $\mathcal{D}_1$ is from a randomized trial then $\pi_a(\xx)$ can be known. Otherwise we would need to estimate it from $\mathcal{D}_1$.
    \item The conditional probability of being selected into the source population (commonly referred to as the participation probability) is $\rho(\vv) = \MP (S=1|\VV=\vv)$.
    \item The conditional mean and variance of the outcomes among subjects receiving treatment $A=a$ in the source population are $\mu_a(\xx) = \mathbb{E}[Y \mid \XX=\xx, A=a, S=1]$ and $\sigma_a^2(\xx) = \text{Var}(Y \mid \XX=\xx, A=a, S=1) $.
    \item The function obtained by further regressing $\mu_a$ on $\VV$ in the source population, $\tau_a(\vv) = \ME[\mu_a(\XX)|\VV=\vv, S=1]$.
    
\end{itemize}

For a univariate function $f$ on variables $\ZZ$ we use $\MP_n [f(\ZZ)] $ or $\MP_n (f)$to denote the sample average $\frac{1}{n}\sum_{i=1}^n f(\ZZ_i)$. For a bivariate function $g$ we use $\MU_n [g(\ZZ_1,\ZZ_2)]$ or $\MU_n (g)$ to denote the U-statistic measure $ \frac{1}{n(n-1)} \sum_{i \neq j}g(\ZZ_i, \ZZ_j)$. The Hellinger distance $H^2(\MP,\mathbb{Q})$ between two distributions $\MP$ and $\mathbb{Q}$ is defined as
\[
H^2(\MP, \mathbb{Q})=\frac{1}{2} \int \left[ \sqrt{p(\xx)}-\sqrt{q(\xx)}\right]^2 \nu(d \xx)
\]
for a dominating measure $\nu$. We say a function $f$ is $s$-smooth if it is $\lfloor s\rfloor$ times continuously differentiable with derivatives up to order $\lfloor s\rfloor$ bounded by some constant $C>0$ and $\lfloor s\rfloor$-order derivatives
Hölder continuous, i.e.
\[
\left|D^\beta f(\xx)-D^\beta f\left(\xx^{\prime}\right)\right| \leq C\left\|\xx-\xx^{\prime}\right\|_2^{s-\lfloor s\rfloor}
\]
for all $\beta = (\beta_1,\dots, \beta_d)$ with $\sum_{i} \beta_i = \lfloor s\rfloor$, where $D^\beta=\frac{\partial^\beta}{\partial x_1^{\beta_1} \ldots \partial x_d^{\beta_d}}$ is the differential operator. Hölder class, denoted by $\mathcal{H}(s)$, is the function class containing all $s$-smooth functions.

We denote the weighted $L_2$ norm with weight function $w$ as $\|f\|_w = \sqrt{\int f(\zz)^2w(\zz) d \MP (\zz)}$ and when the weight function $w=1$ we abbreviate the notation as $\|f\|$. In this paper we mainly use $w = \rho \pi_a$ as a weight function. For a matrix $\mathbf{\Omega}$ we let $\|\mathbf{\Omega}\|$ and $\|\mathbf{\Omega}\|_F$ denote its spectral norm and Frobenius norm, respectively. We write $a_n \lesssim b_n$ if $a_n \leq Cb_n$ for a positive constant $C$ and sufficiently large $n$.

\section{Identification} \label{sec:Identification}
In this section we discuss sufficient conditions to identify functionals in \eqref{theta_a} from the observable data. These assumptions are generalizations of the identification conditions used in \cite{dahabreh2019generalizing} and \cite{dahabreh2020extending} to the case $\VV\subseteq \XX$. (Note in Section \ref{sec:Sensitivity-analysis} we consider sensitivity analysis and allow several of the following assumptions to be violated.)

\begin{assumption}\label{consistency}
  Consistency: 
$Y = Y^a \text{\, if \,} A=a.$

\end{assumption}
\begin{assumption}\label{exchangeability}
No unmeasured confounding in source: 
$(Y^0, Y^1) \independent A \mid \XX, S=1.$
\end{assumption}
\begin{assumption}\label{positivity-treatment}
Treatment positivity in source: $\pi_a(\XX) >0 \quad \MP(\cdot|S=1) \text{ a.s. for all  } a.$
\end{assumption}
Assumption \ref{consistency} is also known as stable unit treatment value assumption (SUTVA) and requires no interference between different subjects, i.e. the outcome for an individual is not affected by other individuals' treatments.
Assumption \ref{exchangeability} is a standard assumption used to identify average treatment effects. It holds if the source dataset comes from a randomized trial or if we collect enough covariates in $\XX$ so that the treatment process is completely explained by $\XX$.
Assumption \ref{positivity-treatment}, also known as the overlap assumption, has been used in causal inference since \cite{rosenbaum1983central}. It guarantees that every
subject in the source population has a positive probability of receiving each treatment $a$. With these three assumptions we are able to identify the average treatment effect in the source population. But we need additional assumptions for generalization and transportation.

\begin{assumption}\label{transportability}
Exchangeability between populations: $
S \independent Y^{a} \mid \VV.$
\end{assumption}

\begin{assumption}\label{positivity-selection}
Positivity of selection: $\rho(\VV)>0 \text{\, a.s.}$
\end{assumption}

Assumption \ref{transportability} is critical to generalizing/transporting the effects from the source population to the target population \citep{kern2016assessing, dahabreh2019generalizing, dahabreh2020extending}. Under assumption \ref{transportability} we have
\[
\mathbb{E}[Y^a|\VV,S=1] = \mathbb{E}[Y^a|\VV] = \mathbb{E}[Y^a|\VV,S=0],
\]
which further implies
\begin{equation}\label{cond_eff_tran}
    \mathbb{E}[Y^1-Y^0|\VV,S=1] = \mathbb{E}[Y^1-Y^0|\VV,S=0].
\end{equation}
Hence, the source population and the target
population have the same conditional average treatment effect. We essentially just need \eqref{cond_eff_tran} to identify the ATE in the target population. To state our results concisely we formalize the assumption in terms of potential outcome $Y^a$ instead of the contrast $Y^1-Y^0$. For equality \eqref{cond_eff_tran} to hold, all effect modifiers that are distributed differently between the source and the target populations must be measured in $\VV$. In practice, researchers can first identify effect modifiers from variables shared by both the source and target populations and include them in 
$\VV$. They can then identify additional covariates in the source population that explain the treatment assignment and include them in 
$\XX$, ensuring that Assumption \ref{exchangeability} is satisfied.
Assumption \ref{positivity-selection} requires that in each stratum of effect modifiers $\VV$, there is a positive probability of being in the source population for every individual. Thus all members in the target population are represented by some individuals in the source population.

\begin{theorem} \label{thm-identification}
Under identification assumptions \ref{consistency}--\ref{positivity-selection}, the estimands $\psi_a$ and $\theta_a$ are identified as 
\begin{equation}\label{eq:identify}
\begin{aligned}
    \psi_a = &\, \mathbb{E} \left\{\mathbb{E}\left[\mathbb{E}(Y \mid \XX, A=a, S=1) \mid \VV, S=1\right] \right\} \\
    = & \, \mathbb{E} \{\mathbb{E}[\mu_a(\XX) \mid \VV, S=1] \}= \ME [\tau_a(\VV)], \\
    \theta_a =&\, \mathbb{E} \left\{\mathbb{E}\left[\mathbb{E}(Y \mid \XX, A=a, S=1) \mid \VV, S=1\right] |S=0\right\} \\
    = & \, \mathbb{E} \{\mathbb{E}[\mu_a(\XX) \mid \VV, S=1] |S=0\}= \ME [\tau_a(\VV)|S=0].
\end{aligned}
\end{equation}
\end{theorem}

We can understand the above functionals by evaluating the three iterative expectations. First we regress the outcome $Y$ on $\XX$ among subjects receiving treatment $A=a$ in the source population and obtain $\mu_a(\xx)$, which contains information on the conditional treatment effect. Then we further regress $\mu_a$ on effect modifiers $\VV$ in the source population to obtain $\tau_a(\VV)$, which summarizes the conditional treatment effects within the subset of covariates $\VV$. Finally we take the mean of $\tau_a$ in the target population and obtain the target functional. The validity of the last step is guaranteed by Assumption \ref{transportability}, which implies the information on treatment effects contained in $\tau_a$ generalizes to the whole population. The proof of identifiability is provided in the appendix. 

Note that the mean potential outcome in the source population can be written as 
\[
\ME [Y^a \mid S=1] = \ME [\tau_a(\VV) \mid S=1].
\]
It differs from the target functional only in the last step, where we average over $\tau_a$ in the target population instead of in the source population. When the distributions of effect modifiers $\VV$ in the source population and the target populations are different, we have $\ME [\tau_a(\VV) \mid S=1] \neq \ME [\tau_a(\VV) \mid S=0]$ and thus the treatment effect in the source population may not generalize to the target population directly.

In applications, one needs to carefully assess the five assumptions above. In general, these assumptions are untestable and their plausibility needs to be evaluated from substantive knowledge on the mechanism of treatment assignment and study participation. If some assumptions are likely to be violated, researchers should perform sensitivity analysis to assess the robustness of their results. We discuss one way of performing sensitivity analysis in the following section.

\subsection{Sensitivity Analysis}\label{sec:Sensitivity-analysis}

When the identification assumptions do not hold simultaneously, it is not guaranteed that identification results \eqref{eq:identify} hold and we cannot identify the target functionals from the observed data. For instance, the transportability assumption may not hold because there remain unmeasured effect modifiers. In this setting, \cite{colnet2022causal} analyzed the bias due to missingness of important effect modifiers in either the source population or target population under a linear conditional average treatment effect model. However, we can still derive bounds on causal effects under other assumptions (e.g. we may relax conditions required to identify the target functionals \citep{luedtke2015statistics} or construct bounds via available instrumental variable \citep{balke1997bounds, levis2023covariate}). In this section we relax exchangeability and transportability assumptions to derive bounds that provide us with a range of plausible values for the ATE and can be useful in some applications.

\begin{assumption}\label{relax-exchangeability}
Relaxation of Assumption \ref{exchangeability}: There exists a positive constant $\delta_1$ such that 
\[
|\mathbb{E}\left[Y^{a} \mid \XX, A=1, S=1\right] - \mathbb{E}\left[Y^{a} \mid \XX, A=0, S=1\right]| \leq \delta_1 \text{ a.s. for all } a.
\]
\end{assumption}

\begin{assumption}\label{relax-transportability}
Relaxation of Assumption \ref{transportability}: There exists a positive constant $\delta_2$ such that 
\[
|\mathbb{E}[Y^a|\VV,S=0] - \mathbb{E}[Y^a|\VV,S=1]| \leq \delta_2 \text{ a.s. for all } a.
\]
\end{assumption}
We note that when Assumption \ref{exchangeability} and Assumption \ref{transportability} hold, we have $\delta_1 = \delta_2 = 0$. Hence Assumption \ref{relax-exchangeability} and Assumption \ref{relax-transportability} are indeed relaxations of Assumption \ref{exchangeability} and Assumption \ref{transportability}. In practice, the value $\delta_1$ and $\delta_2$ may come from the domain knowledge of the problem of interests. The following theorem characterizes the bounds on the ATE in the target population when the treatment is binary.

\begin{theorem}\label{thm-sensitivity}
Under Assumption \ref{consistency}, \ref{positivity-treatment}, \ref{positivity-selection}, \ref{relax-exchangeability} and \ref{relax-transportability}, we have 
\begin{equation*}
    \begin{aligned}
    \psi_a &\,\in \left[\ME [\tau_a(\VV)] -\delta_1 \ME[\MP(A=1-a \mid \VV,S=1)] -\delta_2 \MP(S=0), \right.\\
    &\, \left. \ME [\tau_a(\VV)] +\delta_1 \ME[\MP(A=1-a \mid \VV,S=1)] +\delta_2 \MP(S=0)\right]. \\
    \theta_a &\,\in \left[\ME [\tau_a(\VV)|S=0] -\delta_1 \ME[\MP(A=1-a \mid \VV,S=1) \mid S=0] -\delta_2, \right.\\
    &\, \left. \ME [\tau_a(\VV)|S=0] +\delta_1 \ME[\MP(A=1-a \mid \VV,S=1) \mid S=0] +\delta_2\right].
    \end{aligned}
\end{equation*}
Hence the ATE in the generalization functional $\psi_1 - \psi_0$ is in the interval
\[
\left[ \ME[\tau_1(\VV) - \tau_0(\VV)] -\delta_1 -2\delta_2 \MP(S=0), \ME[\tau_1(\VV) - \tau_0(\VV)] +\delta_1 +2\delta_2 \MP(S=0)  \right],
\]
the ATE in the transportation functional $\theta_1 - \theta_0$ is in the interval
\[
\left[\ME [\tau_1(\VV)-\tau_0(\VV)|S=0] - \delta_1-2\delta_2, \ME [\tau_1(\VV)-\tau_0(\VV)|S=0] + \delta_1+2\delta_2\right].
\]
\end{theorem}

Since the efficiency theory in Section \ref{sec:DR-estimation} directly holds for $\ME [\tau_a(\VV)]$ and $[\ME [\tau_a(\VV)|S=0]$, one can use the doubly robust estimator to estimate them efficiently. When the specific values of $\delta_1$ and $\delta_2$ are available, one can directly construct the bounds in Theorem \ref{thm-sensitivity}. If exact domain knowledge on the precise value of $\delta_1$ and $\delta_2$ is unavailable, we can estimate the bounds as a function of $(\delta_1,\delta_2)$, and for example obtain which values of $\delta_1$ and $\delta_2$ substantially change results (e.g., flip the sign of the treatment effects). For instance, when we are confident about Assumption \ref{exchangeability} (e.g., the source dataset comes from a randomized experiment), we can set $\delta_1 =0$. Then the value for $\delta_2$ changes the sign of the effect is $|\ME [\tau_1(\VV)-\tau_0(\VV)|S=0]|/2$. This value can reflect the robustness of our results when the identification assumptions do not necessarily hold.

\section{Efficiency Theory and Doubly Robust Estimation}\label{sec:Efficiency-theory}

In this section we develop the nonparametric theory for estimation of the ATE in the target population. Namely we first derive the efficient influence function, together with the nonparametric efficiency bound. The nonparametric efficiency bound provides a benchmark for efficient estimation in a nonparametric model, indicating the best possible performance in a local asymptotic minimax sense \citep{van2000asymptotic}. Next we propose a doubly robust estimator of $\psi_a$ and $\theta_a$ based on the influence function, which is shown to be asymptotically normal and attain the efficiency bound under additional mild conditions.  

\subsection{Efficient Influence Function and Efficiency Bound}\label{sec:EIF-EB}

We first introduce the problem faced by the plug-in estimator and motivate the study of efficient influence functions. Denote the plug-in estimators of nuisance functions $(\mu_a, \tau_a, \pi_a, \rho)$ as $(\hat{\mu}_a, \hat{\tau}_a, \hat{\pi}, \hat{\rho})$. Based on the identification result for $\psi_a$ after equation \eqref{eq:identify}, a plug-in estimator of $\psi_a$ is then given by
\[
\hat{\psi}_a = \MP_{n} (\hat{\tau}_a) = \frac{1}{n} \sum_{i=1}^{n} \hat{\tau}_a(\VV_i). 
\]
This plug-in estimator would be $\sqrt{n}$-consistent if we used correct parametric models to estimate the nuisance function $\tau_a$. However, there is generally not sufficient background knowledge to ensure correct
specification of such parametric models; thus analysts often use flexible non-parametric methods to avoid model misspecification. However, under such circumstances, the conditional bias of the plug-in estimator is of order $\|\hat{\tau}_a - \tau_a\|$, which is typically slower than the $\sqrt{n}$-rate, perhaps much slower when the nuisance functions are complex and the number of covariates is large. Hence the plug-in estimator generally suffers from slow convergence rates and a lack of tractable limiting distributions. These drawbacks make it difficult to estimate $\psi_a$ accurately and perform statistical inference with plug-in estimators. 

To address these difficulties, one can derive the efficient influence functions of the target functionals. The efficient influence function is critical in non-parametric efficiency theory \citep{bickel1993efficient, tsiatis2006semiparametric, van2000asymptotic,van2003unified, kennedy2022semiparametric}. Mathematically, the influence function is the derivative in a
Von Mises expansion (i.e., distributional Taylor expansion) of the target statistical functional. In the discrete case, it coincides with the Gateaux derivative of the functional when the contamination distribution is a point-mass. Influence functions are important in the following respects. First, the variance of the influence function is equal to the efficiency bound of the target statistical functional, which characterizes the inherent estimation difficulty of the target functional and provides a benchmark to compare against when we construct estimators. Moreover, it allows us to correct for first-order bias in the plug-in estimator and obtain doubly robust-style estimators, which enjoy appealing statistical properties even if non-parametric methods with relatively slow rates are used in nuisance estimation.

In the following discussions, we will first present the efficient influence function of $\psi_a$ and $\theta_a$. We then derive the doubly robust estimator and establish its $\sqrt{n}$-consistency and asymptotic normality under appropriate conditions. The efficient influence functions and efficiency bounds are summarized in the following results.

\begin{lemma}\label{lem-if}
Under an unrestricted nonparametric model, the efficient influence function of $\psi_a$ is given by 
\[
\phi_a^{ge} (\ZZ) = \frac{I(A=a, S=1)(Y-\mu_a(\XX))}{\rho(\VV) \pi_a(\XX)} + \frac{I(S=1)(\mu_a(\XX) - \tau_a(\VV))}{\rho(\VV)} + \tau_a(\VV) - \psi_a
\]
and the efficient influence function of $\theta_a$ is given by
\begin{equation*}
	\begin{aligned}
        \phi_{a}^{\text{tr}} (Z) &=\frac{1}{\MP(S=0)}\left\{\frac{I(A=a, S=1) (1 - \rho(\VV))\left(Y-\mu_{a}(\XX)\right)}{\rho (\VV) \pi_a(\XX)}\right.\\
        &\left.+\frac{I(S=1) (1 - \rho(\VV))\left(\mu_{a}(\XX)-\tau_{a}(\VV)\right)}{\rho(\VV)} +I(S=0)\left[\tau_{a}(\VV)-\theta_{a}\right]\right\}.
        \end{aligned}
\end{equation*}
\end{lemma}

\begin{theorem}\label{thm-if}
The nonparametric efficiency bound of $\psi_a$ is 
\begin{equation*}
    \begin{aligned}
        \sigma_{a, \text{ge}}^2   &=\mathbb{E}\left[\frac{\MP(S=1|\XX)  \operatorname{Var}(Y \mid \XX, A=a, S=1)}{\rho^2 (\VV) \pi_a(\XX)}\right]\\
        &+\mathbb{E}\left[\frac{ \operatorname{Var}(\mu_a(\XX)|\VV,S=1)}{\rho(\VV)}\right] + \operatorname{Var}(\tau_a(\VV)) 
    \end{aligned}
\end{equation*}
and the nonparametric efficiency bound of $\theta_a$ is 
\begin{equation*}
    \begin{aligned}
        \sigma_{a, \text{tr}}^2 &=\frac{1}{\MP^2(S=0)}\left\{\mathbb{E}\left[\frac{\MP(S=1|\XX) (1 - \rho(\VV))^2  \operatorname{Var}(Y \mid \XX, A=a, S=1)}{\rho^2 (\VV) \pi_a(\XX)}\right]\right.\\
        &\left.+\mathbb{E}\left[\frac{(1 - \rho(\VV))^2 \operatorname{Var}(\mu_a(\XX)|\VV,S=1)}{\rho(\VV)}\right] + \MP(S=0)\operatorname{Var}(\tau_a(\VV)|S=0) \right\}.
    \end{aligned}
\end{equation*}
\end{theorem}

The efficiency bounds in Theorem \ref{thm-if} show how particular nuisance quantities determine the estimation difficulty of our target functionals. Specifically, the efficiency bounds of $\psi_a$ and $\theta_a$ both depend on 
\begin{itemize}
    \item The inverse propensity score $1/\pi_a(\XX)$ measuring how likely an individual will receive treatment $A=a$. 
    \item The conditional variance $\sigma_a^2(\xx) =\text{Var}(Y \mid \XX=
    \xx, A=a, S=1)$, which measures how much variation of $Y$ can be explained by $\XX$ for subjects receiving treatment $a$ in the source population. 
    \item The conditional variance $ \text{Var}(\mu_a(\XX) \mid \VV=\vv, S=1)$, which measures how much variation of $\mu_a(\XX)$ can be explained by $\VV$ for subjects in the source population.
    \item The variance of $\tau_a(\VV)$ in the target population, i.e. $\text{Var}(\tau_a(\VV))$ in the generalization case and $\text{Var}(\tau_a(\VV)|S=0)$ in the transportation case.
\end{itemize}

There are also some differences in the efficiency bounds of the two functionals. First, the efficiency bound in the transportation case depends on the probability of being in the target population $\MP (S=0)$ explicitly. Moreover, the first term in the efficiency bound $\sigma_{a,\text{ge}}^2$ depends on $\rho(\VV)$ via its reciprocal while the first term in $\sigma_{a,\text{tr}}^2$ depends on the inverse odds of being in the source population $(1-\rho(\VV))/\rho(\VV)$. This implies the inverse odds ratio $(1-\rho(\VV))/\rho(\VV)$ may be a more fundamental quantity than $1/\rho(\VV)$ in transportation problems, and we will see this phenomenon in Section \ref{sec:Quadratic} as well. In addition to parameterizing the inverse odds of being in the source population through $\rho$, one can also express it as:
\[
\frac{1-\rho(\VV)}{\rho(\VV)} = \frac{\MP(S=0)\MP(\VV\mid S=0)}{\MP(S=1)\MP(\VV \mid S=1)}
\]
and parameterize the density ratio $\MP(\VV\mid S=0)/\MP(\VV \mid S=1)$. Then in the estimation step, the inverse odds ratio can be obtained by directly estimating the density ratio \citep{qin1998inferences, sugiyama2012density}.

The efficient influence functions in Theorem \ref{thm-if} generalize those of \citet{dahabreh2019generalizing, dahabreh2020extending} to the setting where there can be a mismatch between the covariates $\VV$ in the target population, and the covariates $\XX$ in the source population. To be concrete, in the special case $\VV=\XX$ the efficient influence functions of $\psi_a$ and $\theta_a$ are

\[
\phi_a^{ge} (\ZZ) = \frac{I(A=a, S=1)(Y-\mu_a(\XX))}{\rho(\XX) \pi_a(\XX)} +  \mu_a(\XX) - \psi_a,
\]
\[
\phi_a^{tr} (\ZZ) = \frac{1}{\MP(S=0)} \left\{\frac{I(A=a, S=1)(1-\rho(\XX))(Y-\mu_a(\XX))}{\rho(\XX) \pi_a(\XX)} + I(S=0)( \mu_a(\XX) - \theta_a) \right\}
\]
and the corresponding efficiency bounds are
\[
\sigma_{a, \text{ge}}^2   =\mathbb{E}\left[\frac{  \operatorname{Var}(Y \mid \XX, A=a, S=1)}{\rho (\XX) \pi_a(\XX)}\right]\\
 + \operatorname{Var}(\mu_a(\XX)) 
\]
\[
\sigma_{a, \text{tr}}^2 =\frac{1}{\MP^2(S=0)}\left\{\mathbb{E}\left[\frac{ (1 - \rho(\XX))^2  \operatorname{Var}(Y \mid \XX, A=a, S=1)}{\rho (\XX) \pi_a(\XX)}\right] + \MP(S=0)\operatorname{Var}(\mu_a(\XX)|S=0) \right\}.
\]
These results can be derived separately starting from the functional $\psi_a = \ME[\ME(Y|\XX,A=a,S=1)]$ and $\theta_a = \ME[\ME(Y|\XX,A=a,S=1)|S=0]$. Alternatively, one can set $\VV=
\XX$ in Lemma \ref{lem-if} and Theorem \ref{thm-if} to obtain the results by noting the second terms vanish in each formula due to $\mu_a(\XX) = \tau_a(\VV)$ and $\text{Var}(\mu_a(\XX)|\XX,S=1) = 0$. From this perspective, the second terms in the general case $\VV \subset \XX$ come from an extra step where we regress the conditional mean $\mu_a(\XX)$ on possible effect modifiers $\VV$. 

As we mentioned above the efficiency bound characterizes the fundamental statistical difficulty of estimating the target functionals, and acts as a nonparametric analog of the Cramer-Rao bound. Specifically, no estimator can have a smaller mean square error than the efficiency bound in a local asymptotic minimax sense, as summarized in the following Corollary \ref{cor:local-minimax}.  

\begin{corollary}\label{cor:local-minimax}
For any estimators $\hat{\psi}_a$ and $\hat{\theta}_a$, we have
\begin{equation*}
    \begin{aligned}
        &\, \inf _{\delta>0} \liminf _{n \rightarrow \infty} \sup _{\mathbb{Q}:\text{TV}(\MP, \mathbb{Q})<\delta} n \mathbb{E}_{\mathbb{Q}}\left[\{\widehat{\psi}_a-\psi_a(\mathbb{Q})\}^2\right] \geq  \sigma_{a, \text{ge}}^2 (\MP)\\
        &\, \inf _{\delta>0} \liminf _{n \rightarrow \infty} \sup _{\mathbb{Q}:\text{TV}(\MP, \mathbb{Q})<\delta} n \mathbb{E}_{\mathbb{Q}}\left[\{\widehat{\theta}_a-\theta_a(\mathbb{Q})\}^2\right] \geq  \sigma_{a, \text{tr}}^2 (\MP)
    \end{aligned}
\end{equation*}
where $\text{TV}(\MP, \mathbb{Q})$ is the total variation distance between $\MP$ and $\mathbb{Q}$ and $\sigma_{a, \text{ge}}^2 (\MP)$ and $\sigma_{a, \text{tr}}^2 (\MP)$ are the nonparametric efficiency bounds in Theorem \ref{thm-if} evaluated at $\MP$.
\end{corollary}

We have characterized the efficiency bounds with efficient influence functions, which implies that (without further assumptions) the asymptotic local minimax mean squared error of any estimator scaled by a factor of $n$ cannot be smaller than these bounds \citep{van2000asymptotic}. The next step is to correct for the first-order bias of plug-in-style estimators by instead deriving doubly robust estimators, as detailed in the following subsection.

\subsection{Doubly Robust Estimation}\label{sec:DR-estimation}

After deriving the influence functions, we can correct for the first-order bias of the plug-in estimator via the following doubly robust estimators:
\begin{equation}\label{dr-est-generalization}
\hat{\psi}_{a}^{dr}  = \MP_n \left\{\frac{I(A=a, S=1) \left(Y-\hat{\mu}_{a}(\XX)\right)}{\hat{\rho} (\VV) \hat{\pi}_a(\XX)} +\frac{I(S=1) \left(\hat{\mu}_{a}(\XX)-\hat{\tau}_{a}(\VV)\right)}{\hat{\rho}(\VV)} +\hat{\tau}_{a}(\VV)\right\} , 
\end{equation}
\begin{equation}\label{dr-est-transportation}
\begin{aligned}
    \hat{\theta}_{a}^{dr}  &=\frac{1}{\hat{\MP}(S=0)}\MP_n \left\{\frac{I(A=a, S=1) (1 - \hat{\rho}(\VV))\left(Y-\hat{\mu}_{a}(\XX)\right)}{\hat{\rho} (\VV) \hat{\pi}_a(\XX)}\right.\\
    &\left.+\frac{I(S=1) (1 - \hat{\rho}(\VV))\left(\hat{\mu}_{a}(\XX)-\hat{\tau}_{a}(\VV)\right)}{\hat{\rho}(\VV)} +I(S=0)\hat{\tau}_{a}(\VV)\right\}.
\end{aligned}    
\end{equation}

The doubly robust estimators combine simple plug-in-style estimators with inverse-probability-weighted estimators to correct for the first-order bias. For instance, in the doubly robust estimator $\hat{\psi}_{a}^{dr}$ the last term $\MP_n[\hat{\tau}_a(\VV)]$ is the outcome regression-based plug-in estimator and the first two terms are centered inverse-probability-weighted terms motivating from the influence function in Lemma \ref{lem-if}. Compared with the well-known doubly robust estimator for the ATE 
\[
\MP_n \left \{ \frac{I(A=a)(Y-\hat{\mu}_a(\XX))}{\hat{\pi}_a(\XX)} + \hat{\mu}_a(\XX)  \right \}
\]
in the generalization and transportation setting the participation probability also needs to be modeled and incorporated into the reweighting terms. Moreover, an extra term appears in \eqref{dr-est-generalization} and \eqref{dr-est-transportation} due to further regressing $\mu_a(\XX)$ on $\VV$, as similarly discussed in Section \ref{sec:EIF-EB}.

For simplicity we define the uncentered influence function terms in the brackets above as
 \begin{equation*}
 \varphi_{a}^{\text{ge}} (\ZZ) =\frac{I(A=a, S=1) \left(Y-\mu_{a}(\XX)\right)}{\rho (\VV) \pi_a(\XX)} +\frac{I(S=1) \left(\mu_{a}(\XX)-\tau_{a}(\VV)\right)}{\rho(\VV)} +\tau_{a}(\VV).
\end{equation*}
\begin{equation*}
    \begin{aligned}
    \varphi_{a}^{\text{tr}} (\ZZ) &=\frac{I(A=a, S=1) (1 - \rho(\VV))\left(Y-\mu_{a}(\XX)\right)}{\rho (\VV) \pi_a(\XX)}\\
    &+\frac{I(S=1) (1 - \rho(\VV))\left(\mu_{a}(\XX)-\tau_{a}(\VV)\right)}{\rho(\VV)} +I(S=0)\tau_{a}(\VV).
    \end{aligned}
\end{equation*}

The following theorems characterize the properties of these new doubly robust estimators.

\begin{theorem}[Doubly robust estimation of generalization functional]\label{thm-dr-generalization}
Suppose the nuisance functions $(\hat{\mu}_a, \hat{\tau}_a ,\hat{\pi}_a, \hat{\rho})$ are estimated from a separate independent sample. Further assume our estimates satisfy $\left\|\hat{\varphi}_{a}^{\text{ge}}-\varphi_{a}^{\text{ge}}\right\|_{2}=o_{\MP}(1)$, and $\hat{\rho}(\VV), \hat{\pi}_a (\XX) \geq \epsilon >0$ for some positive constant $\epsilon$. 
Then we have
\[
\hat{\psi}_a^{\text{dr}} - \psi_a = \mathbb{P}_n (\phi_a^{\text{ge}}) + O_\MP\Big(\left\|\hat{\mu}_{a}-\mu_{a}\right\|\|\hat{\pi}_a - \pi_a\| + \|\hat{\rho}-\rho\|\left\|\hat{\tau}_a-\tau_{a}\right\| \Big) + o_{\MP}(1/\sqrt{n}) .
\]
If the nuisance estimators further satisfy the following convergence rate
\begin{equation*}
	   \begin{aligned}
&\left\|\hat{\mu}_{a}-\mu_{a}\right\|\|\hat{\pi}_a - \pi_a\|=o_{\MP}\left(1/\sqrt{n}\right), \\
&\|\hat{\rho}-\rho\|\left\|\hat{\tau}_a-\tau_{a}\right\|=o_{\MP}\left(1/\sqrt{n}\right),
\end{aligned}
\end{equation*}
then $\hat{\psi}_a^{\text{dr}}$ is $\sqrt{n}$-consistent and asymptotically normal, with asymptotic variance equal to the nonparametric efficiency bound equal to $\sigma_{a, \text{ge}}^2$ of Theorem \ref{thm-if}, and so also locally asymptotic minimax optimal in the sense of Corollary \ref{cor:local-minimax}.
\end{theorem}

\begin{theorem}[Doubly robust estimation of transportation functional]\label{thm-dr-transportation}
Suppose the nuisance functions $\hat{\mu}_a, \hat{\tau}_a ,\hat{\pi}_a, \hat{\rho}$ are estimated from a separate independent sample. Further assume our estimates satisfy $\left\|\hat{\varphi}_{a}^{\text{tr}}-\varphi_{a}^{\text{tr}}\right\|_{2}=o_{\MP}(1), \MP(S=0)>0, \hat{\rho}(\VV), \hat{\pi}_a (\XX) \geq \epsilon >0$ for some positive constant $\epsilon$. Then we have
\[
\hat{\theta}_a^{\text{dr}} - \theta_a = \mathbb{P}_n (\phi_a^{\text{tr}}) +O_\MP\Big(\left\|\hat{\mu}_{a}-\mu_{a}\right\|\|\hat{\pi}_a - \pi_a\| + \|\hat{\rho}-\rho\|\left\|\hat{\tau}_a-\tau_{a}\right\| \Big) + o_{\MP}(1/\sqrt{n}).
\]
If the nuisance estimators further satisfy the following convergence rate
\begin{equation*}
	   \begin{aligned}
&\left\|\hat{\mu}_{a}-\mu_{a}\right\|\|\hat{\pi}_a - \pi_a\|=o_{\MP}\left(1/\sqrt{n}\right), \\
&\|\hat{\rho}-\rho\|\left\|\hat{\tau}_{a}-\tau_{a}\right\|=o_{\MP}\left(1/\sqrt{n}\right),
        \end{aligned}
\end{equation*}
then $\hat{\theta}_a^{\text{dr}}$ is $\sqrt{n}$-consistent and asymptotically normal, with asymptotic variance equal to the nonparametric efficiency bound equal to $\sigma_{a, \text{tr}}^2$ of Theorem \ref{thm-if}, and so also locally asymptotic minimax optimal in the sense of Corollary \ref{cor:local-minimax}.
\end{theorem}
\begin{remark}
For simplicity, we assume all the nuisance estimators are constructed from a separate independent sample with the same size $n$ as the estimation sample over which $\MP_n$ takes an average. Using the same sample to both estimate the nuisance functions and average the (uncentered) influence functions may also yield similar results by further assuming empirical process assumptions to avoid overfitting. For instance, one can assume the nuisance functions and their estimates belong to a Donsker class and arrive at similar estimation guarantees. However, such assumptions are hard to verify in practice and sample splitting enables us to get rid of them: one can randomly split the data in folds and use different folds to estimate the nuisance functions and average the influence functions. To recover full sample size efficiency, one can swap the folds, repeat the same procedures and finally average the results, known as cross-fitting and commonly used in the literature \citep{bickel1988estimating, robins2008higher, zheng2010asymptotic, chernozhukov2018double, kennedy2020sharp}. In this paper all the results are based on a single split procedure, with the understanding that extending to procedures based on cross-fitting is straightforward. \\
\end{remark}
We note that in Theorem \ref{thm-dr-generalization} and \ref{thm-dr-transportation} we do not require that each individual nuisance function converges at $\sqrt{n}$-rate as might be required in the plug-in estimator case. The condition is instead on the product of convergence rates, i.e. $\left\|\hat{\mu}_{a}-\mu_{a}\right\|\|\hat{\pi}_a - \pi_a\| = o_{\MP}(1/\sqrt{n})$ and $\left\|\hat{\tau}_{a}-\tau_{a}\right\|\|\hat{\rho} - \rho\| = o_{\MP}(1/\sqrt{n})$. This shows the key property of doubly robust estimators: after we correct for the first-order bias, the error only involves second-order products and hence is ``doubly small". In applications, such conditions on the convergence rate are much easier to satisfy. Estimators like ours that have errors that involve multiple nuisance functions are sometimes referred to as multiply robust estimators \citep{tchetgen2012semiparametric}, since there are multiple ways in which the error term is $o_{\MP}(1/\sqrt{n})$. For instance, (1) quarter rate $\left\|\hat{\mu}_{a}-\mu_{a}\right\| = o_{\MP}(n^{-1/4})$ and $\|\hat{\pi}_a - \pi_a\| = O_{\MP}(n^{-1/4})$ or (2) $\left\|\hat{\mu}_{a}-\mu_{a}\right\| = o_{\MP}(1)$ and $\|\hat{\pi}_a - \pi_a\| = O_{\MP}(n^{-1/2})$ (e.g., we know exactly the parametric model for propensity score) both satisfy the condition; further, $n^{-1/4}$-style rates can be attained under appropriate smoothness, sparsity, or other structural assumptions. So we may apply flexible non-parametric methods (e.g. random forests) or high dimensional models (e.g. Lasso regression) to estimate the nuisance functions and still maintain the $\sqrt{n}$-consistency and asymptotic normality of our effect estimator; this is the main advantage of doubly robust estimators over plug-ins. In the special case $\VV=\XX$ (i.e. the source and the target population share the same sets of covariates), conditions in Theorem \ref{thm-dr-generalization} and Theorem \ref{thm-dr-transportation} are reduced to $\left\|\hat{\mu}_{a}-\mu_{a}\right\|\|\hat{\pi}_a - \pi_a\| = o_{\MP}(1/\sqrt{n})$ and $\left\|\hat{\mu}_{a}-\mu_{a}\right\|\|\hat{\rho} - \rho\| = o_{\MP}(1/\sqrt{n})$. Compared with the ATE case where we only require $\|\hat{\mu}_{a}-\mu_{a}\|\|\hat{\pi}_a - \pi_a\| = o_{\MP}(1/\sqrt{n})$, extra conditions on the convergence rates of modeling participation probability is needed in the generalization and transportation problem. 

We conclude this section with additional comments on estimating $\tau_a(\VV)$. In practice, a simple plug-in approach is to first estimate $\mu_a(\XX)$ for each data point and then further regress $\hat{\mu}_a$ on partial covariates $\VV$ in the source dataset, known as regression with estimated or imputed outcomes \citep{kennedy2023towards, foster2019orthogonal}. However, the estimation error of $\tau_a$ is closely tied to that of $\mu_a$. When $\hat{\mu}_a$ suffers from a slow convergence rate, $\hat{\tau}_a$ may inherit this slow rate as well. Faster convergence rates of estimating $\tau_a(\VV)$ can be achieved by adopting the stability framework in \cite{kennedy2023towards} or orthogonal statistical learning framework in \cite{foster2019orthogonal}. For example, instead of directly regressing $\hat{\mu}_a(\XX)$ on $\VV$, one can estimate ${\mu}_a$ and $\pi_a$, and construct a pseudo-outcome
\[
\hat{g}(\ZZ)=\frac{I(A=a)(Y-\hat{\mu}_a(\XX))}{\hat{\pi}_a(\XX)} + \hat{\mu}_a(\XX)
\]
for each data point in the source dataset. Similar to how doubly robust estimators reduce bias in estimating the ATE, regressing $\hat{g}(\ZZ)$ on $\VV$ yields an estimator of $\tau_a$ with an error rate depending on the product of the estimation rates for $\hat{\mu}_a$ and $\hat{\pi}_a$ (under appropriate conditions from \citet{kennedy2023towards}). This estimator offers a faster convergence rate than naive plug-in approaches and can achieve the oracle estimation rate for $\tau_a$ in a broader regime. For further details, we refer readers to \citet{kennedy2023towards}.

\section{Minimax Lower Bounds}\label{sec:Minimax}

In Section \ref{sec:Efficiency-theory} we established local asymptotic minimax optimality of doubly robust estimators (under certain conditions). In this section we will examine the minimax lower bounds from a global perspective, i.e., the minimax rate over suitable model classes, more generally when parametric $\sqrt{n}$ rates are not attainable. The minimax rate (in terms of mean-squared-error) of a statistical functional $\theta = \theta(\MP)$ over model class $\mathcal{P}$ is defined as
\[
\inf_{\hat{\theta}} \sup_{\MP \in \mathcal{P}} \left( \mathbb{E}_{\MP} (\hat{\theta} - \theta(\MP))^2 \right),
\]
where the infimum ranges over all possible estimators. It provides an important benchmark to compare against in constructing estimators. If one estimator has estimation error guarantees matching the minimax rate, then one may stop searching for estimators that can achieve a smaller estimation error and conclude the estimator is optimal in terms of worst-case rates. If the minimax rate is not achieved, one may study alternative estimators with smaller statistical error or study sharper bounds for the problem. In this section we derive the fundamental minimax lower bounds in estimating ATE in the target population under the special case $\VV=\XX$ (so the source population and target population share the same set of covariates). We introduce the ideas and techniques that are useful in deriving minimax rates in Appendix \ref{sec:minimax-general}. The minimax lower bounds for estimating the ATE in the target population, derived using these techniques, are presented in the remainder of this section.

\subsection{Minimax Lower Bounds in Generalization and Transportation}\label{sec:minimax-rates}
In the generalization and transportation setting, consider the target functionals in the case $\VV=\XX$. When the identification assumptions hold we can identify the effects as
\begin{equation}\label{eq:identification-v=x}
    \begin{aligned}
    \psi_a =&\, \mathbb{E}[Y^a] = \mathbb{E}\{\mathbb{E}[Y|\XX,S=1,A=a]\} \\
    \theta_a =&\, \mathbb{E}[Y^a|S=0] = \mathbb{E}\{\mathbb{E}[Y|\XX,S=1,A=a]|S=0\}.
    \end{aligned}
\end{equation}
We restrict the range of covariates $\XX$ as $[0,1]^d$ in this section. Consider the following model class for the generalization and transportation functionals, respectively:
\begin{equation*}
    \begin{aligned}
    \mathcal{P}_{\text{ge}}=\{(f, \rho, \pi_a, \mu_a): &\, \frac{1}{\rho \pi_a} \text{ is } \alpha\text{-smooth },\mu_a \text{ is } \beta\text{-smooth }, f\rho \pi_a=1 / 2, \\
    &\, \rho\pi_a \text{ and } \mu_a \text{ are bounded away from 0 and 1}\},
    \end{aligned}
\end{equation*}
\begin{equation*}
    \begin{aligned}
    \mathcal{P}_{\text{tr}}=\{(f, \rho, \pi_a, \mu_a): &\, \frac{1-\rho}{\rho \pi_a} \text{ is } \alpha\text{-smooth },\mu_a \text{ is } \beta\text{-smooth }, f\rho \pi_a=1 / 2, \\
    &\, \rho \pi_a \text{ and } \mu_a \text{ are bounded away from 0 and 1}\}.
    \end{aligned}
\end{equation*}
Here $f$ is the density of covariates $\XX$. We note that the selection/treatment probability $\rho \pi_a$ is parameterized together in the model class $\mathcal{P}_{\text{ge}}$. In other words, one can intuitively view $I(S=1, A=a)$ (i.e., both selected in the source population and receive treatment $a$) as a new treatment. The probability of getting treated under this ``compound" treatment is exactly $\rho \pi_a$. One minor difference between $\mathcal{P}_{\text{ge}}$ and $\mathcal{P}_{\text{tr}}$ is that in $\mathcal{P}_{\text{ge}}$ we impose smoothness conditions on $1/\rho \pi_a$ while in $\mathcal{P}_{\text{tr}}$ the smoothness condition is imposed on $(1-\rho)/\rho \pi_a$. As pointed out in the discussion on efficiency bounds in Section \ref{sec:EIF-EB} and quadratic estimation in Section \ref{sec:Quadratic}, the inverse odds ratio $(1-\rho)/\rho$ turns out to be a more fundamental quantity in transportation problems. To be consistent with these results, we also impose a smoothness condition on the inverse odds ratio in deriving the minimax rate.


The following theorem characterizes the minimax lower bounds in estimating generalization and transportation functionals. The strategy to construct two distribution classes $\MP_{\boldsymbol{\lambda}}$ and $\mathbb{Q}_{\boldsymbol{\lambda}}$ is similar to the techniques in proving the minimax rate for ATE, as in \cite{robins2009semiparametric}. 

\begin{theorem}\label{thm-minimax}
Let $s=(\alpha+\beta)/2$ denote the average smoothness of the nuisance functions. The minimax rate of generalization functional $\psi_a$ over $\mathcal{P}_{\text{ge}}$ is lower bounded by 
\[
\inf_{\hat{\psi}_a} \sup_{\MP \in \mathcal{P}_{\text{ge}}} \left( \mathbb{E}_{\MP} (\hat{\psi}_a - \psi_a)^2 \right)^{1/2} \gtrsim \begin{cases}n^{-1/(1+d/4s)} & \text { if } s<d/4 \\ n^{-1/2} & \text { otherwise. }\end{cases}
\]
The minimax rate of transportation functional $\theta_a$ over $\mathcal{P}_{\text{tr}}$ is lower bounded by 
\[
\inf_{\hat{\theta}_a} \sup_{\MP \in \mathcal{P}_{\text{tr}}} \left( \mathbb{E}_{\MP} (\hat{\theta}_a - \theta_a)^2 \right)^{1/2} \gtrsim \begin{cases}n^{-1/(1+d/4s)} & \text { if } s<d/4 \\ n^{-1/2} & \text { otherwise. }\end{cases}
\]
\end{theorem}

Although the minimax lower bounds above are the same as those for the ATE, we believe these results are not trivial, and presenting them helps us gain a more complete and comprehensive insight into generalization and transportation functionals. Following the efficiency theory for the ATE, the doubly robust estimator of $\psi_a$ and $\theta_a$ may achieve the minimax rate under some regimes. For instance, under the conditions in Theorem \ref{thm-dr-generalization} and assuming  $\|\hat{\pi}_a \hat{\rho} - \pi \rho\| \|\hat{\mu}_a - \mu_a\| = o_{\MP}(1/\sqrt{n})$, the doubly robust estimator is $\sqrt{n}-$consistent and attains the minimax rate (this corresponds to the smooth regime $s \geq d/4$). However, in the less smooth regime $s < d/4$, even if we can estimate $\rho \pi_a$ at a minimax rate (i.e. $\|\hat{\pi}_a \hat{\rho} - \pi \rho\| = O_{\MP}(n^{-\frac{\alpha}{2\alpha + d}})$ and $\|\hat{\mu}_a - \mu_a\| = O_{\MP}(n^{-\frac{\beta}{2\beta + d}})$), the necessary condition for the doubly robust estimator to achieve the rate $n^{-(2 \alpha+2 \beta) /(2 \alpha+2 \beta+d)}$ is
\[
\frac{\alpha}{2\alpha+d} + \frac{\beta}{2\beta+d} \geq \frac{2\alpha + 2\beta}{2\alpha + 2\beta +d},
\]
which can be a restrictive condition and motivates us to improve the doubly robust estimator. The main potential drawback of the doubly robust estimator is that we only correct for the first-order bias of the plug-in estimator. In the following section, we propose a higher-order/quadratic estimator based on second-order Von Mises expansion \citep{robins2009quadratic}, which also takes second-order bias into consideration and achieves the minimax lower bound in a broader regime.

\section{Higher-Order Estimation}\label{sec:Quadratic}

In this section we study a higher-order (quadratic) estimator of the functionals of interest. An introduction to quadratic Von Mises calculus is provided in Appendix \ref{sec:Quadratic-VonMises} . We apply the second-order Von Mises expansion to generalization and transportation functionals in the case $\VV=\XX$ and propose quadratic estimators of them in Section \ref{sec:Quadratic-estimation}. The error guarantees of our estimators are then established under mild assumptions.

\subsection{Quadratic Estimators in Generalization and Transportation} \label{sec:Quadratic-estimation}

We now present our higher-order estimator. 
We first define some objects that will be useful in presenting the theory of this section. Namely let:
\begin{equation*}
    \begin{aligned}
    \mathbf{\Omega} = &\, \int \bb(\xx)\bb(\xx)^{\top} \rho(\xx) \pi_a(\xx) dF(\xx), \\
    \widehat{\mathbf{\Omega}} = &\, \int \bb(\xx)\bb(\xx)^{\top} \hat{\rho}(\xx) \hat{\pi}_a(\xx) d\hat{F}(\xx), \\
    \Pi_{\bb} (g) (\xx) =&\,  \bb(\xx)^{\top} \mathbf{\Omega}^{-1}\int \bb(\uu)g(\uu)\rho(\uu) \pi_a(\uu) dF(\uu),\\
    \widehat{\Pi}_{\bb} (g) (\xx) =&\,  \bb(\xx)^{\top} \widehat{\mathbf{\Omega}}^{-1}\int \bb(\uu)g(\uu)\rho(\uu) \pi_a(\uu) dF(\uu).\\
    \end{aligned}
\end{equation*}
Here $\bb: \mathbb{R}^d \mapsto \mathbb{R}^k$ is a $k$-dimensional basis ($k$ corresponds to the dimension of the subspace $L$ in Section \ref{sec:Quadratic-VonMises}). $\mathbf{\Omega}$ is the Gram matrix of basis $\bb$ if we define the inner product between two functions $f,g$ as $\langle f,g \rangle = \int f(\xx)g(\xx)\rho(\xx)\pi_a(\xx) dF(\xx)$. $\widehat{\mathbf{\Omega}}$ is the estimated Gram matrix where all nuisance functions are replaced with their estimators. ${\Pi}_{\bb} (g)$ is the projection of $g$ onto the linear span of $\bb$ and $\widehat{\Pi}_{\bb} (g)$ is the estimated projection of $g$ where the gram matrix $\mathbf{\Omega}$ is replaced with its estimator $\widehat{\mathbf{\Omega}}$. With this notation, the non-centered first-order and (approximate) second-order influence functions of the generalization functional are 
\begin{equation*}
    \begin{aligned}
    \phi_{a,1}^{\text{ge}}(\ZZ) = &\, \frac{I(S=1,A=a)}{\rho(\XX)\pi_a(\XX)}(Y-\mu_a(\XX)) + \mu_a(\XX), \\
    \phi_{a,2}^{\text{ge}}(\ZZ_1, \ZZ_2) = &\, -I(S_1=1,A_1=a)(Y_1-\mu_a(\XX_1))\bb(\XX_1)^{\top}\mathbf{\Omega}^{-1}\bb(\XX_2) \frac{I(S_2=1,A_2=a) - \rho(\XX_2)\pi_a(\XX_2)}{\rho(\XX_2)\pi_a(\XX_2)}.
    \end{aligned}
\end{equation*}
Following the discussions in Section \ref{sec:Quadratic-VonMises}, a quadratic estimator is 
\[
\hat{\psi}_{a}^{qr} = \mathbb{P}_n[\hat{\phi}_{a,1}^{\text{ge}}(\ZZ)] + \mathbb{U}_n [\hat{\phi}_{a,2}^{\text{ge}}(\ZZ_1, \ZZ_2)],
\]
where all the nuisance functions and $\mathbf{\Omega}$ are replaced with their estimators. Note that the plug-in estimator $\psi(\hat{\MP})$ in the general quadratic estimator \eqref{eq:qr-estimator} cancels the centralization term in the centered first-order influence function $\phi_1(\ZZ,\hat{\MP})$. Hence in $\hat{\psi}_a^{qr}$ (where $\phi_{a,1}^{\text{ge}}(\ZZ)$ is the non-centered influence function) the plug-in term disappears.
In the following discussions, we first present a general theorem summarizing the conditional bias and variance of the quadratic estimator $\hat{\psi}_a^{qr}$, without assuming any conditions on the convergence rates of nuisance estimation. Then we examine the assumptions needed for the quadratic estimator $\hat{\psi}_a^{qr}$ to achieve the minimax optimal rate.

\begin{theorem}[Quadratic estimation of generalization functional]\label{thm-qr-generalization}

Assume all the nuisance functions $\hat{\mu}_a, \hat{\rho}, \hat{\pi}_a, \hat{F}$ are estimated from a separate training sample $D^n$. Further assume that $\rho, \pi, \hat{\rho}, \hat{\pi}$ are all bounded away from zero, the eigenvalues of $\mathbf{\Omega}, \widehat{\mathbf{\Omega}}$ are bounded away from zero and infinity. Then the conditional bias and variance of $\hat{\psi}_a^{qr}$ (giving the training data to estimate the nuisance functions) are bounded as
\begin{equation*}
    \begin{aligned}
    |\mathbb{E}[\hat{\psi}_a^{qr}|D^n] - \psi_a| \lesssim &\, \left\|(\II-\Pi_{\bb})   \left(\frac{1}{\hat{\rho}\hat{\pi}_a} - \frac{1}{\rho\pi_a} \right) \right\|_w  \left\|(\II-\Pi_{\bb}) \left( \hat{\mu}_a - \mu_a\right) \right\|_w \\
    + &\, \left\| \frac{1}{\hat{\rho}\hat{\pi}_a} - \frac{1}{\rho\pi_a} \right\|_w \left\|  \hat{\mu}_a - \mu_a \right\|_w \|\widehat{\mathbf{\Omega}}^{-1} - \mathbf{\Omega}^{-1}\|. \\
    \operatorname{Var}(\hat{\psi}_a^{qr}|D^n) \lesssim &\, \frac{1}{n}+ \frac{k}{n^2}
    \end{aligned}
\end{equation*}
where the weighted $L_2$ norm of a function $g$ is $\|g\|_w^2 = \int g^2(\xx)\pi_a(\xx)\rho(\xx) dF(\xx)$.
\end{theorem}

The boundedness assumption on the eigenvalues of $\mathbf{\Omega}$ can be implied by boundedness of $\rho, \pi_a$ and $\frac{d F}{d \nu}$ (the density of $F$ with respect to an underlying measure $\nu$) together with the assumption that $\int \bb(\xx) \bb(\xx)^{\top} d\nu(\xx)$ is positive definite. See Proposition 2.1 in \cite{belloni2015some} and Proposition 8 in \cite{kennedy2022minimax} for more detailed discussions. From Theorem \ref{thm-qr-generalization} we see the conditional bias is mainly composed of two parts. The first term 
\[
\left\|(\II-\Pi_{\bb})   \left(\frac{1}{\hat{\rho}\hat{\pi}_a} - \frac{1}{\rho\pi_a} \right) \right\|_w  \left\|(\II-\Pi_{\bb}) \left( \hat{\mu}_a - \mu_a\right) \right\|_w
\]
is the approximation error of the projection $\Pi_{\bb}$ and corresponds to the representational error discussed in Section \ref{sec:Quadratic-VonMises}. The second term 
\[
\left\| \frac{1}{\hat{\rho}\hat{\pi}_a} - \frac{1}{\rho\pi_a} \right\|_w \left\|  \hat{\mu}_a - \mu_a \right\|_w \|\widehat{\mathbf{\Omega}}^{-1} - \mathbf{\Omega}^{-1}\|
\]
is a third-order error term and corresponds to the remainder term $R_3(\hat{\MP}, \MP)$ discussed in Section \ref{sec:Quadratic-VonMises}. There is also an extra term $k/n^2$ in the conditional variance of the proposed higher-order estimator $\hat{\psi}_a^{qr}$. Hence we need to carefully select $k$ to balance the representational error, third-order error term, and the extra term in the variance. 

Note that we need some approximation guarantees of the basis $\bb$ to ensure that the representational error is small. Concretely, we impose the following uniform approximation assumption on $\bb$.

\begin{assumption}\label{assume-approximation}
For any $s>0$ and $g \in \mathcal{H}(s)$, the basis $\bb$ satisfies
\[
\left\|\left(\II-\Pi_{\bb}\right)g\right\|_w \lesssim k^{-s / d}.
\]
\end{assumption}
Assumption \ref{assume-approximation} holds for a wide class of basis functions. For instance, when the function class containing $g$ is supported on a convex and compact subset of $\mathbb{R}^d$, the approximation is valid even with the uniform norm $\|\cdot\|_{\infty}$ when the basis uses spline, CDV wavelet, or local polynomial partition series.  

Under the conditions in Theorem \ref{thm-qr-generalization}, if we assume $1/\rho\pi_a, 1/\hat{\rho}\hat{\pi}_a \in \mathcal{H}(\alpha) $ and $\hat{\mu}_a, \mu_a \in \mathcal{H}(\beta)$, by the approximation property of the basis $\bb$, we have
\[
\left\|(\II-\Pi_{\bb})   \left(\frac{1}{\hat{\rho}\hat{\pi}_a} - \frac{1}{\rho\pi_a} \right) \right\|_w \lesssim k^{-\alpha / d} \quad,  \left\|(\II-\Pi_{\bb}) \left( \hat{\mu}_a - \mu_a\right) \right\|_w \lesssim k^{-\beta / d}.
\]
In the less smooth regime $s = (\alpha + \beta)/2 < d/4$, set $k \sim n^{2 d /(d+2 \alpha+2 \beta)}$ to balance the representational error and variance. Further assume 
\[
\left\| \frac{1}{\hat{\rho}\hat{\pi}_a} - \frac{1}{\rho\pi_a} \right\|_w \left\|  \hat{\mu}_a - \mu_a \right\|_w \|\widehat{\mathbf{\Omega}}^{-1} - \mathbf{\Omega}^{-1}\| \lesssim n^{-(2 \alpha+2 \beta) /(2 \alpha+2 \beta+d)},
\]
Then we see that this estimator can achieve the minimax lower bound. In the smooth regime $s=(\alpha + \beta)/2 \geq d/4$, we can set $k$ similarly and assume 
\[
\left\| \frac{1}{\hat{\rho}\hat{\pi}_a} - \frac{1}{\rho\pi_a} \right\|_w \left\|  \hat{\mu}_a - \mu_a \right\|_w \|\widehat{\mathbf{\Omega}}^{-1} - \mathbf{\Omega}^{-1}\| \lesssim n^{-1/2},
\]
and again this estimator achieves the minimax lower bound. Compared with the condition required for the doubly robust estimator to achieve the minimax lower bound
\[
\|\hat{\rho}\hat{\pi}_a - \rho \pi_a\|\|\hat{\mu}_a - \mu_a\| = o_{\MP}(n^{-1/2}),
\]
since we have a third-order error term for the quadratic estimator, our higher-order/quadratic estimator achieves the minimax lower bound in a broader regime.


For the transportation functional $\theta_a$, we consider a related functional 
\[
\eta_a = \mathbb{E}[I(S=0)\mathbb{E}(Y|\XX,A=a,S=1)]=\mathbb{E}[I(S=0)\mu_a(\XX)] = \MP(S=0) \theta_a.
\]
Since $\MP(S=0)$ can be estimated at $\sqrt{n}$-rate, to estimate $\theta_a$ at a minimax optimal rate we only need to estimate $\eta_a$ at an optimal rate. The first-order and approximate second-order influence functions for $\eta_a$ are 
\begin{equation*}
    \begin{aligned}
    \phi_{a,1}^{\text{tr}}(\ZZ) = &\, \frac{I(S=1,A=a)(1-\rho(\XX))}{\rho(\XX)\pi_a(\XX)}(Y-\mu_a(\XX)) + I(S=0)\mu_a(\XX) \\
    \phi_{a,2}^{\text{tr}}(\ZZ_1, \ZZ_2) = &\, I(S_1=1,A_1=a)(Y_1-\mu_a(\XX_1))\bb(\XX_1)^{\top}\mathbf{\Omega}^{-1}\bb(\XX_2)\\
    &\, \times \frac{I(S_2=0)\rho(\XX_2)\pi_a(\XX_2) - (1-\rho(\XX_2))I(A_2=a, S_2=1)}{\rho(\XX_2)\pi_a(\XX_2)}
    \end{aligned}
\end{equation*}
The quadratic estimator is 
\[
\hat{\eta}_{a}^{qr} = \mathbb{P}_n[\hat{\phi}_{a,1}^{\text{tr}}(\ZZ)] + \mathbb{U}_n [\hat{\phi}_{a,2}^{\text{tr}}(\ZZ_1, \ZZ_2)].
\]
The following theorem, which is the analogous version of Theorem \ref{thm-qr-generalization} in the transportation setting, summarizes the estimation guarantee of the quadratic estimator $\hat{\eta}_{a}^{qr}$. 
\begin{theorem}[Quadratic estimation of transportation functional]\label{thm-qr-transportation}

Assume all the nuisance functions $\hat{\mu}_a, \hat{\rho}, \hat{\pi}_a, \hat{F}$ are estimated from a separate training sample $D^n$. Further assume that $\rho, \pi, \hat{\rho}, \hat{\pi}$ are all bounded away from zero, the eigenvalues of $\mathbf{\Omega}, \widehat{\mathbf{\Omega}}$ are bounded away from zero and infinity. The conditional bias and variance of $\hat{\eta}_{a}$ (giving the training data to estimate the nuisance functions) are bounded as
\begin{equation*}
    \begin{aligned}
    |\mathbb{E}[\hat{\eta}_{a}^{qr}|D^n] - \eta_a| \lesssim &\, \left\|(\II-\Pi_{\bb})   \left(\frac{1-\hat{\rho}}{\hat{\rho}\hat{\pi}_a} - \frac{1-\rho}{\rho\pi_a} \right) \right\|_w  \left\|(\II-\Pi_{\bb}) \left( \hat{\mu}_a - \mu_a\right) \right\|_w \\
    + &\, \left\| \frac{1-\hat{\rho}}{\hat{\rho}\hat{\pi}_a} - \frac{1-\rho}{\rho\pi_a} \right\|_w \left\|  \hat{\mu}_a - \mu_a \right\|_w \|\widehat{\mathbf{\Omega}}^{-1} - \mathbf{\Omega}^{-1}\|. \\
    \operatorname{Var}(\hat{\eta}_{a}^{qr}|D^n) \lesssim &\, \frac{1}{n}+ \frac{k}{n^2}
    \end{aligned}
\end{equation*}
where the weighted $L_2$ norm of a function $g$ is $\|g\|_w^2 = \int g^2(\xx)\pi_a(\xx)\rho(\xx) dF(\xx)$
\end{theorem}

Now the story is the same as the generalization case except we need to replace $1/\rho\pi_a$ with $(1-\rho)/\rho\pi_a$ (i.e. we impose smoothness assumption on $(1-\rho)/\rho\pi_a$ and $(1-\hat{\rho})/\hat{\rho}\hat{\pi}_a$). With the approximation property of basis $\bb$ (Assumption \ref{assume-approximation}) and further assume a bound on the convergence rate of the nuisance function estimators 
\[
\left\| \frac{1-\hat{\rho}}{\hat{\rho}\hat{\pi}_a} - \frac{1-\rho}{\rho\pi_a} \right\|_w \left\|  \hat{\mu}_a - \mu_a \right\|_w \|\widehat{\mathbf{\Omega}}^{-1} - \mathbf{\Omega}^{-1}\| ,
\]
we can prove that the quadratic estimator $\hat{\eta}_{a}^{qr}$ achieves the minimax lower bound in a broader regime than the doubly robust estimator.



\section{Simulation Study}\label{sec:Simulation}

In this section we examine the performance of doubly robust estimators empirically. Specifically, we examine the performance of the plug-in-style estimator and the doubly robust estimator under varying nuisance convergence rates in Section \ref{sec:simu-rate} and in the presence of potential model misspecification in Section \ref{sec:simu-misspecification}.

\subsection{Estimation Error and Nuisance Estimation Error}\label{sec:simu-rate}
Consider the following setting: $\XX = (X_1, X_2, X_3, X_4, X_5) \sim N(0, \II_5)$, $\VV = (X_1, X_2, X_3)$. Given $n$ samples, generate $S$ according to $\rho (\VV) = \MP(S=1 | \VV)  = 0.5$ (participation probability). In the source population, set $\pi_1(\xx) = \text{expit}(0.3x_1 - 0.3x_3)$ and simulate the treatment $A \sim \text{Bernoulli }(\pi_1(\XX))$. Consider the following linear potential outcome model
\[
\mu_1 (\xx) = 1.5x_1+x_4+1, \, \mu_0(\xx) = x_1,
\]
and $Y = A\mu_1(\XX) + (1-A)\mu_0(\XX) + N(0,1)$. So we have
\[
\tau_1(\vv) = \mathbb{E}[\mu_1(\XX)|\VV=\vv,S=1] = 1.5x_1 +1, \tau_0(\vv) = x_1.
\]
The transportation functional is $\theta_1 = \ME [\tau_1(\VV) \mid S=0] = 1$. The nuisance estimators are $\hat{\mu}_a(\xx) = \mu_a(\xx) + \epsilon_{1,n}, \hat{\tau}_a(\vv) = \tau_a(\vv) + \epsilon_{2,n}, \hat{\rho}(\vv) = \text{expit}(\text{logit}(\rho(\vv))+\epsilon_{3,n}), \hat{\pi}(\xx) = \text{expit}(\text{logit}(\pi_1(\xx))+\epsilon_{4,n})$, where $\epsilon_{i,n} \sim N(n^{-\alpha}, n^{-2\alpha})$. This construction guarantees that the root mean square errors (RMSE) of $\hat{\pi}_1, \hat{\rho}, \hat{\mu}_a, \hat{\tau}_a$ are of order $O(n^{-\alpha})$. Then we can use different values of $\alpha$ to evaluate the performance of the doubly robust estimator and plug-in estimator when the nuisance functions are estimated with different convergence rates $O(n^{-\alpha})$. Specifically, we set the possible values of $\alpha$ to be a sequence ranging from 0.1 to 0.5 by a step of 0.05. In each replication, we generate the data and use the doubly robust estimator and the plug-in estimator to estimate the functional $\theta_1 =1$. This process is replicated 1000 times for sample size $n=100, 1000, 5000$ and the RMSE is computed. We report the simulation results in Figure \ref{simu-results}.

\begin{figure}[H]
	\centering
	\subfigure[n=100]{
		\begin{minipage}[t]{0.3\linewidth}
			\centering
			\includegraphics[width=2in]{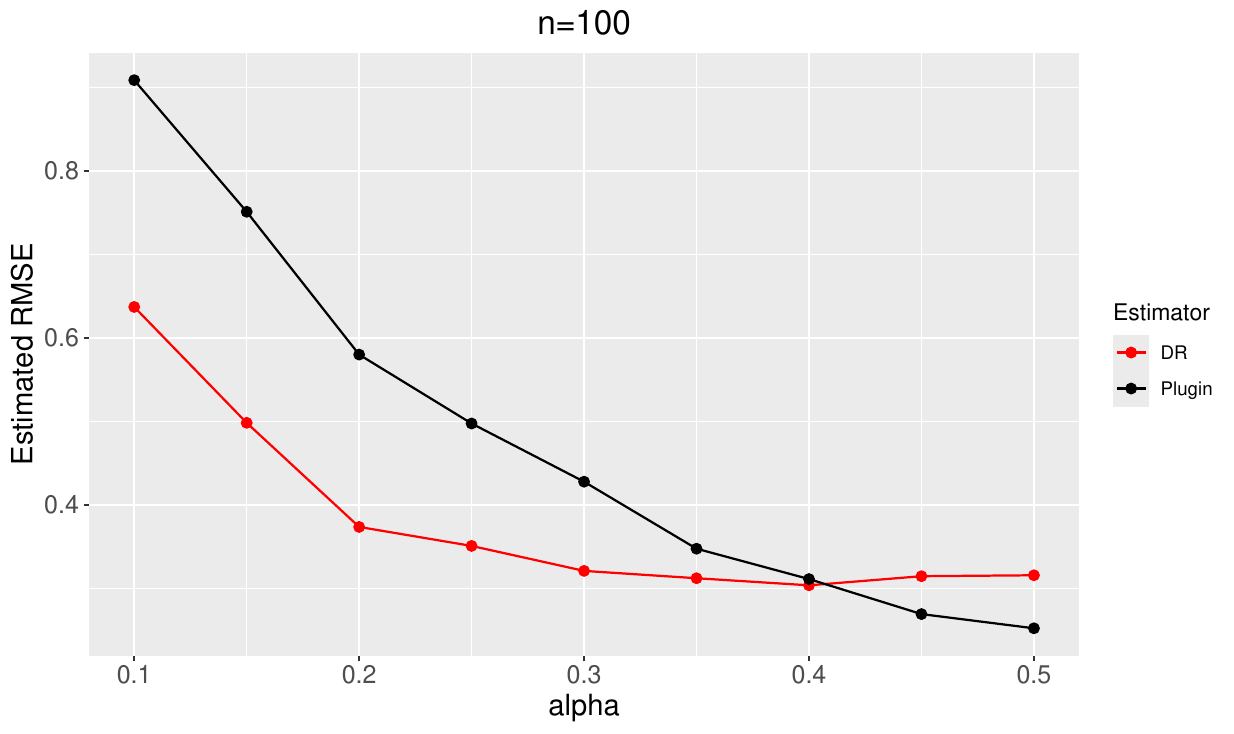}
	\end{minipage}}
	\subfigure[n=1000]{
		\begin{minipage}[t]{0.3\linewidth}
			\centering
			\includegraphics[width=2in]{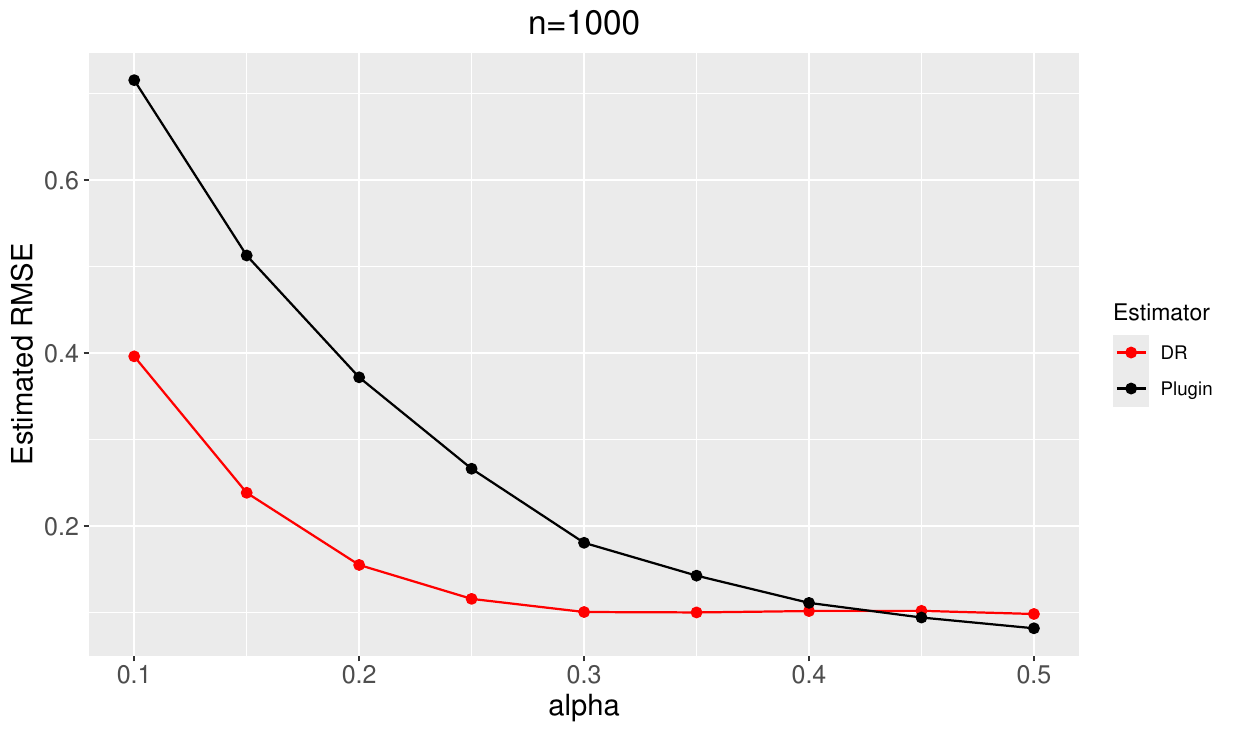}
	\end{minipage}}
	\subfigure[n=5000]{
		\begin{minipage}[t]{0.3\linewidth}
			\centering
			\includegraphics[width=2in]{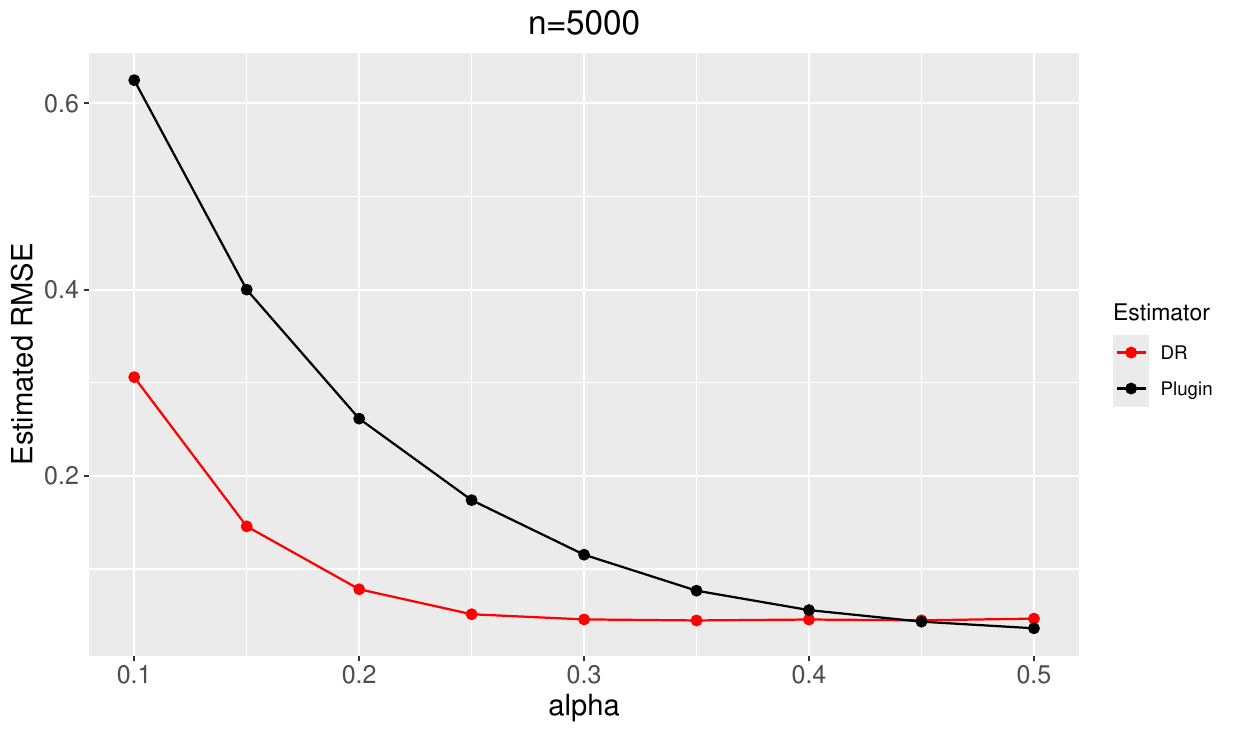}
	\end{minipage}}
	\centering
	\caption{RMSE V.S. $\alpha$ }
	\label{simu-results}
\end{figure}

From Figure \ref{simu-results} we see when the sample size is large and the convergence rate of the nuisance functions is slower than the parametric rate $O(n^{-1/2})$, the RMSE of doubly robust estimators can be much lower than the plug-in estimators. This coincides with our theoretical results. Note that although the performance of the plug-in estimator can be as good as the doubly robust estimator when the convergence rate is approximately the parametric rate, in practice we may be unable to correctly specify a parametric model for nuisance functions, and hence we are unlikely to achieve the parametric rate. So the doubly robust estimator is recommended in generalizing and transporting causal effects from the source population to the target population in real data analysis. In future work, we will compare against the performance of our higher-order estimator. 

\subsection{Estimation Error and Model Misspecification}\label{sec:simu-misspecification}

Consider the following setting: $\XX = (X_1, X_2, X_3, X_4) \sim N(0, \II_4)$, $\VV = (X_1, X_2, X_3)$. Given $n$ samples, generate $S$ according to $\rho (\vv) = \MP(S=1 | \vv)  = 0.2x_1 + 0.4x_2-0.5x_3^2$. In the source population, set $\pi_1(\xx) = \text{expit}(0.3x_1 - 0.2x_2 + 0.5x_3^2)$ and simulate the treatment $A \sim \text{Bernoulli }(\pi_1(\XX))$. Consider the following linear potential outcome model
\[
\mu_1 (\xx) = 1.5x_1+0.5x_2 + x_3^2+x_4,
\]
and $Y = A\mu_1(\XX) + N(0,1)$. So we have
\[
\tau_1(\vv) = \mathbb{E}[\mu_1(\XX)|\VV=\vv,S=1] = 1.5x_1+0.5x_2 + x_3^2.
\]
The generalization functional is $\psi_1= \ME [\tau_1(\VV) ] = 1$. We fit (generalized) linear regression models for the nuisance functions $\mu_1, \tau_1, \pi_1, \rho$. To evaluate estimator performance under model misspecification, we consider scenarios where the outcome models $\mu_1, \tau_1$ or the propensity score models $\pi_1, \rho$ are correctly specified or misspecified, with misspecification treating the quadratic term in $X_3$ as linear. We generate samples with size $n\in \{500, 1000, \dots, 5000\}$, apply two estimators under different model misspecifications, and compute the RMSE. The results are summarized in Figure \ref{fig:misspecification}.

\begin{figure}[H]
    \centering
    \includegraphics[width=3in]{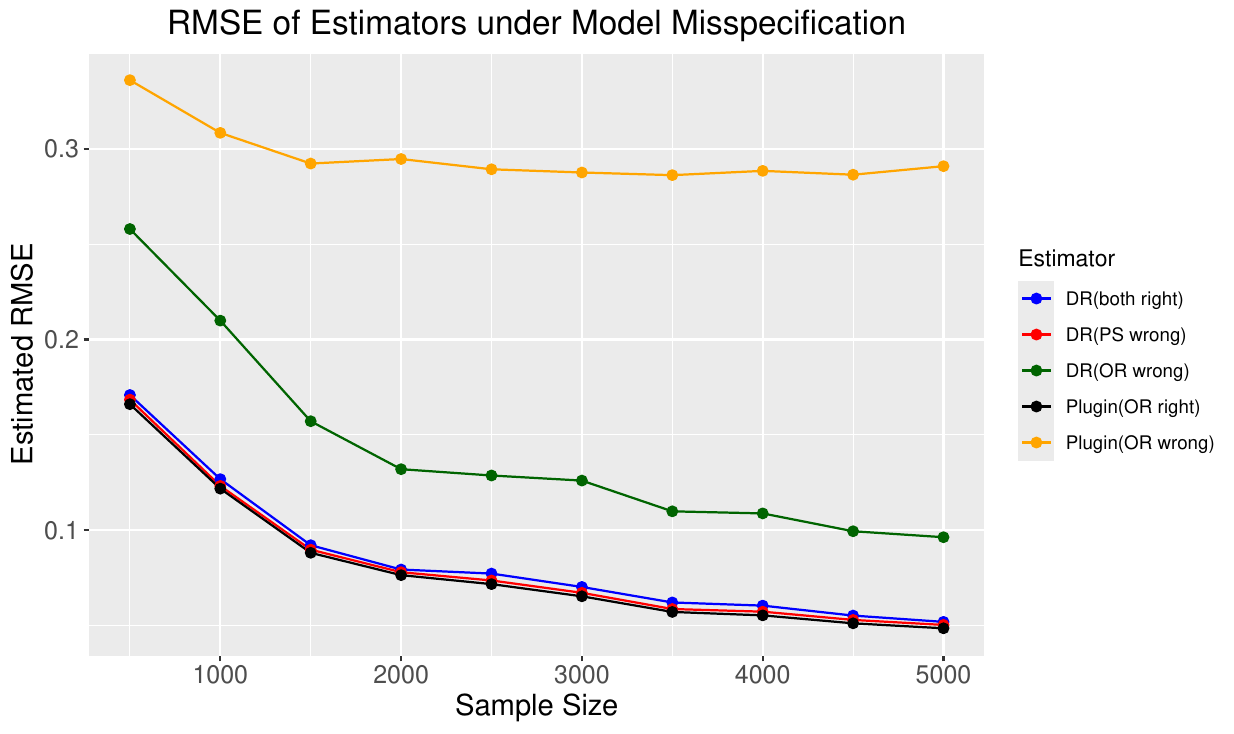}
    \caption{Performance of estimators under model misspecification}
    \label{fig:misspecification}
\end{figure}

When the outcome model is misspecified, the plug-in-style estimator, which relies solely on outcome modeling, becomes inconsistent, as illustrated in Figure \ref{fig:misspecification}. In contrast, the doubly robust estimator models both the outcome and treatment assignment, making it more robust to model misspecification. As long as either the outcome models $\tau_1, \mu_1$ or the propensity score models $\pi_1, \rho$ are correctly specified, the doubly robust estimator remains consistent, aligning with the theoretical results in Theorem \ref{thm-dr-generalization}. 

\section{Data Analysis}\label{sec:Real-data}

In this section we illustrate the proposed method with a real data example. Specifically, we aim to transport the causal effects of fruit and vegetable intake on adverse pregnancy outcomes from an observational study to all pregnancies in the U.S. We first introduce the background and motivation of the problem in Section \ref{sec:Real-data-background}. The necessity of transportation is assessed in Section \ref{sec:Real-data-difference}, i.e. we show the distributions of covariates in two populations are different and some covariates may modify effects. In Section \ref{sec:Real-data-transportation}, we apply the proposed method to two datasets and estimate the ATE of different dietary components in the target population (\textit{the whole U.S. pregnant female population}). Finally sensitivity analysis is performed when we are not confident in the exchangeability and transportability assumption in Section \ref{sec:Real-data-sensitivity}.

\subsection{Background and Motivation}\label{sec:Real-data-background}
Preterm birth, small-for-gestational-age birth, preeclampsia, and gestational diabetes are adverse pregnancy and birth outcomes that contribute to one-quarter of infant deaths in the U.S. and pose a tremendous economic and emotional burden for societies and families \citep{butler2007preterm, stevens2017short, dall2014economic}. Maternal nutrition is one of the few known modifiable risk factors for adverse pregnancy outcomes \citep{stephenson2018before}. Preventing poor outcomes by optimizing preconception dietary patterns is therefore a public health priority. We previously showed that diets with a high density of fruits and vegetables were associated with a reduced risk of poor pregnancy outcomes \citep{bodnar2020machine}. In the current analysis, we estimated the causal effects of fruit and vegetable intake on adverse pregnancy outcomes in \textit{the whole U.S. population of pregnant women}. Ideally, we would conduct randomized controlled trials on this population, or a random sample of it, and apply standard causal inference techniques to estimate the ATE. However, sampling and doing experiments on the U.S. pregnant population is not feasible. Therefore, we will transport the causal effects obtained from our prior work using data from the Nulliparous Pregnancy Outcomes
Study: monitoring mothers-to-be (nuMoM2b) to the U.S. pregnant population. 

As described in detail previously \citep{haas2015description}, nuMoM2b enrolled 10,083 people in 8 US medical centers from 2010 to 2013. Eligibility criteria included a viable singleton pregnancy, 6-13 completed weeks of gestation at enrollment, and no previous pregnancy that lasted $\geq 20$ weeks of gestation. All participants provided informed, written consent and the study was approved by each site’s institutional review board. At enrollment (6-13 completed weeks of gestation), women completed a semi-quantitative food frequency questionnaire querying usual periconceptional dietary intake. The full dietary assessment approach has been published \citep{bodnar2020machine}. The treatments of interest are servings of fruits and vegetables per day. Other food groups served as important confounders: whole grains, dairy products, total protein foods, seafood and plant proteins, fatty acids, refined grains, sodium and ``empty" calories. At least 30 days after delivery, a trained certified chart abstractor recorded final birth outcomes, medical history, and delivery diagnoses and complications. This information provides us with data on responses: preterm birth, SGA birth, gestational diabetes and preeclampsia. Data were also ascertained on maternal age, race, education, prepregnancy body mass index, smoking status, marital status, insurance status and working status. In our analysis, let the threshold be the 80\% quantile of total fruit or vegetable intake. Total fruit/vegetable intake higher than the 80\% quantile is considered treated $(A=1)$, otherwise not treated $(A=0)$. For the outcome, when the adverse pregnancy outcome occurs, we let $Y=1$. Otherwise $Y=0$. We will focus on the causal effects of fruit and vegetable intake on the aforementioned adverse pregnancy outcomes.

However, participants in the nuMoM2b cohort may not be representative of the U.S. pregnant population. For instance, 23\% of participants received an education beyond college compared with 10\% of the U.S. pregnant population. Hence the distributions of covariates in study participants and target population are quite different. The nuMoM2b cohort may not be representative of the target population. It is possible that some covariates, such as age and education level, will modify the effects of
dietary components. Then the estimates in \cite{bodnar2020machine} based on Numom study will not immediately generalize to the
U.S. population of women. To estimate the ATE in \textit{the whole U.S. pregnant female population}, we find a U.S. representative sample from the National Survey of Family Growth (NSFG), which contains information on 9553 women in the U.S. The documentation of the data is available at \verb|https://www.cdc.gov/nchs/nsfg/nsfg_2015_2017_puf.htm|.

Before applying the proposed methods, we formally assess the necessity of transportation in Section \ref{sec:Real-data-difference} in the Appendix \citep{zeng2024supple}.

\subsection{Transportation}\label{sec:Real-data-transportation}

Having observed that transportation methods are necessary based on results in Section \ref{sec:Real-data-difference}, here we assume the five identification assumptions in Section \ref{sec:Identification} and use our doubly robust estimator to estimate the ATE in the target population. The covariate sets that we use are $\XX=$\{education, age, race, marital status, insurance, work,
smoking, number of cigarettes, BMI, other dietary
components(e.g. protein)\} and $ \VV = $ \{education, age, race, marital status, insurance, work,
smoking, number of cigarettes\}. In particular, we adapt our estimators to these two datasets in two aspects. First, there exist some units with extremely small propensity score $\hat{\pi}_a(\XX)$ in the source dataset. Since we need to reweight each sample in the nuMoM2b cohort by the inverse of the probability of getting treated, our doubly robust estimator may suffer from high instability if we directly use these small propensity scores. So we enforce all the propensity scores $\pi_a(\XX_i)$'s to be in the range [0.01, 0.99] (i.e. the positivity constant $\epsilon$ is set to be 0.01 in our analysis). Similarly, we also enforce all the participation probabilities $\rho(\VV_i)$'s to be greater than 0.01. 

The other point is that the sampling process of NSFG dataset undergoes multiple stages and is more complicated than simple random sampling (SRS). According to examples given by CDC, we can approximate that sampling process with a stratified cluster sampling procedure and estimate the mean of a variable with appropriate weights. The variance of the estimator can also be obtained from the theory of stratified cluster sampling. The details on adjusting our estimator based on stratified cluster sampling can be found in the supplementary materials.

Our methods with above adjustments are applied to each combination of treatment $A \in \{\text{fruit intake}, \text{vegetable intake}\}$ and outcome $Y \in \{\text{preterm Birth, SGA birth, gestational diabetes,}$  preeclampsia$\}$. All the nuisance functions are fitted with ``SL.ranger" (random forests) and ``SL.glmnet" (penalized GLM) and ``SL.mean" in the R package ``SuperLearner". Five-fold cross fitting is used to guarantee the sample splitting condition in Theorem \ref{thm-dr-transportation}. The results are summarized in Figure \ref{fig:effects}.

\begin{figure}[H]
	\centering
	\subfigure[Fruit]{
		\begin{minipage}[t]{0.45\linewidth}
			\centering
			\includegraphics[width=3in]{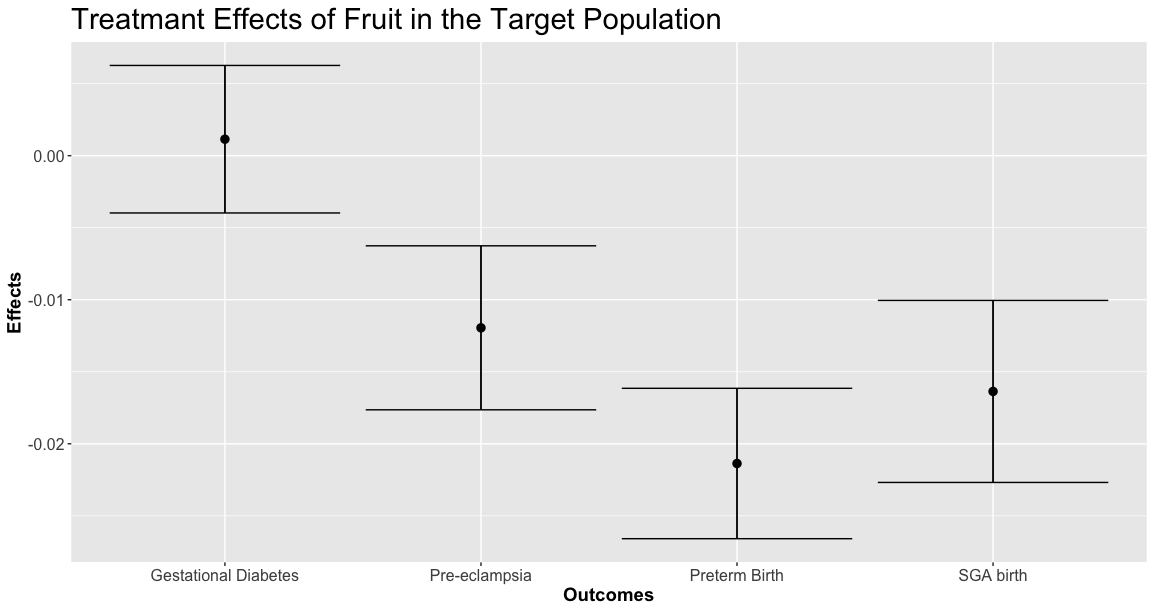}
	\end{minipage}}
	\subfigure[Vegetables]{
		\begin{minipage}[t]{0.45\linewidth}
			\centering
			\includegraphics[width=3in]{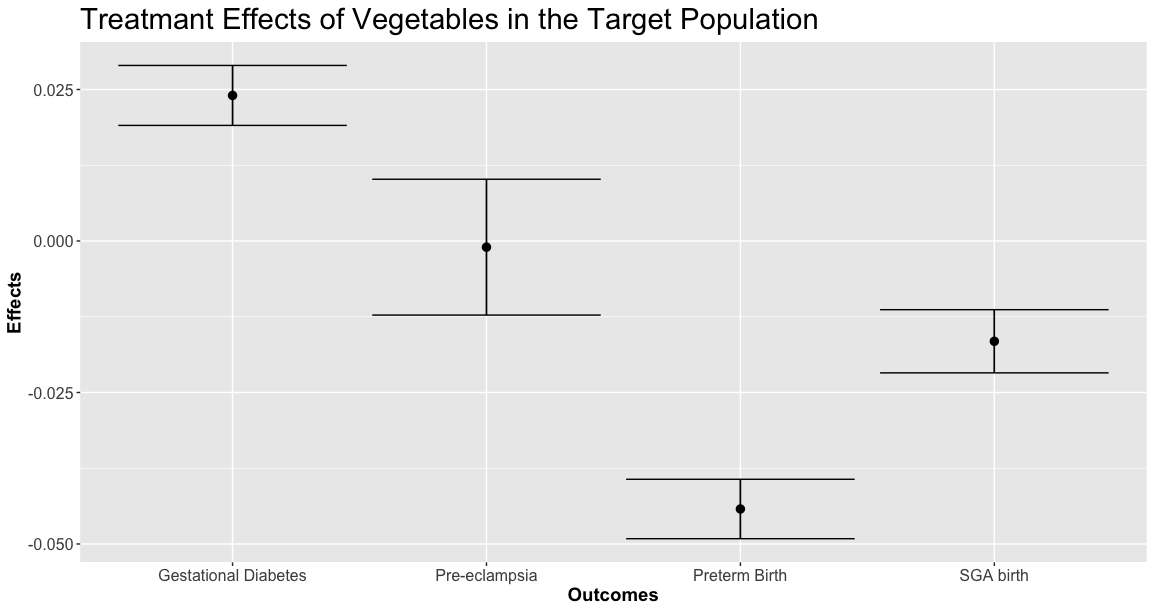}
	\end{minipage}}
	\centering
	\caption{Effects of fruit and vegetable intake in \textit{the whole U.S. pregnant female population}}
	\label{fig:effects}
\end{figure}

Detailed analysis on model selection (i.e. choice of models in SuperLearner), choice of positivity constant $\epsilon$ and potential outliers is presented in the supplementary materials, which justifies the choice in our analysis. The effects of high fruit intake on preterm birth, preeclampsia, gestational diabetes and SGA birth are -0.0214 (95\% CI [-0.0266,  -0.0162]), -0.012 (95\% CI [-0.0176, -0.00626]), 0.00114 (95\% CI [-0.00398, 0.00626]), -0.0164 (95\% CI [-0.0227, -0.0101]), respectively. The effects of high vegetable intake on preterm birth, preeclampsia, gestational diabetes and SGA birth are -0.0442 (95\% CI [-0.0491,  -0.0393]), -0.00102 (95\% CI [-0.0122,  0.0102]), 0.024 (95\% CI [0.0191, 0.029]), -0.0166 (95\% CI [-0.0218, -0.0113]).

From the results above, we see the effects of high fruit intake on preterm birth, preeclampsia and SGA birth are significantly negative at level 0.05 in the target population, which implies eating more fruit potentially causes a lower risk of these adverse pregnancy outcomes. For the results on vegetables, the effect of high vegetable intake on preterm birth and SGA birth is significantly negative. The strict interpretation for the effects of high vegetable intake on preterm birth is: Compared with participants whose vegetable intake is
less than 80\% quantile of vegetable intake in the nuMoM2b cohort, women with
higher vegetable intake have 4 fewer preterm births for every 100
women in \textit{the whole U.S. pregnant female population}. Other combinations of treatments and outcomes can be similarly interpreted. We also see (significant) positive effects of high vegetable intake on gestational diabetes, which shows eating more vegetables may potentially increase the risk of getting gestational diabetes. This result is counter-intuitive. One potential problem is the identification assumptions may not hold, so we cannot interpret our results from a causal perspective. As a result, we will perform sensitivity analysis in the following section, to deal with possible violations of the identification assumptions.

\subsection{Sensitivity Analysis and Discussions}\label{sec:Real-data-sensitivity}

Here we focus on the effects of high vegetable intake on gestational diabetes as an example. According to the discussions in Section \ref{sec:Sensitivity-analysis}, under Assumption \ref{relax-exchangeability} and Assumption \ref{relax-transportability} the bound for the ATE in the target population is given by
\[
[ 0.024 - \delta_1 - 2\delta_2,  0.024 + \delta_1 + 2\delta_2].
\]
We visualize the range of the bounds when only exchangeability is violated ($\delta_2=0$) or only transportability is violated ($\delta_1=0$) in Figure \ref{fig:sensitivity}, respectively. 

\begin{figure}[H]
	\centering
	\subfigure[Only exchangeability is violated ($\delta_2=0$)]{
		\begin{minipage}[t]{0.45\linewidth}
			\centering
			\includegraphics[width=3in]{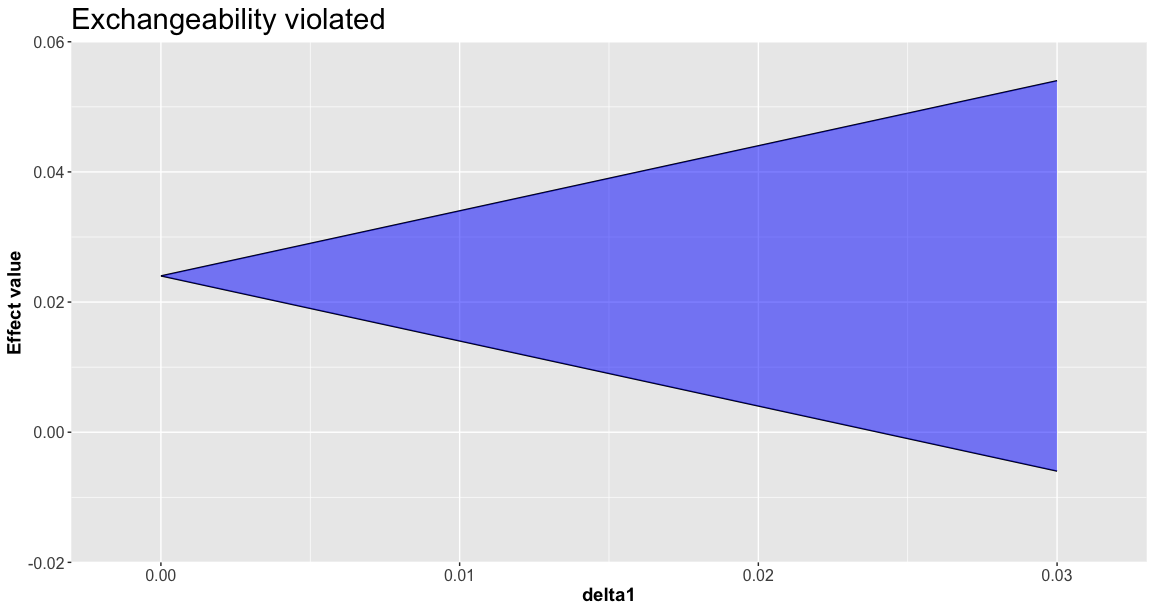}
	\end{minipage}}
	\subfigure[Only transportability is violated($\delta_1=0$)]{
		\begin{minipage}[t]{0.45\linewidth}
			\centering
			\includegraphics[width=3in]{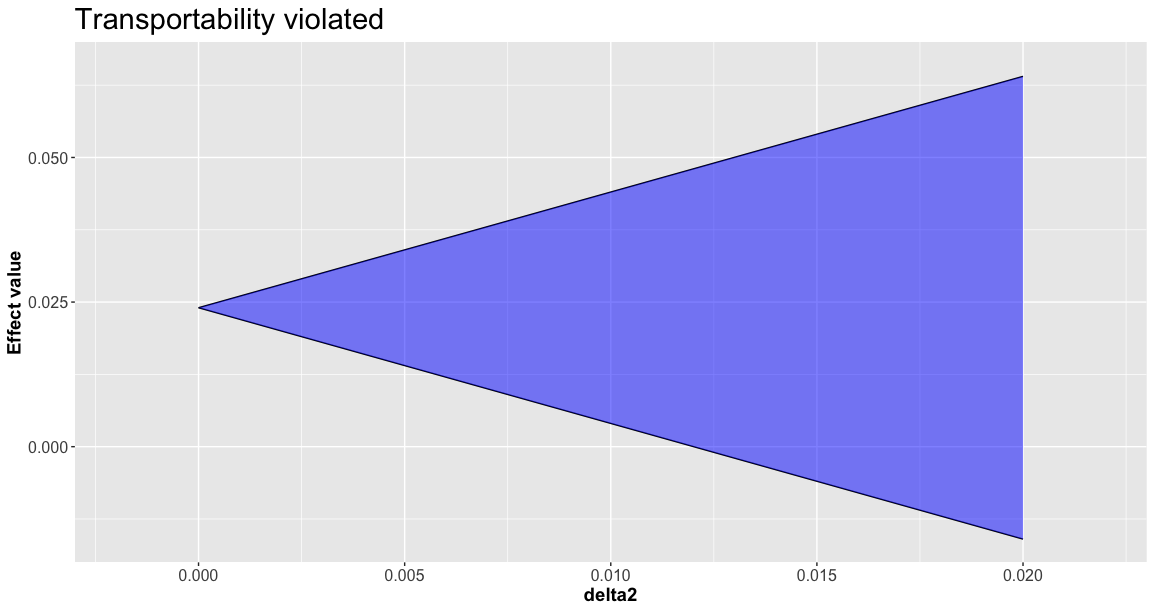}
	\end{minipage}}
	\centering
	\caption{Sensitivity analysis when treatment is high vegetable intake and outcome is gestational diabetes}
	\label{fig:sensitivity}
\end{figure}

The plots should be understood as follows: (a) For a specific value $\delta_1$, the interval that covers the ATE is given by the intersection of the blue region and the line $x=\delta_1$. (b) can be interpreted similarly. From Figure \ref{fig:sensitivity} we see if only exchangeability is violated (i.e. $\delta_2=0$), then the critical value for $\delta_1$ to turn around the result is $0.024$. If only transportability is violated (i.e. $\delta_1=0$), then the critical value for $\delta_2$ to turn around the result is $0.012$.

Considering the covariate sets $\XX=$\{education, age, race, marital status, insurance, work,
smoking, number of cigarettes, BMI, other dietary
components(e.g. protein)\} and $ \VV = $ \{education, age, race, marital status, insurance, work,
smoking, number of cigarettes\}, we find most of the covariates are categorical. The information contained in categorical variables is usually less than that in continuous ones. Furthermore, there may be some confounders or effect modifiers that are not measured in the dataset. For instance, body mass index (BMI) may be an effect modifier but we do not have information on it in the NSFG dataset and hence cannot take it into account. Therefore the categorical variables we include in the analysis may not provide sufficient information for the identification assumption \ref{exchangeability} and \ref{transportability} to hold. Hence sensitivity analysis is quite necessary in our problem.

From the perspective of the high vegetable intake treatment itself, it may be that the increased risk is due to the majority of vegetables being fried potatoes, which are less healthful than other vegetables.

\section{Discussions}\label{sec:Discussion}

In this paper we extend the generalization and transportation techniques in \citep{dahabreh2019generalizing, dahabreh2020extending} to the case where different sets of covariates can be used in the source population and the target population. To be concrete, we summarize the sufficient assumptions to identify the ATE in the target population. We also provide methods to perform sensitivity analysis when the identification assumptions fail. The first-order influence functions and efficiency bounds are derived for the target statistical functionals when they are identified. We further propose a doubly robust estimator based on the first-order influence functions and establish its asymptotic normality under proper conditions. Simulation studies show the advantage of the doubly robust estimator over the plug-in estimator. We also study the minimax lower bounds and higher-order/quadratic estimation based on second-order Von Mises expansion in the case where the source population and the target population share the same set of covariates (i.e. $\VV=\XX$). Although we rely on similar techniques used in the ATE case \citep{robins2009quadratic, robins2009semiparametric}, these results are non-trivial and important in understanding the properties of generalization and transportation functional. Finally we illustrate the proposed methods with an interesting example, where we transport the causal effects of fruit and vegetable intake on adverse pregnancy outcomes from an observational study to \textit{the whole U.S. pregnant female population}. 

In this paper we only consider the minimax rate and quadratic estimation in the special case $\VV=\XX$. In the general case $\XX=(\WW,\VV)$, the target functionals in \eqref{eq:identify} involve a triple integral and can be viewed as ``cubic functionals" \citep{tchetgen2008minimax, mukherjee2015lepski}. Such functionals are not well understood in the literature and it is more challenging to construct adversarial settings to establish the minimax lower bounds or derive higher-order influence functions to correct for the second-order bias. One special case is to assume the covariates $\VV$ are discrete and we can write
\[
\psi_a = \sum_{\vv}p(\vv) \ME[\ME(Y|\XX,A=a,S=1)|\VV=\vv,S=1]. 
\]
Now we only need to estimate $\ME[\ME(Y|\XX,A=a,S=1)|\VV=\vv,S=1] = \ME[\ME(Y|\WW,A=a,\VV=\vv,S=1)|\VV=\vv,S=1]$. Note that this functional is equivalent to estimating ATE among individuals with $\VV=\vv, S=1$ using covariate set $\WW$ and hence the standard ATE theory applies. We leave the general case as future work. Moreover, while \citet{bonvini2022fast} and \citet{zeng2024causal} have demonstrated that higher-order estimators perform well in settings with continuous treatments or discrete covariates, the implementation and selection of tuning parameters remain open problems for future research. 
 
Another relevant literature focuses on combining multiple data sources to estimate treatment effects \citep{yang2020combining, han2021federated, han2023multiply}, where some sources may have unmeasured confounders or may not represent the target population. Extending our methods to this setting, where generalization and transportation are considered with multiple datasets from the source population, is a potential avenue for future research. 
Additionally, other interesting questions in generalization and transportation include whether it is possible to generalize or transport time-varying treatment effects, where the treatment or exposure changes over time. In the classical ATE framework, several approaches exist for sensitivity analysis, where different sensitivity models are used to address potential violations of identification assumptions \citep{vanderweele2017sensitivity, zhang2019semiparametric, bonvini2022sensitivity, yadlowsky2022bounds}. Formalizing similar sensitivity models for use in generalization and transportation settings remains an open question for future investigation. A further practical challenge is how to appropriately select confounders $\XX$ and effect modifiers $\VV$. \citet{rudolph2024improving} demonstrated that researchers can focus on variables that are both effect modifiers and differentially distributed across groups in $\VV$. However, identifying this set of variables is typically not straightforward in applications. Therefore, it is crucial to develop principled methods for selecting variables to adjust for in the generalization and transportation of causal effects.

\bibliographystyle{apalike}
\bibliography{refer}

\newpage

\appendix

\section*{Appendix}

\section{Background}\label{sec:appendix-background}

 \subsection{General Strategies to Prove Minimax Lower Bounds} \label{sec:minimax-general}

We will use the following lemma (adapted from Le Cam's convex hull method), which summarizes the main ingredients in deriving minimax lower bound for a statistical functional, to derive the minimax lower bounds.

\begin{lemma}\label{lemma-minimax}
Let $\MP_{\boldsymbol{\lambda}}$ and $\mathbb{Q}_{\boldsymbol{\lambda}}$ denote distributions in $\mathcal{P}$ indexed by a vector $\boldsymbol{\lambda} = (\lambda_1, \dots, \lambda_k)$ with corresponding n-dimensional product measure $\MP_{\boldsymbol{\lambda}}^n$ and $\mathbb{Q}_{\boldsymbol{\lambda}}^n$. Let $\pi$ be a probability distribution of $\boldsymbol{\lambda}$. The functional of interest is $\psi: \mathcal{P} \mapsto \mathbb{R}$. If
\[
H^2\left(\int \MP_{\boldsymbol{\lambda}}^n d \pi(\boldsymbol{\lambda}), \int \mathbb{Q}_{\boldsymbol{\lambda}}^n d \pi(\boldsymbol{\lambda})\right) \leq \alpha<2
\]
and 
\[
\left|\psi\left(\MP_{\boldsymbol{\lambda}}\right)-\psi\left(\mathbb{Q}_{\boldsymbol{\lambda}}\right)\right| \geq s>0
\]
for all $\boldsymbol{\lambda}$, then 
\[
\inf _{\widehat{\psi}} \sup _{\MP \in \mathcal{P}} \mathbb{E}_{\MP}\{|\widehat{\psi}-\psi(\MP)|\} \geq \frac{s}{2}\left(\frac{1-\sqrt{\alpha(1-\alpha / 4)}}{2}\right).
\]
\end{lemma}

As pointed out in Lemma \ref{lemma-minimax}, the general method to derive minimax lower bound is to construct two distributions that are statistically indistinguishable (i.e. the distance between two distributions is small) but the target parameters of two distributions are sufficiently separated. In nonparametric regression problems, to arrive at the minimax rate it suffices to apply Le Cam's two-point method and construct a pair of distributions \citep{Tsybakov2009}. In nonlinear functional estimation, typically at least one distribution needs to be a mixture to obtain a tight bound \citep{robins2009semiparametric}. Intuitively, by considering mixture distributions, the signal from the distribution is further neutralized, which yields distributions with smaller statistical distance while the target parameter is still separated and thus gives tighter lower bounds. New techniques for bounding the distance between mixture distributions have been developed in the literature. For instance, \cite{birge1995estimation} considered the case where only one distribution is a mixture and \cite{robins2009semiparametric} further extended the bound to the case where both distributions are mixtures. In our analysis, we will use the following bound on Hellinger Distance from \cite{robins2009semiparametric}.

\begin{lemma}[\cite{robins2009semiparametric}]\label{bound-Hellinger}
Let $\MP_{\boldsymbol{\lambda}}$ and $\mathbb{Q}_{\boldsymbol{\lambda}}$ denote distributions in $\mathcal{P}$ indexed by a vector $\boldsymbol{\lambda} = (\lambda_1, \dots, \lambda_k)$ with corresponding n-dimensional product measure $\MP_{\boldsymbol{\lambda}}^n$ and $\mathbb{Q}_{\boldsymbol{\lambda}}^n$ and let $\mathcal{Z}=\cup_{j=1}^k \mathcal{Z}_j$ be a partition of the sample space satisfying
\begin{enumerate}
    \item $\MP_{\boldsymbol{\lambda}}\left(\mathcal{Z}_j\right)=\mathbb{Q}_{\boldsymbol{\lambda}}\left(\mathcal{Z}_j\right)=p_j$ for all $\boldsymbol{\lambda}$;
    \item the conditional distributions $I_{\mathcal{Z}_j} d \MP_{\boldsymbol{\lambda}} / p_j $ and $ I_{\mathcal{Z}_j} d \mathbb{Q}_{\boldsymbol{\lambda}} / p_j$ only depend on $\lambda_j$ (do not depend on $\lambda_k$ for $k \neq j$).
\end{enumerate}
Let $p_{\boldsymbol{\lambda}}$ and $q_{\boldsymbol{\lambda}}$ denote the density of $\MP_{\boldsymbol{\lambda}}$ and $\mathbb{Q}_{\boldsymbol{\lambda}}$, respectively. Given a product distribution $\pi$ on $\boldsymbol{\lambda}$, define $p=\int p_{\boldsymbol{\lambda}} d \pi(\boldsymbol{\lambda})$, $q=\int q_{\boldsymbol{\lambda}} d \pi(\boldsymbol{\lambda})$ and 
\begin{equation*}
    \begin{aligned}
\delta_1 &=\max _j \sup _{\boldsymbol{\lambda}} \int_{\mathcal{Z}_j} \frac{\left(p_{\boldsymbol{\lambda}}-p\right)^2}{p_{\boldsymbol{\lambda}} p_j} d \nu \\
\delta_2 &=\max _j \sup _{\boldsymbol{\lambda}} \int_{\mathcal{Z}_j} \frac{\left(q_{\boldsymbol{\lambda}}-p_{\boldsymbol{\lambda}}\right)^2}{p_{\boldsymbol{\lambda}} p_j} d \nu \\
\delta_3 &=\max _j \sup _{\boldsymbol{\lambda}} \int_{\mathcal{Z}_j} \frac{(q-p)^2}{p_{\boldsymbol{\lambda}} p_j} d \nu
\end{aligned}
\end{equation*}
for a dominating measure $\nu$. If $n p_j(1 \vee \delta_1 \vee \delta_2) \leq b$ for all $j$ and $b^{-1} \leq p_{\boldsymbol{\lambda}} \leq b$ for a positive constant $b$, then there exists a constant $C$ that only depends on $b$ such that 
\[
H^2\left(\int \MP_{\boldsymbol{\lambda}}^n d \pi(\boldsymbol{\lambda}), \int \mathbb{Q}_{\boldsymbol{\lambda}}^n d \pi(\boldsymbol{\lambda})\right) \leq C n^2\left(\max _j p_j\right)\left(\delta_1 \delta_2+\delta_2^2\right)+C n \delta_3.
\]
\end{lemma}
The idea of deriving minimax lower bound now becomes clear. We first carefully construct two classes of distributions indexed by $\boldsymbol{\lambda}$ such that the target parameters of two distribution classes are separated. Then we can apply Lemma \ref{bound-Hellinger} to bound the Hellinger distance between the mixture distributions of two distribution classes (with the same mixture distribution on $\boldsymbol{\lambda}$). Finally we apply Lemma \ref{lemma-minimax} to obtain the minimax lower bound. In Section \ref{sec:minimax-rates} we apply this idea to the generalization and transportation setting.

\subsection{Quadratic Von Mises Calculus}\label{sec:Quadratic-VonMises}

As discussed in Section \ref{sec:EIF-EB}, the first-order influence function is the derivative in the first-order Von Mises expansion of the statistical functional $\psi=\psi(\MP)$. Mathematically, we have
\[
\psi(\hat{\mathbb{P}})-\psi(\mathbb{P})=-\int \phi_1(\zz,\hat{\mathbb{P}}) d \mathbb{P}(\zz)+R_2(\hat{\mathbb{P}}, \mathbb{P}),
\]
where $\phi_1(\zz,\MP)$ is the influence function of $\psi(\MP)$, $\psi(\hat{\mathbb{P}})$ is the plug-in estimator and $R_2(\hat{\mathbb{P}}, \mathbb{P})$ is the second-order reminder. This suggests that the plug-in estimator has first-order bias $-\int \phi_1(\zz,\hat{\mathbb{P}}) d \mathbb{P}(\zz)$. Equivalently, we can write
\[
\psi(\mathbb{P})=\psi(\hat{\mathbb{P}})+\int \phi_1(\zz,\hat{\mathbb{P}}) d \mathbb{P}(\zz) +R_2(\hat{\mathbb{P}}, \mathbb{P}),
\]
which motivates us to correct for the first-order bias and arrive at the doubly robust estimator
\[
\hat{\psi}^{dr} = \psi(\hat{\mathbb{P}})+\mathbb{P}_n\{\phi_1(\ZZ,\hat{\mathbb{P}})\}.
\]
Under further empirical process assumptions or sample splitting assumptions, we can show the dominating term in the conditional bias of $\hat{\psi}^{dr}$ is $R_2(\hat{\mathbb{P}}, \mathbb{P})$, which is usually a second-order error term (e.g. $R_2(\hat{\mathbb{P}}, \mathbb{P})$ is of order $\|\hat{\mu}_a - \mu_a\|\|\hat{\pi}_a -\pi_a\|+ \|\hat{\rho} -\rho\| \|\hat{\tau}_a -\tau_a\|$ in generalization and transportation functionals). At this point, one may consider further expressing $\psi(\MP)$ as a second-order Von Mises expansion and hopefully the remainder error term can be of ``third-order small". This motivates the study of second-order 
influence functions. Mathematically, we can write
\begin{equation}\label{eq:2nd-expansion}
\begin{aligned}
    \psi(\mathbb{P})-\psi(\hat{\mathbb{P}})=& \,\int \phi_1(\zz,\hat{\mathbb{P}}) d( \mathbb{P} -\hat{\MP})(\zz)+\frac{1}{2} \int \phi_2(\zz_1,\zz_2, \hat{\MP})\, d(\MP -\hat{\MP}) \times (\MP -\hat{\MP}) (\zz_1, \zz_2) +R_3(\hat{\mathbb{P}}, \mathbb{P})\\
    =&\, \int \phi_1(\zz,\hat{\mathbb{P}}) d \mathbb{P}(\zz)+\frac{1}{2} \int \phi_2(\zz_1,\zz_2, \hat{\MP})\, d(\MP  \times \MP ) (\zz_1, \zz_2) +R_3(\hat{\mathbb{P}}, \mathbb{P}),
\end{aligned}
\end{equation}
where $\phi_1(\zz, \MP)$ is the first-order influence function and $\phi_2(\zz_1,\zz_2,\MP)$ is the second-order influence function. Here we assume $\int \phi_1(\zz, \MP)d \MP (\zz) = 0$ and $\int \phi_2 (\zz_1,\zz_2,\MP) d \MP(\zz_i) = 0$ for $i=1, 2$. A strict mathematical definition of second-order influence function requires the notion of second-order score and we refer to \cite{robins2009quadratic} for more detailed discussions. When the second-order influence function exists, it can be derived by differentiating the first-order influence function (i.e. deriving the influence function of the first-order influence function), as summarized in the following Lemma.

\begin{lemma}
\cite{robins2009quadratic}
Suppose $\{\MP_t, t\geq 0\}$ is a sufficiently smooth submodel and $\phi_1 : \mathcal{Z} \mapsto \mathbb{R}, \phi_2: \mathcal{Z}\times \mathcal{Z} \mapsto \mathbb{R}$ are measurable functions that satisfy 
\begin{equation*}
\begin{aligned}
    \frac{\partial \psi(\MP_t)}{\partial t} =&\,  \int \phi_1(\zz,\MP_t) \frac{\partial p_t(\zz)  }{\partial t} d\zz, \, t\geq 0, \\
    \frac{\partial \phi_1(\zz_1, \MP_t)}{\partial t} \bigg|_{t=0} =&\, \int \phi_2 (\zz_1,\zz_2,\MP) \frac{\partial \log p_t(\zz_2)}{\partial t} \bigg|_{t=0} d\MP(\zz_2), \, \forall \zz_1 \in \mathcal{Z}.
\end{aligned}
\end{equation*}
Then the function $\phi_2$ is a second-order influence function of $\psi$.
\end{lemma}

After deriving the second-order influence function, we replace the integral with respect to $\MP$ and $\MP \times \MP$ with  a sample average and U-statistics measure in \eqref{eq:2nd-expansion}, and arrive at the general form of the quadratic estimator for a statistical functional $\psi$ as follows
\begin{equation}\label{eq:qr-estimator}
    \hat{\psi}^{qr} = \psi(\hat{\mathbb{P}})+\mathbb{P}_n\{\phi_1(\ZZ,\hat{\mathbb{P}})\} + \frac{1}{2} \MU_n \{ \phi_2(\ZZ_1,\ZZ_2,\hat{\MP}) \}.
\end{equation}

Since we have corrected for both first-order and second-order bias in the plug-in estimator, ideally we may hope $R_3(\hat{\MP}, \MP)$ is a third-order error term and thus the bias of $\hat{\psi}^{qr}$ is smaller than $\hat{\psi}^{dr}$. Unfortunately, the nonexistence of a second-order influence function is typical in many functionals, implying that a third-order error term $R_3(\hat{\MP}, \MP)$ may not be attainable. However, we may still be able to correct for ``partial second-order bias". The key is to find a finite-dimensional subspace $L$ and define an ``approximate second-order influence function" based on $L$, which induces an additional representational error to $R_3(\hat{\MP},\MP)$. The dimension of $L$ is carefully selected to balance the representation error, variance of the quadratic estimator and the remainder term in Von Mises expansion. In Section \ref{sec:Quadratic-estimation} we illustrate how to construct quadratic estimators based on second-order influence functions in generalization and transportation settings.

\section{Proof of Theorem \ref{thm-identification}}

\begin{proof}
For the generalization functional, we have
\begin{equation*}
    \begin{aligned}
    \psi_a = & \, \ME[Y^a] \\
    = & \, \ME [\ME( Y^a |\VV ) ] \\
    = & \, \ME [\ME( Y^a |\VV, S=1 ) ] \\
    = & \, \ME \{ \ME[\ME(Y^a | \XX, S=1)|\VV,S=1] \} \\
    = & \, \ME \{ \ME[\ME(Y^a | \XX, S=1, A=a)|\VV,S=1] \} \\
    = & \, \ME \{ \ME[\ME(Y | \XX, S=1, A=a)|\VV,S=1] \}.
    \end{aligned}
\end{equation*}
The second and fourth equations follow from the tower property of conditional expectation. The third equation follows from Assumption \ref{transportability}, which implies $\ME( Y^a |\VV ) = \ME (Y^a |\VV, S=1)$. The fifth equation follows from Assumption \ref{exchangeability} (conditional exchangeability in the source population). The last equation holds due to Assumption \ref{consistency} (consistency). The identification of transportation can be proved similarly.
\end{proof}

\section{Proof of Theorem \ref{thm-sensitivity}}
\begin{proof}
\begin{equation*}
    \begin{aligned}
    &\,\ME[Y^a] \\
    =&\, \ME \{ \ME[Y^a|\VV]\} \\
    =&\, \ME \{ \ME[Y^a|\VV,S=1]\MP(S=1|\VV) + \ME[Y^a|\VV,S=0]\MP(S=0|\VV)  \} \\
    \leq &\, \ME \{ \ME[Y^a|\VV,S=1]\MP(S=1|\VV) + \ME[Y^a|\VV,S=1]\MP(S=0|\VV) + \delta_2 \MP(S=0|\VV)  \}\\
    &\, \text{ (By Assumption \ref{relax-transportability}) } \\
    = &\, \ME \{ \ME[Y^a|\VV,S=1]\} + \delta_2 \MP(S=0) \\
    = &\, \ME \{ \ME[ \ME(Y^a|\XX,S=1)|\VV,S=1]\} + \delta_2 \MP(S=0) \\
    = &\, \ME \{ \ME[ \ME(Y^a|\XX,A=a,S=1) \pi_a(\XX) + \ME(Y^a|\XX,A=1-a,S=1)(1-\pi_a(\XX))|\VV,S=1]\} + \delta_2 \MP(S=0)\\
    \leq &\, \ME \{ \ME[ \ME(Y^a|\XX,A=a,S=1) \pi_a(\XX) + \ME(Y^a|\XX,A=a,S=1)(1-\pi_a(\XX)) + \delta_1(1-\pi_a(\XX))|\VV,S=1]\} \\
    &\, + \delta_2 \MP(S=0) \text{ (By Assumption \ref{relax-exchangeability}) }\\
    = &\, \ME \{ \ME[ \ME(Y^a|\XX,A=a,S=1)|\VV,S=1] \}  + \delta_1\ME[\MP(A=1-a|\VV,S=1)]+ \delta_2 \MP(S=0)\\
    = &\, \ME \{ \ME[ \ME(Y|\XX,A=a,S=1)|\VV,S=1] \}  + \delta_1\ME[\MP(A=1-a|\VV,S=1)]+ \delta_2 \MP(S=0)
    \end{aligned}
\end{equation*}
The lower bound can be proved similarly. For the transportation functional, we have
\begin{equation*}
    \begin{aligned}
    &\, \ME\{ Y^a|S=0\}\\
    =&\, \ME\{ \ME[Y^a|\VV,S=0] |S=0\} \\
    \leq &\, \ME\{ \ME[Y^a|\VV,S=1] |S=0\} + \delta_2 \\
    =&\, \ME\{ \ME[ \ME(Y^a|\XX,S=1)|\VV,S=1] |S=0\} + \delta_2  \text{ (By Assumption \ref{relax-transportability} )}\\
    =&\, \ME\{ \ME[ \ME(Y^a|\XX,A=a,S=1)\pi_a(\XX)+\ME(Y^a|\XX,A=1-a,S=1)(1-\pi_a(\XX))|\VV,S=1] |S=0\} + \delta_2 \\
    \leq &\, \ME\{ \ME[ \ME(Y^a|\XX,A=a,S=1)\pi_a(\XX)+\ME(Y^a|\XX,A=a,S=1)(1-\pi_a(\XX))|\VV,S=1] |S=0\} \\
    &\, + \delta_1 \ME[ \ME(1-\pi_a(\XX)|\VV,S=1) |S=0] + \delta_2 \text{ (By Assumption \ref{relax-exchangeability}) }\\
    =&\, \ME \{ \ME[ \ME(Y|\XX,A=a,S=1) |\VV,S=1]|S=0  \} + \delta_1 \ME[\MP(A=1-a|\VV,S=1)|S=0]+\delta_2.
    \end{aligned}
\end{equation*}
The lower bound can be achieved similarly. The bounds on the ATE in the target population are obtained by combining the lower bounds and the upper bounds.
\end{proof}

\section{Proof of Theorem \ref{thm-if}}

\begin{proof}
Let $p_{\epsilon}(\zz) = p(\zz ; \epsilon)$ denote a parametric submodel with parameter $\epsilon \in \mathbb{R}$. Denote $S_{\epsilon} (\UU | \OO) = \left.\{\partial \log p_{\epsilon}(\UU|\OO) / \partial \varepsilon\}\right|_{\varepsilon=0}$ as the scores on the parametric submodels $p_{\epsilon}(\uu|\mathbf{o})$, where $\UU$ and $\OO$ are some set of random variables. In our setting, we can decompose the score on the parametric submodels as (recall we assume $\VV \subseteq \XX$ and can write $\XX = (\WW, \VV)$)
\[
S_{\epsilon} (Y,A,\WW,\VV,S) = S_{\epsilon} (Y|A,\WW,\VV,S) + S_{\epsilon} (A | \WW,\VV,S) + S_{\epsilon} (\WW|\VV,S) + S_{\epsilon} (\VV|S) + S_{\epsilon} (S).
\]
We will prove the theorem by checking $\phi_a^{ge}$ is a pathwise derivative in the sense that $\partial \psi_a \left(\MP_{\epsilon}\right) /\left.\partial \epsilon\right|_{\epsilon=0}=\ME\left(\phi_a^{ge} S_{\epsilon}\right)$. First consider the term $\partial \psi_a \left(\MP_{\epsilon}\right) /\left.\partial \epsilon\right|_{\epsilon=0}$. By definition,
\[
\psi_a (\MP) = \iiint y p(y|A=a, S=1, \ww, \vv) \, dy \, p(\ww|\vv,S=1) \, d\ww \, p(\vv) \, d\vv.
\]
So we have
\begin{equation*}
    \begin{aligned}
    \frac{\partial \psi_a(\MP_{\epsilon})}{\partial \epsilon} = & \, \iiint y \frac{\partial p_{\epsilon} (y| A=a, S=1, \ww, \vv)}{\partial \epsilon}  \, dy \, p_{\epsilon}(\ww|\vv, S=1) \, d\ww\, p_{\epsilon}(\vv) \, d\vv \\
    &\, + \iiint y  p_{\epsilon}(y|A=a, S=1, \ww,\vv) \, dy \, \frac{\partial p_{\epsilon} (\ww|\vv, S=1)}{\partial \epsilon}  \, d\ww\, p_{\epsilon}(\vv) \, d\vv \\
    & \, + \iiint y  p_{\epsilon}(y|A=a, S=1, \ww,\vv) \, dy \,  p_{\epsilon}(\ww|\vv, S=1) \, d\ww\, \frac{\partial p_{\epsilon} (\vv)}{\partial \epsilon} \, d\vv. \\
    \end{aligned}
\end{equation*}

\begin{equation}\label{derivative-generalization}
    \begin{aligned}
    \frac{\partial \psi_a(\MP_{\epsilon})}{\partial \epsilon} \bigg|_{\epsilon=0} = & \, \iiint y S_{\epsilon}(y|A=a, S=1, \ww, \vv) p(y|A=a, S=1, \ww,\vv) \, dy \, p(\ww|\vv, S=1) \, d\ww\, p(\vv) \, d\vv \\
    &\, + \iiint y  p(y|A=a, S=1, \ww,\vv) \, dy \, S_{\epsilon} (\ww|\vv, S=1) p(\ww|\vv, S=1) \, d\ww\, p(\vv) \, d\vv \\
    & \, + \iiint y  p(y|A=a, S=1, \ww,\vv) \, dy \,  p(\ww|\vv, S=1) \, d\ww\, S_{\epsilon}(\vv) p(\vv) \, d\vv. \\
    \end{aligned}
\end{equation}
On the other hand, recall
\[
\phi_a^{ge} (\ZZ) = \frac{I(A=a, S=1)(Y-\mu_a(\XX))}{\rho(\VV) \pi_a(\XX)} + \frac{I(S=1)(\mu_a(\XX) - \tau_a(\VV))}{\rho(\VV)} + \tau_a(\VV) - \psi_a.
\]
Then we compute $\ME\left(\phi_a^{ge} S_{\epsilon}\right)$. First note that $\ME[S_{\epsilon}(Y,A,\XX,S)] = 0$ so $\psi_a \ME[S_{\epsilon}(Y,A,\XX,S)] = 0 $. To evaluate $\ME[\tau_a(\VV) S_{\epsilon}]$, we write
\[
\ME[\tau_a(\VV) S_{\epsilon}(Y,A,\XX,S)] = \ME[\tau_a(\VV) S_{\epsilon}(Y,A,\WW,S|\VV)] + \ME[\tau_a(\VV) S_{\epsilon} (\VV)].
\]
By the property of score, we have $\ME[S_{\epsilon}(Y,A,\WW,S|\VV) | \VV] = 0$, which implies
\[
\ME[\tau_a(\VV) S_{\epsilon}(Y,A,\WW,S|\VV)] = \ME \left\{ \tau_a(\VV) \ME[S_{\epsilon}(Y,A,\WW,S|\VV) | \VV] \right\} = 0.
\]
Note that $\ME[\tau_a(\VV) S_{\epsilon} (\VV)]$ is exactly the last term in \eqref{derivative-generalization}.

Next we evaluate 
\[
\ME \left\{ \frac{I(S=1)(\mu_a(\XX) - \tau_a(\VV))}{\rho(\VV)} S_{\epsilon}(Y,A,\XX,S) \right \}.
\]
We decompose $S_{\epsilon}(Y,A,\XX,S)$ as $S_{\epsilon}(Y,A,\XX,S) = S_{\epsilon}(Y,A|\XX,S) + S_{\epsilon}(\WW|\VV,S) + S_{\epsilon}(\VV,S)$ and compute each term separately.
\begin{equation*}
    \begin{aligned}
     \ME \left\{ \frac{I(S=1)(\mu_a(\XX) - \tau_a(\VV))}{\rho(\VV)} S_{\epsilon}(\VV,S) \right \} = & \, \ME\left\{ \frac{I(S=1)S_{\epsilon}(\VV,S)}{\rho(\VV)}  \ME[\mu_a(\XX) - \tau_a(\VV) | \VV, S] \right \} \\
     = &\, \ME\left\{ \frac{I(S=1)S_{\epsilon}(\VV,S)}{\rho(\VV)}  \ME[\mu_a(\XX) - \tau_a(\VV) | \VV, S=1] \right \} = 0
    \end{aligned}
\end{equation*}
since $\ME[\mu_a(\XX) - \tau_a(\VV) | \VV, S=1]  = 0$ by definition. 

\begin{equation*}
    \begin{aligned}
        \ME \left\{ \frac{I(S=1)(\mu_a(\XX) - \tau_a(\VV))}{\rho(\VV)} S_{\epsilon}(\WW|\VV,S) \right \} =& \, \ME \left\{ \frac{I(S=1)\mu_a(\XX)}{\rho(\VV)} S_{\epsilon}(\WW|\VV,S) \right \} \\
        - & \,\ME \left\{ \frac{I(S=1)\tau_a(\VV)}{\rho(\VV)} S_{\epsilon}(\WW|\VV,S) \right \}.
    \end{aligned}
\end{equation*} 
For the second term we have
\[
\ME \left\{ \frac{I(S=1)\tau_a(\VV)}{\rho(\VV)} S_{\epsilon}(\WW|\VV,S) \right \} = \ME \left\{ \frac{I(S=1)\tau_a(\VV)}{\rho(\VV)} \ME [S_{\epsilon}(\WW|\VV,S)| \VV,S]  \right \} = 0.
\]
For the first term
\begin{equation*}
    \begin{aligned}
     & \,\ME \left\{ \frac{I(S=1)\mu_a(\XX)}{\rho(\VV)} S_{\epsilon}(\WW|\VV,S) \right \}\\
     = &\, \ME \left\{ \frac{I(S=1)\mu_a(\XX)}{\rho(\VV)} S_{\epsilon}(\WW|\VV,S=1) \right \} \\
     = & \, \ME \left\{ \frac{I(S=1)}{\rho(\VV)}   \ME[\mu_a(\XX)S_{\epsilon}(\WW|\VV,S=1)|\VV,S=1 ] \right \}(\text{Condition on } \VV,S) \\
     = & \, \ME \left\{ \ME[\mu_a(\XX)S_{\epsilon}(\WW|\VV,S=1)|\VV,S=1 ] (\text{Condition on } \VV)\right\}.
    \end{aligned}
\end{equation*}
This is the second term in \eqref{derivative-generalization}. The remaining term 
\[
\ME \left\{ \frac{I(S=1)(\mu_a(\XX) - \tau_a(\VV))}{\rho(\VV)} S_{\epsilon}(Y,A|\XX,S) \right \} = \ME \left\{ \frac{I(S=1)(\mu_a(\XX) - \tau_a(\VV))}{\rho(\VV)} \ME[S_{\epsilon}(Y,A|\XX,S)|\XX,S] \right \}=0
\]
since $\ME[S_{\epsilon}(Y,A|\XX,S)|\XX,S] = 0$.

Finally we evaluate
\[
\ME \left\{ \frac{I(A=a, S=1)(Y-\mu_a(\XX))}{\rho(\VV) \pi_a(\XX)} S_{\epsilon}(Y,A,\XX,S) \right\}.
\]
First note that 
\begin{equation*}
    \begin{aligned}
    & \,\ME \left\{ \frac{I(A=a, S=1)(Y-\mu_a(\XX))}{\rho(\VV) \pi_a(\XX)} S_{\epsilon}(A,\XX,S) \right\} \\
    = & \, \ME \left\{ \frac{I(A=a, S=1)S_{\epsilon}(A,\XX,S)}{\rho(\VV) \pi_a(\XX)} \ME [Y-\mu_a(\XX)|\XX,A=a, S=1] \right\} (\text{Condition on } A,\XX,S)\\
    = & \, 0
    \end{aligned}
\end{equation*}
by definition of $\mu_a$. 
\begin{equation*}
    \begin{aligned}
     & \,\ME \left\{ \frac{I(A=a, S=1)Y}{\rho(\VV) \pi_a(\XX)} S_{\epsilon}(Y|A,\XX,S) \right\} \\
     = & \,\ME \left\{ \frac{I(A=a, S=1)Y}{\rho(\VV) \pi_a(\XX)} S_{\epsilon}(Y|A=a,\XX,S=1) \right\}  \\
     = & \, \ME \left\{ \frac{I(A=a, S=1)}{\rho(\VV) \pi_a(\XX)} \ME[YS_{\epsilon}(Y|A=a,\XX,S=1) |A=a, \XX, S=1]\right\}  (\text{Condition on } A,\XX,S)\\
     = & \, \ME \left\{ \frac{I(S=1)}{\rho(\VV) }  \ME[YS_{\epsilon}(Y|A=a,\XX,S=1) |A=a, \XX, S=1] \right\} (\text{Condition on } \XX,S) \\
     = & \, \ME \left\{ \frac{I(S=1)}{\rho(\VV) }  \ME \{\ME[YS_{\epsilon}(Y|A=a,\XX,S=1) |A=a, \XX, S=1] |\VV, S=1\} \right\} (\text{Condition on } \VV,S) \\
     = & \, \ME \left\{   \ME \{\ME[YS_{\epsilon}(Y|A=a,\XX,S=1) |A=a, \XX, S=1] |\VV, S=1\} \right\} (\text{Condition on } \VV). \\
    \end{aligned}
\end{equation*}
This is the first term in \eqref{derivative-generalization}. The remaining term
\[
\ME \left\{ \frac{I(A=a, S=1)\mu_a(\XX)}{\rho(\VV) \pi_a(\XX)} S_{\epsilon}(Y|A,\XX,S) \right\} = \ME \left\{ \frac{I(A=a, S=1)\mu_a(\XX)}{\rho(\VV) \pi_a(\XX)} \ME[S_{\epsilon}(Y|A,\XX,S)|A,\XX,S] \right\}=0
\]
since $\ME[S_{\epsilon}(Y|A,\XX,S)|A,\XX,S] =0$. This proves $\partial \psi_a \left(\MP_{\epsilon}\right) /\left.\partial \epsilon\right|_{\epsilon=0}=\ME\left(\phi_a^{ge} S_{\epsilon}\right)$.

We then prove the influence function of transportation functional similarly. First the target functional $\theta_a$ is
\[
\theta_a (\MP) = \iiint y p(y|A=a, S=1, \ww, \vv) \, dy \, p(\ww|\vv,S=1) \, d\ww \, p(\vv | S=0) \, d\vv.
\]

So we have
\begin{equation*}
    \begin{aligned}
    \frac{\partial \theta_a(\MP_{\epsilon})}{\partial \epsilon} = & \, \iiint y \frac{\partial p_{\epsilon} (y| A=a, S=1, \ww, \vv)}{\partial \epsilon}  \, dy \, p_{\epsilon}(\ww|\vv, S=1) \, d\ww\, p_{\epsilon}(\vv | S=0) \, d\vv \\
    &\, + \iiint y  p_{\epsilon}(y|A=a, S=1, \ww,\vv) \, dy \, \frac{\partial p_{\epsilon} (\ww|\vv, S=1)}{\partial \epsilon}  \, d\ww\, p_{\epsilon}(\vv | S=0) \, d\vv \\
    & \, + \iiint y  p_{\epsilon}(y|A=a, S=1, \ww,\vv) \, dy \,  p_{\epsilon}(\ww|\vv, S=1) \, d\ww\, \frac{\partial p_{\epsilon} (\vv | S=0)}{\partial \epsilon} \, d\vv. \\
    \end{aligned}
\end{equation*}

\begin{equation}\label{derivative-transportation}
    \begin{aligned}
    \frac{\partial \theta_a(\MP_{\epsilon})}{\partial \epsilon} \bigg|_{\epsilon=0} = & \, \iiint y S_{\epsilon}(y|A=a, S=1, \ww, \vv) p(y|A=a, S=1, \ww,\vv) \, dy \, p(\ww|\vv, S=1) \, d\ww\, p(\vv|S=0) \, d\vv \\
    &\, + \iiint y  p(y|A=a, S=1, \ww,\vv) \, dy \, S_{\epsilon} (\ww|\vv, S=1) p(\ww|\vv, S=1) \, d\ww\, p(\vv|S=0) \, d\vv \\
    & \, + \iiint y  p(y|A=a, S=1, \ww,\vv) \, dy \,  p(\ww|\vv, S=1) \, d\ww\, S_{\epsilon}(\vv|S=0) p(\vv|S=0) \, d\vv. \\
    \end{aligned}
\end{equation}
Recall the claimed influence function is
\begin{equation*}
	    \begin{aligned}
        \phi_{a}^{\text{tr}} (Z) &=\frac{1}{\MP(S=0)}\left\{\frac{I(A=a, S=1) (1 - \rho(\VV))\left(Y-\mu_{a}(\XX)\right)}{\rho (\VV) \pi_a(\XX)}\right.\\
        &\left.+\frac{I(S=1) (1 - \rho(\VV))\left(\mu_{a}(\XX)-\tau_{a}(\VV)\right)}{\rho(\VV)} +I(S=0)\left[\tau_{a}(\VV)-\theta_{a}\right]\right\}.
        \end{aligned}
\end{equation*}
We first compute 
\begin{equation*}
\begin{aligned}
    & \, \frac{1}{\MP(S=0)}\ME \left[ I(S=0)(\tau_{a}(\VV)-\theta_{a}) S_{\epsilon} (Y,A,\XX,S)\right] \\
    = & \, \frac{1}{\MP(S=0)}\ME \left[ I(S=0)(\tau_{a}(\VV)-\theta_{a}) S_{\epsilon} (S)\right] + \frac{1}{\MP(S=0)}\ME \left[ I(S=0)\tau_{a}(\VV) S_{\epsilon} (Y,A,\XX|S)\right] \\
    - & \,  \frac{\theta_a}{\MP(S=0)}\ME \left[ I(S=0) S_{\epsilon} (Y,A,\XX|S)\right].
\end{aligned}
\end{equation*}
For the first term, we condition on $S$ and have
\[
\ME[I(S=0)(\tau_{a}(\VV)-\theta_{a}) S_{\epsilon} (S)] = \ME[I(S=0)S_{\epsilon} (S) \ME(\tau_{a}(\VV)-\theta_{a} | S=0) ] = 0
\]
since we have $ \ME(\tau_{a}(\VV)-\theta_{a} | S=0)=0$ by definition of $\theta_a$. For the third term we have
\[
\ME \left[ I(S=0) S_{\epsilon} (Y,A,\XX|S)\right] = \ME \left[ I(S=0) \ME \left(S_{\epsilon} (Y,A,\XX|S)|S \right)\right]=0
\]
by the property of score $\ME \left(S_{\epsilon} (Y,A,\XX|S)|S \right)=0$.
For the second term, we further write $S_{\epsilon} (Y,A,\XX|S) = S_{\epsilon} (Y,A,\WW|\VV,S) + S_{\epsilon} (\VV|S)$ and note that 
\begin{equation*}
    \begin{aligned}
    &\, \frac{1}{\MP(S=0)}\ME[I(S=0)\tau_a(\VV)S_{\epsilon}(\VV|S) ] \\
    = &\, \frac{1}{\MP(S=0)}\ME[I(S=0)\tau_a(\VV)S_{\epsilon}(\VV|S=0) ] \\
    = &\, \frac{1}{\MP(S=0)}\ME\left\{I(S=0)\ME[\tau_a(\VV)S_{\epsilon}(\VV|S=0)|S=0] \right\} \text{ (Condition on S) }\\
    = & \, \ME[\tau_a(\VV)S_{\epsilon}(\VV|S=0)|S=0]. 
    \end{aligned}
\end{equation*}
This is the third term in the expression of pathwise derivative in \eqref{derivative-transportation}. The remaining term is 
\[
\ME [I(S=0)\tau_a(\VV)S_{\epsilon} (Y,A,\WW|\VV,S)] = \ME\left\{I(S=0)\tau_a(\VV) \ME[S_{\epsilon} (Y,A,\WW|\VV,S) |\VV,S] \right\}=0.
\]
Next we evaluate 
\[
\frac{1}{\MP (S=0)}\ME \left[\frac{I(S=1) (1 - \rho(\VV))\left(\mu_{a}(\XX)-\tau_{a}(\VV)\right)}{\rho(\VV)} S_{\epsilon}(Y,A,\XX,S) \right].
\]
We decompose $S_{\epsilon}(Y,A,\XX,S) = S_{\epsilon}(\VV,S) + S_{\epsilon}(\WW|\VV,S) + S_{\epsilon}(Y,A|\XX,S)$ and note
\begin{equation*}
    \begin{aligned}
    & \, \ME \left[\frac{I(S=1) (1 - \rho(\VV))\left(\mu_{a}(\XX)-\tau_{a}(\VV)\right)}{\rho(\VV)} S_{\epsilon}(\VV,S) \right] \\
     = & \, \ME \left[ \frac{I(S=1) (1 - \rho(\VV))}{\rho(\VV)} S_{\epsilon}(\VV,S) \ME \left(\mu_{a}(\XX)-\tau_{a}(\VV) |\VV,S=1\right)\right] \\
     = & \, 0
    \end{aligned}
\end{equation*}
since by definition of $\tau_a$ we have $\tau_a(\VV) = \ME [\mu_a(\XX)|\VV,S=1]$. The second term is 
\begin{equation*}
    \begin{aligned}
    &\, \frac{1}{\MP (S=0)}\ME \left[\frac{I(S=1) (1 - \rho(\VV))\left(\mu_{a}(\XX)-\tau_{a}(\VV)\right)}{\rho(\VV)} S_{\epsilon}(\WW|\VV,S) \right] \\
    = &\, \frac{1}{\MP (S=0)}\ME \left[\frac{I(S=1) (1 - \rho(\VV))\mu_{a}(\XX)}{\rho(\VV)} S_{\epsilon}(\WW|\VV,S) \right] \\
    &\, - \frac{1}{\MP (S=0)}\ME \left[\frac{I(S=1) (1 - \rho(\VV))\tau_{a}(\VV)}{\rho(\VV)} S_{\epsilon}(\WW|\VV,S) \right]
    \end{aligned}
\end{equation*}
The first part is
\begin{equation*}
    \begin{aligned}
    & \, \frac{1}{\MP (S=0)}\ME \left[\frac{I(S=1) (1 - \rho(\VV))\mu_{a}(\XX)}{\rho(\VV)} S_{\epsilon}(\WW|\VV,S) \right] \\
    = &\, \frac{1}{\MP (S=0)}\ME \left[\frac{I(S=1) (1 - \rho(\VV))}{\rho(\VV)} \ME(\mu_{a}(\XX)S_{\epsilon}(\WW|\VV,S)|\VV,S=1) \right] (\text{Condition on } \VV,S) \\
    = &\, \frac{1}{\MP (S=0)}\ME \left[ (1 - \rho(\VV)) \ME(\mu_{a}(\XX)S_{\epsilon}(\WW|\VV,S)|\VV,S=1) \right] (\text{Condition on }\VV) \\
    = & \, \iiint y  p(y|A=a, S=1, \ww,\vv) \, dy \, S_{\epsilon} (\ww|\vv, S=1) p(\ww|\vv, S=1) \, d\ww\, \frac{p(\vv)(1-\rho(\vv))}{\MP(S=0)} \, d\vv \\
    = & \, \iiint y  p(y|A=a, S=1, \ww,\vv) \, dy \, S_{\epsilon} (\ww|\vv, S=1) p(\ww|\vv, S=1) \, d\ww\, p(\vv|S=0) \, d\vv
    \end{aligned}
\end{equation*}
where the last equation follows from
\[
\frac{p(\vv)(1-\rho(\vv))}{\MP(S=0)} = \frac{p(\vv)\MP(S=0|\vv)}{\MP(S=0)} = p(\vv|S=0).
\]
This is the second term in \eqref{derivative-transportation}. By conditioning on $(\VV,S)$ we have 
\[
\ME \left[\frac{I(S=1) (1 - \rho(\VV))\tau_{a}(\VV)}{\rho(\VV)} S_{\epsilon}(\WW|\VV,S) \right] = \ME \left[\frac{I(S=1) (1 - \rho(\VV))\tau_{a}(\VV)}{\rho(\VV)} \ME(S_{\epsilon}(\WW|\VV,S)|\VV,S) \right] = 0.
\]
The third term is 
\begin{equation*}
    \begin{aligned}
    &\, \frac{1}{\MP (S=0)}\ME \left[\frac{I(S=1) (1 - \rho(\VV))\left(\mu_{a}(\XX)-\tau_{a}(\VV)\right)}{\rho(\VV)} S_{\epsilon}(Y,A|\XX,S) \right] \\
    = & \, \frac{1}{\MP (S=0)}\ME \left[\frac{I(S=1) (1 - \rho(\VV))\left(\mu_{a}(\XX)-\tau_{a}(\VV)\right)}{\rho(\VV)} \ME [S_{\epsilon}(Y,A|\XX,S)|\XX,S] \right]\\
    = &\, 0.
    \end{aligned}
\end{equation*}
Finally we evaluate 
\[
\frac{1}{\MP(S=0)}\ME \left[\frac{I(A=a, S=1) (1 - \rho(\VV))\left(Y-\mu_{a}(\XX)\right)}{\rho (\VV) \pi_a(\XX)} S_{\epsilon} (Y,A,\XX,S) \right].
\]
By conditioning on $(A,\XX,S)$, we have
\begin{equation*}
    \begin{aligned}
    &\, \ME \left[\frac{I(A=a, S=1) (1 - \rho(\VV))\mu_{a}(\XX)}{\rho (\VV) \pi_a(\XX)} S_{\epsilon} (Y,A,\XX,S) \right] \\
    =&\, \ME \left[\frac{I(A=a, S=1) (1 - \rho(\VV))\mu_{a}(\XX)}{\rho (\VV) \pi_a(\XX)} \ME(S_{\epsilon} (Y,A,\XX,S)|A,\XX,S) \right] \\
    =&\, 0.
    \end{aligned}
\end{equation*}
The remaining term is 
\begin{equation*}
    \begin{aligned}
    &\, \frac{1}{\MP (S=0)}\ME \left[\frac{I(A=a, S=1) (1 - \rho(\VV))Y}{\rho (\VV) \pi_a(\XX)} S_{\epsilon} (Y,A,\XX,S) \right] \\
    =&\, \frac{1}{\MP (S=0)}\ME \left[\frac{I(A=a, S=1) (1 - \rho(\VV))}{\rho (\VV) \pi_a(\XX)}\ME[ YS_{\epsilon} (Y,A,\XX,S)|A=a,S=1,\XX] \right] (\text{Condition on }A,\XX,S) \\
    =&\, \frac{1}{\MP (S=0)}\ME \left[\frac{I(S=1) (1 - \rho(\VV))}{\rho (\VV) }\ME[ YS_{\epsilon} (Y,A,\XX,S)|A=a,S=1,\XX] \right] (\text{Condition on }\XX,S) \\
    =&\, \frac{1}{\MP (S=0)}\ME \left[\frac{I(S=1) (1 - \rho(\VV))}{\rho (\VV) } \ME[ \ME( YS_{\epsilon} (Y,A,\XX,S)|A=a,S=1,\XX) |\VV,S=1] \right] (\text{Condition on }\VV,S) \\
    =&\, \frac{1}{\MP (S=0)}\ME \left[(1 - \rho(\VV)) \ME[ \ME( YS_{\epsilon} (Y,A,\XX,S)|A=a,S=1,\XX) |\VV,S=1] \right] (\text{Condition on }\VV) \\
    = &\, \iiint y S_{\epsilon}(y|A=a, S=1, \ww, \vv) p(y|A=a, S=1, \ww,\vv) \, dy \, p(\ww|\vv, S=1) \, d\ww\, p(\vv|S=0) \, d\vv
    \end{aligned}
\end{equation*}
This is the first term in the pathwise derivative in \eqref{derivative-transportation}.

The efficiency bounds are given by the variance of the influence functions. Note that the means of all cross-terms are zero, for instance,
\begin{equation*}
    \begin{aligned}
    &\, \ME\left[\frac{I(A=a, S=1)(Y-\mu_a(\XX))}{\rho(\VV) \pi_a(\XX)} \frac{I(S=1)(\mu_a(\XX) - \tau_a(\VV))}{\rho(\VV)} \right] \\
    = &\, \ME\left[\frac{I(A=a, S=1)}{\rho(\VV) \pi_a(\XX)} \frac{I(S=1)(\mu_a(\XX) - \tau_a(\VV))}{\rho(\VV)} \ME (Y-\mu_a(\XX)|\XX,S,A)\right]\\
    = &\, \ME\left[\frac{I(A=a, S=1)}{\rho(\VV) \pi_a(\XX)} \frac{I(S=1)(\mu_a(\XX) - \tau_a(\VV))}{\rho(\VV)} \ME (Y-\mu_a(\XX)|\XX,S=1,A=a)\right]\\
    =&\, 0.
    \end{aligned}
\end{equation*}
Then the variance is given by the expectation of three individual terms squared, as summarized in the theorem.
\end{proof}

\section{Proof of Theorem \ref{thm-dr-generalization}}
\begin{proof}

In this section, all expectations are taken over a new sample $(Y,A,\XX,S)$ independent of the samples used to train the nuisance functions. For a general function $f$ on the sample $(Y,A,\XX,S)$, we have
\begin{equation}\label{eq:decompose-error}
    \MP_n[\hat{f}] - \ME[f] = (\MP_n - \ME) (\hat{f} - f) + (\MP_n -\ME)(f) + \ME(\hat{f}-f)
\end{equation}
We apply the decomposition of error \eqref{eq:decompose-error} to $\varphi_a^{\text{ge}}$. Note that $\ME[\varphi_a^{\text{ge}}] = \psi_a$. Since we are using sample splitting, by Lemma 2 in \cite{kennedy2020sharp} and $\left\|\hat{\varphi}_{a}^{\text{ge}}-\varphi_{a}^{\text{ge}}\right\|_{2}=o_{p}(1)$ we have
\[
(\MP_n - \ME) (\hat{\varphi}_{a}^{\text{ge}} - \varphi_{a}^{\text{ge}})=o_{p}\left(1/\sqrt{n}\right).
\]
We then focus on the conditional bias term
\begin{equation*}
    \begin{aligned}
    &\, \ME(\hat{\varphi}_{a}^{\text{ge}} - \varphi_{a}^{\text{ge}}) \\
    = &\, \ME \left[ \frac{I(A=a,S=1)}{\hat{\pi}_a(\XX)\hat{\rho}(\VV)}(\mu_a(\XX)-\hat{\mu}_a(\XX)) + \frac{I(S=1)(\tilde{\tau}_a(\VV) -\tau_a(\VV))}{\hat{\rho}(\VV)} + \frac{(\hat{\tau}_a(\VV) -\tau_a(\VV))(\hat{\rho}(\VV) - I(S=1))}{\hat{\rho}(\VV)} \right],
    \end{aligned}
\end{equation*}
where $\tilde{\tau}_a(\VV) = \ME[\hat{\mu}_a(\XX)|\VV,S=1]$.
By conditioning on $(\VV,S)$ we have
\[
\ME\left[ \frac{I(S=1)(\hat{\mu}_a(\XX) -\mu_a(\XX))}{\hat{\rho}(\VV)}\right] = \ME\left[ \frac{I(S=1)(\tilde{\tau}_a(\VV) -\tau_a(\VV))}{\hat{\rho}(\VV)}\right].
\]
Hence the conditional bias can be rewritten as
\[
\ME \left \{\left[ \frac{I(S=1)}{\hat{\rho}(\VV)}-\frac{I(A=a,S=1)}{\hat{\pi}_a(\XX)\hat{\rho}(\VV)} \right](\hat{\mu}_a(\XX)-\mu_a(\XX)) + \frac{(\hat{\tau}_a(\VV) -\tau_a(\VV))(\hat{\rho}(\VV) - \rho(\VV))}{\hat{\rho}(\VV)} \right\}.
\]
Since $\hat{\rho} \geq \epsilon$ and by Cauchy-Schward inequality the second term is bounded by
\[
\frac{\|\hat{\tau}_a - \tau_a\| \|\hat{\rho} -\rho\|}{\epsilon}
\]
and by assumption this term is $o_{p}\left(1/\sqrt{n}\right).$ The first term can be expressed as
\begin{equation*}
    \begin{aligned}
    &\, \ME \left \{\left[ \frac{I(S=1)}{\hat{\rho}(\VV)}-\frac{I(A=a,S=1)}{\hat{\pi}_a(\XX)\hat{\rho}(\VV)} \right](\hat{\mu}_a(\XX)-\mu_a(\XX)) \right \} \\
    = &\, \ME \left\{ \frac{\MP(S=1|\XX)}{\hat{\rho}(\VV) \hat{\pi}_a(\XX)} (\hat{\mu}_a(\XX) - \mu_a(\XX))(\hat{\pi}_a(\XX) - \pi_a(\XX))  \right\}
    \end{aligned}
\end{equation*}
By the boundness assumption on $\hat{\rho}, \hat{\pi}_a$ this term is bounded by
\[
\frac{\|\hat{\mu}_a - \mu_a\| \|\hat{\pi}_a -\pi_a\|}{\epsilon^2} = o_{p}\left(1/\sqrt{n}\right). 
\]
Thus under the conditions in Theorem \ref{thm-dr-generalization}, we have
\[
\MP_n(\hat{\varphi}_a^{\text{ge}}) -\ME(\varphi_a^{\text{ge}}) = (\MP_n - \ME) (\varphi_a^{\text{ge}}) + o_{p}\left(1/\sqrt{n}\right).
\]
The proof is completed by noting $(\MP_n - \ME) (\varphi_a^{\text{ge}}) = \MP_n (\phi_a^{\text{ge}})$.

\end{proof}

\section{Proof of Theorem \ref{thm-dr-transportation}}
\begin{proof}
We apply error decomposition \eqref{eq:decompose-error} to $\varphi_a^{\text{tr}}$ and similar to generalization functional we have
\[
(\MP_n - \ME) (\hat{\varphi}_{a}^{\text{tr}} - \varphi_{a}^{\text{tr}})=o_{p}\left(1/\sqrt{n}\right)
\]
under the conditions in the Theorem. The conditional bias is 
\begin{equation*}
    \begin{aligned}
    \ME(\hat{\varphi}_{a}^{\text{tr}} - \varphi_{a}^{\text{tr}}) =&\, \ME \left[ \frac{I(A=a,S=1) (1-\hat{\rho}(\VV))}{\hat{\rho}(\VV)\hat{\pi}_a(\XX)} (\mu_a(\XX)-\hat{\mu}_a(\XX)) \right] \\
    &\, + \ME \left[ \frac{I(S=1)(1-\hat{\rho}(\VV))}{\hat{\rho}(\VV)}(\tilde{\tau}_a(\VV) - \hat{\tau}_a(\VV))\right] + \ME[I(S=0)(\hat{\tau}_a(\VV) - \tau_a(\VV)) ]
    \end{aligned}
\end{equation*}
We write the second term as 
\[
\ME \left[ \frac{I(S=1)(1-\hat{\rho}(\VV))}{\hat{\rho}(\VV)}(\tilde{\tau}_a(\VV) - \tau_a(\VV))\right] - \ME \left[ \frac{I(S=1)(1-\hat{\rho}(\VV))}{\hat{\rho}(\VV)}(\hat{\tau}_a(\VV) - \tau_a(\VV))\right]
\]
and note 
\[
\ME \left[ \frac{I(S=1)(1-\hat{\rho}(\VV))}{\hat{\rho}(\VV)}(\tilde{\tau}_a(\VV) - \tau_a(\VV))\right] = \ME \left[ \frac{I(S=1)(1-\hat{\rho}(\VV))}{\hat{\rho}(\VV)}(\hat{\mu}_a(\XX) - \mu_a(\XX))\right].
\]
Thus the conditional bias can be written as 
\begin{equation}\label{eq:bias-dr-transportation}
    \begin{aligned}
    &\, \ME \left[ \frac{(1-\hat{\rho}(\VV))(\hat{\mu}_a(\XX) -\mu_a(\XX))}{\hat{\rho}(\VV)} \left( I(S=1) - \frac{I(A=a,S=1)}{\hat{\pi}_a(\XX)} \right)\right] \\
    &\, + \ME \left[ \left( I(S=0) -\frac{I(S=1)(1-\hat{\rho}(\VV))}{\hat{\rho}(\VV)} \right) (\hat{\tau}_a(\VV) -\tau_a(\VV)) \right].
    \end{aligned}
\end{equation}
For the first term in \eqref{eq:bias-dr-transportation}, we condition on $\XX$ and have
\begin{equation*}
    \begin{aligned}
    &\, \ME \left[ \frac{(1-\hat{\rho}(\VV))(\hat{\mu}_a(\XX) -\mu_a(\XX))}{\hat{\rho}(\VV)} \left( I(S=1) - \frac{I(A=a,S=1)}{\hat{\pi}_a(\XX)} \right)\right] \\
    = &\, \ME \left[ \frac{\MP(S=1|\XX)(1-\hat{\rho}(\VV))(\hat{\mu}_a(\XX) -\mu_a(\XX))}{\hat{\rho}(\VV)\hat{\pi}_a(\XX)} \left( \hat{\pi}_a(\XX) - \pi_a(\XX) \right)\right].
    \end{aligned}
\end{equation*}
Then by the boundness of $\pi_a, \hat{\rho}$ and Cauchy-Schwarz inequality this term is bounded by
\[
\frac{\|\hat{\mu}_a- \mu_a\|\|\hat{\pi}_a -\pi_a\|}{\epsilon^2} = o_p(1/\sqrt{n})
\]
Conditioning on $\VV$, the second term in \eqref{eq:bias-dr-transportation} is 
\[
\ME \left[ \left( I(S=0) -\frac{I(S=1)(1-\hat{\rho}(\VV))}{\hat{\rho}(\VV)} \right) (\hat{\tau}_a(\VV) -\tau_a(\VV)) \right] = \ME \left[ \frac{(\hat{\rho}(\VV) - \rho(\VV))(\hat{\tau}_a(\VV) -\tau_a(\VV))}{\hat{\rho}(\VV)} \right].
\]
This term is bounded by 
\[
\frac{\|\hat{\rho}-\rho\|\|\hat{\tau}_a -\tau_a\|}{\epsilon}=o_p(1/\sqrt{n}).
\]
Hence we have
\[
\MP_n (\hat{\varphi}_a^{\text{tr}}) - \ME(\varphi_a^{\text{tr}}) = \MP_n[\varphi_a^{\text{tr}} - \ME(\varphi_a^{\text{tr}})] + o_p(1/\sqrt{n}).
\]
Note that 
\[
\frac{\ME[\varphi_a^{\text{tr}}]}{\MP(S=0)} = \theta_a.
\]
We have
\begin{equation*}
\begin{aligned}
    &\,\hat{\theta}_a^{dr} - \theta_a \\
    = &\, \frac{\MP_n (\hat{\varphi}_a^{\text{tr}})}{\MP_n(1-S)} - \theta_a \\
    =&\, \frac{\MP_n (\hat{\varphi}_a^{\text{tr}}) - \theta_a \MP_n(1-S)}{\MP_n(1-S)}  \\
    =&\, \frac{\MP_n (\hat{\varphi}_a^{\text{tr}}) - \theta_a \MP_n(1-S)}{\MP(S=0)} + \left(\frac{1}{\MP_n(1-S)} - \frac{1}{\MP(S=0)} \right) (\MP_n (\hat{\varphi}_a^{\text{tr}}) - \theta_a \MP_n(1-S)) \\
    =&\, \MP_n (\phi_a^{\text{tr}}) + \frac{\MP_n(1-S) - \MP(S=0)}{\MP_n(1-S)\MP(S=0)}(\MP_n (\hat{\varphi}_a^{\text{tr}}) - \theta_a \MP_n(1-S)).
\end{aligned}
\end{equation*}
Since $\MP(S=0) > 0$ we have
\[
\frac{\MP_n(1-S) - \MP(S=0)}{\MP_n(1-S)\MP(S=0)} = O_p(1/\sqrt{n}).
\]
Also note
\begin{equation*}
    \begin{aligned}
    \MP_n (\hat{\varphi}_a^{\text{tr}}) - \theta_a \MP_n(1-S) =&\, \MP_n (\hat{\varphi}_a^{\text{tr}}) -\ME(\varphi_a^{\text{tr}})- \theta_a (\MP_n(1-S) -\MP(S=0)) \\
    =&\,\MP_n[\varphi_a^{\text{tr}} - \ME(\varphi_a^{\text{tr}})] - \theta_a (\MP_n(1-S) -\MP(S=0))+ o_p(1/\sqrt{n}) \\
    = &\, O_p(1/\sqrt{n})
    \end{aligned}
\end{equation*}
since after multiplying this term by $\sqrt{n}$ it is asymptotically normal. Hence we have
\[
\frac{\MP_n(1-S) - \MP(S=0)}{\MP_n(1-S)\MP(S=0)}(\MP_n (\hat{\varphi}_a^{\text{tr}}) - \theta_a \MP_n(1-S)) = O_{p}(1/n) = o_p(1/\sqrt{n}),
\]
which completes the proof.
\end{proof}

\section{Proof of Theorem \ref{thm-minimax}}
\begin{proof}
The proof mainly mimics the strategy in \cite{robins2009semiparametric} and we adapt it to our setting. First in the smooth regime $\alpha + \beta \geq d/2$ the minimax rate can be proved by a standard two-point method and we focus on the under-smooth regime. Let $H: \mathbb{R}^d \mapsto \mathbb{R}$ be a $C^{\infty}$ function supported on $[-1/2,1/2]^d$ such that $\int_{[-1/2,1/2]^d} H(\xx) d\xx=0$ and $\int_{[-1/2,1/2]^d} H^2(\xx) d\xx=1$. Let $\mathcal{X}_1, \dots, \mathcal{X}_k$ denote a partition of $[0,1]^d$ into $k$ cubes of equal size with midpoints $m_1, \dots, m_k$. Each $\mathcal{X}_i$ has side length $k^{-1/d}$. Note that $H\left(k^{1/d}(\xx-m_j)\right)$ is only supported on $\mathcal{X}_j$ and by change of variables
\[
\int_{\mathcal{X}_j} H^2\left(k^{1/d}(\xx-m_j) \right) d\xx = \frac{1}{k}\int_{[-1/2,1/2]^d} H^2(\xx) d\xx = \frac{1}{k}.
\]
Set 
\begin{equation*}
    \begin{aligned}
    \frac{1}{(\pi_a\rho)_{\boldsymbol{\lambda}} (\xx)} =&\, 2 + k^{-\alpha/d}\sum_{i=1}^k \lambda_i H \left( k^{1/d}(\xx-m_i) \right),\\
    (\mu_a)_{\boldsymbol{\lambda}}(\xx) =&\, \frac{1}{2} + k^{-\beta/d} \sum_{i=1}^k \lambda_i H \left( k^{1/d}(\xx-m_i) \right),
    \end{aligned}
\end{equation*}
where $\lambda_1,\dots, \lambda_k$ are i.i.d. Rademacher variables. When $k$ is sufficiently large any $\MP_{\boldsymbol{\lambda}}$ with $(\mu_a)_{\boldsymbol{\lambda}}, (\pi_a\rho)_{\boldsymbol{\lambda}}$ defined as above is in class $\mathcal{P}_{\text{ge}}$. Partition the sample space $\mathcal{X} \times \mathcal{S} \times \mathcal{A} \times \mathcal{Y} = [0,1]^d \times \{0,1\} \times \{0,1\} \times \{0,1\}$ as $\cup_{i=1}^k (\mathcal{X}_i \times \{0,1\} \times \{0,1\} \times \{0,1\}) = \cup_{i=1}^k \mathcal{Z}_i$. For any distribution $\MP_{\boldsymbol{\lambda}} \in \mathcal{P}_{\text{ge}}$ with nuisance functions defined above (so the density $f=\frac{1}{2\pi_a\rho}$), we have
\[
p_j = \MP_{\boldsymbol{\lambda}}(\XX \in \mathcal{X}_j) = \frac{1}{k}.
\]
Let $T=I(S=1, A=a)$. The likelihood (where we take $Y=1$ when $T=0$) is
\[
p(\xx,t,y) = \begin{cases} (1-\pi_a(\xx)\rho(\xx))f(\xx) = \frac{1}{2} (\frac{1}{\rho(\xx)\pi_a(\xx)}-1) & \text { if } t=0 \\ f(\xx)\rho(\xx)\pi_a(\xx) \mu_a(\xx)^y (1-\mu_a(\xx))^{1-y} = \frac{1}{2} \mu_a(\xx)^y (1-\mu_a(\xx))^{1-y}& \text { if } t=1.\end{cases}
\]
The conditional distribution $I_{\mathcal{Z}_j} d \MP_{\boldsymbol{\lambda}} / p_j$ clearly only depends on $\lambda_j$ since $H\left(k^{1/d}(\xx-m_i)\right)=0$ on $\mathcal{X}_j$ for $i\neq j$. If we set distribution $\MP_{\boldsymbol{\lambda}}$ with nuisance values $(\frac{1}{\pi_a \rho}, \mu_a, f\pi_a \rho) = (\frac{1}{(\pi_a\rho)_{\boldsymbol{\lambda}} }, \frac{1}{2}, \frac{1}{2})$ or $(2, (\mu_a)_{\boldsymbol{\lambda}}, \frac{1}{2})$ the generalization functional equals
\[
\psi_a(\MP_{\boldsymbol{\lambda}}) = \int \mu_a(\xx)f(\xx) d\xx = \frac{1}{2}.
\]
And the functional at $\mathbb{Q}_{\boldsymbol{\lambda}}$ with nuisance values $ (\frac{1}{(\pi_a\rho)_{\boldsymbol{\lambda}} }, (\mu_a)_{\boldsymbol{\lambda}}, \frac{1}{2})$ equals (note functions $H \left( k^{1/d}(\xx-m_i) \right)$'s are support on disjoint sets)
\begin{equation*}
    \begin{aligned}
    \psi_a(\mathbb{Q}_{\boldsymbol{\lambda}}) =&\, \int \left[\frac{1}{2} + k^{-\beta/d} \sum_{i=1}^k \lambda_i H \left( k^{1/d}(\xx-m_i) \right) \right] \left[1 + \frac{1}{2} k^{-\alpha/d}\sum_{i=1}^k \lambda_i H \left( k^{1/d}(\xx-m_i) \right) \right] d\xx \\
    =&\, \frac{1}{2} + \frac{1}{2} k^{-(\alpha+\beta)/d}\sum_{i=1}^k \int H^2 \left( k^{1/d}(\xx-m_i) \right) d\xx\\
    =&\, \frac{1}{2} + \frac{1}{2} k^{-(\alpha+\beta)/d}\sum_{i=1}^k \frac{1}{k} \\
    =&\, \frac{1}{2} + \frac{1}{2} k^{-(\alpha+\beta)/d}.
    \end{aligned}
\end{equation*}
In the case $\alpha \geq \beta$, we let $\MP_{\boldsymbol{\lambda}} = (2,(\mu_a)_{\boldsymbol{\lambda}}, \frac{1}{2})$ and $\mathbb{Q}_{\boldsymbol{\lambda}} = (\frac{1}{(\pi_a\rho)_{\boldsymbol{\lambda}}}, (\mu_a)_{\boldsymbol{\lambda}}, \frac{1}{2})$, if $t=0$ we have
\[
p_{\boldsymbol{\lambda}}(\xx,t,y) = \frac{1}{2}, q_{\boldsymbol{\lambda}}(\xx,t,y) = \frac{1}{2} \left(\frac{1}{(\pi_a\rho)_{\boldsymbol{\lambda}}} (\xx) -1 \right), \bar{p} = \int p_{\boldsymbol{\lambda}} d\pi(\boldsymbol{\lambda}) = \frac{1}{2}, \bar{q} = \int q_{\boldsymbol{\lambda}}d\pi(\boldsymbol{\lambda}) = \frac{1}{2}.
\]
If $t=1$ we have
\[
p_{\boldsymbol{\lambda}}(\xx,t,y) = \frac{1}{2}[y(\mu_a)_{\boldsymbol{\lambda}}(\xx) + (1-y)(1-(\mu_a)_{\boldsymbol{\lambda}}(\xx))]=\frac{1}{2}\left[ \frac{1}{2} + (2y-1)k^{-\beta/d} \sum_{i=1}^k \lambda_i H \left( k^{1/d}(\xx-m_i) \right) \right]
\]
\[
\bar{p} = \frac{1}{4}, q_{\boldsymbol{\lambda}} = p_{\boldsymbol{\lambda}}, \bar{q}=\bar{p}.
\]
So we can write 
\[
p_{\boldsymbol{\lambda}} - \bar{p} = \frac{1}{2}t(2y-1)[(\mu_a)_{\boldsymbol{\lambda}}(\xx) -\frac{1}{2}],
\]
\[
q_{\boldsymbol{\lambda}} - p_{\boldsymbol{\lambda}} = \frac{1}{2}(1-t)\left[\frac{1}{(\pi_a\rho)_{\boldsymbol{\lambda}}(\xx)}-2 \right],
\]
\[
\bar{p} -\bar{q}=0.
\]
Hence we can bound $\delta_i$'s in Lemma \ref{bound-Hellinger}. Note that if we set $k$ sufficiently large, $p_{\boldsymbol{\lambda}}$ will be bounded away from 0 and $p_j=1/k$. We have
\begin{equation*}
    \begin{aligned}
    \delta_1 \lesssim &\, k \max_j \int_{\mathcal{X}_j} k^{-2\beta/d} \left[\sum_{i=1}^k \lambda_i H \left( k^{1/d}(\xx-m_i) \right) \right]^2 d\xx \\
    \lesssim & \, k^{-2\beta/d+1} \max_j \int_{\mathcal{X}_j} \sum_{i=1}^k H^2 \left( k^{1/d}(\xx-m_i) \right) d\xx \\
    =&\, k^{-2\beta/d+1} \max_j \int_{\mathcal{X}_j}  H^2 \left( k^{1/d}(\xx-m_j) \right) d\xx \\
    =&\, k^{-2\beta/d}
    \end{aligned}
\end{equation*}
where we again use the fact that the functions $H \left( k^{1/d}(\xx-m_i) \right)$'s are support on disjoint sets $\mathcal{X}_i$'s.
Similarly we have
\[
\delta_2 \lesssim k^{-2\alpha/d}, \delta_3=0.
\]
By Lemma \ref{bound-Hellinger} we can show the following bound on Hellinger distance
\[
H^2\left(\int \MP_{\boldsymbol{\lambda}}^n d \pi(\boldsymbol{\lambda}), \int \mathbb{Q}_{\boldsymbol{\lambda}}^n d \pi(\boldsymbol{\lambda})\right) \leq  \frac{Cn^2}{k} (k^{-2(\alpha+\beta)/d} + k^{-4\alpha/d}).
\]
The dominating term in the bracket is $k^{-2(\alpha+\beta)/d}$ since $\alpha \geq \beta$. By taking $k = C^{\prime} n^{2d/(2\alpha+ 2\beta+d)}$ with $C^{\prime}$ sufficiently large we have 
\[
H^2\left(\int \MP_{\boldsymbol{\lambda}}^n d \pi(\boldsymbol{\lambda}), \int \mathbb{Q}_{\boldsymbol{\lambda}}^n d \pi(\boldsymbol{\lambda})\right) \leq c < 2.
\]
By Lemma \ref{lemma-minimax} the minimax lower bound is at least 
\[
\inf _{\widehat{\psi}_a} \sup _{\MP \in \mathcal{P}_{g e}}\left(\mathbb{E}_\MP\left(\widehat{\psi}_a-\psi_a\right)^2\right)^{1 / 2} \geq  \inf _{\widehat{\psi}_a} \sup _{\MP \in \mathcal{P}_{g e}}\left(\mathbb{E}_\MP|\widehat{\psi}_a-\psi_a|\right) \geq k^{-(\alpha+\beta)/d} \gtrsim n^{-2(\alpha+\beta)/(2\alpha+2\beta+d)}.
\]
The results on transportation functional can be similarly proved by further setting 
\[
\frac{1}{(\pi_a)_{\boldsymbol{\lambda}} (\xx)} =\, \frac{3}{4}\left[2 + k^{-\alpha/d}\sum_{i=1}^k \lambda_i H \left( k^{1/d}(\xx-m_i) \right) \right],
\]
\[
\rho_{\boldsymbol{\lambda}}(\xx) = \frac{3}{4}.
\]
Then $S$ is independent of $\XX$ and the transportation functional equals the generalization functional. The calculations on the bounds of Hellinger distance still hold.

In the case $\alpha \leq \beta$, we let $\MP_{\boldsymbol{\lambda}} = (\frac{1}{(\pi_a\rho)_{\boldsymbol{\lambda}}},\frac{1}{2}, \frac{1}{2})$ and $\mathbb{Q}_{\boldsymbol{\lambda}} = (\frac{1}{(\pi_a\rho)_{\boldsymbol{\lambda}}}, (\mu_a)_{\boldsymbol{\lambda}}, \frac{1}{2})$ and can similarly prove the results.
\end{proof}

\section{Proof of Theorem \ref{thm-qr-generalization}}
We first prove a lemma that will be useful in the following derivations.
\begin{lemma}\label{lemma-basis}
Let $\bb : \mathbb{R}^d \mapsto \mathbb{R}^k$ be a $k$-dimensional basis and $w$ be a weight function. Define
\begin{equation*}
    \begin{aligned}
\mathbf{\Omega} &=\int \bb(\xx) \bb(\xx)^{\top} w(\xx) d F(\xx), \\
\widehat{\mathbf{\Omega}} &=\int \bb(\xx) \bb(\xx)^{\top} \hat{w}(\xx) d \hat{F}(\xx), \\
\Pi_{\bb}(g)(\xx) &=\bb(\xx)^{\top} \mathbf{\Omega}^{-1} \int \bb(\uu) g(\uu) w(\uu) d F(\uu), \\
\widehat{\Pi}_{\bb}(g)(\xx) &=\bb(\xx)^{\top} \widehat{\mathbf{\Omega}}^{-1} \int \bb(\uu) g(\uu) w(\uu) d F(\uu) .
\end{aligned}
\end{equation*}
Then we have
\begin{enumerate}
    \item For any two functions $g_1,g_2$ with $\|g_1\|_w < \infty, \|g_2\|_w < \infty$, 
    \[
    \left|\int g_1(\xx) (\Pi_{\bb} -\hat{\Pi}_{\bb}) (g_2) (\xx) w(\xx) dF(\xx) \right| \leq \|\mathbf{\Omega}\| \|g_1\|_w\|g_2\|_w\|\mathbf{\Omega}^{-1} - \widehat{\mathbf{\Omega}}^{-1}\|.
    \]
    \item 
    \[
    \int \left( \bb(\xx_1)^{\top} \widehat{\mathbf{\Omega}}^{-1} \bb(\xx_2) \right)^2 w(\xx_1) w(\xx_2) dF(\xx_1) dF(\xx_2) \leq 2k(1+\|\mathbf{\Omega}\|^2\|\widehat{\mathbf{\Omega}}^{-1} - \mathbf{\Omega}^{-1}\|^2).
    \]
    \item For any function $g$ with $\|g\|_w < \infty$, the estimated projection of $g$ has norm bound
    \[
    \|\hat{\Pi}_{\bb} (g)\|_w^2 \leq \|\mathbf{\Omega}\|\|\widehat{\mathbf{\Omega}}^{-1}\mathbf{\Omega} \widehat{\mathbf{\Omega}}^{-1}\| \|g\|_w^2
    \]
\end{enumerate}
\end{lemma}
\begin{proof}
(1) Let $\ttt_g = \mathbf{\Omega}^{-1} \int \bb(\xx)g(\xx)w(\xx)dF(\xx)$ be the coefficient vector of the projection, we have
\begin{equation*}
    \begin{aligned}
    \int g_1(\xx) \Pi_{\bb} (g_2) (\xx) w(\xx) d F(\xx) =&\, \int \Pi_{\bb}(g_1)(\xx) \Pi_{\bb} (g_2) (\xx) w(\xx) d F(\xx) \\
    =&\, \int \ttt_{g_1}^\top \bb(\xx) \bb(\xx)^{\top} \ttt_{g_2} w(\xx) dF(\xx) = \ttt_{g_1}^{\top}\mathbf{\Omega} \ttt_{g_2}.
    \end{aligned}
\end{equation*}
Note that 
\[
\hat{\Pi}_{\bb} (g) (\xx) = \bb(\xx)^{\top}\widehat{\mathbf{\Omega}}^{-1}\mathbf{\Omega}\ttt_{g}
\]
is still a linear combination of $\bb$, we have
\begin{equation*}
    \begin{aligned}
    \int g_1(\xx) \hat{\Pi}_{\bb} (g_2) (\xx) w(\xx) d F(\xx) =&\, \int \Pi_{\bb}(g_1)(\xx) \hat{\Pi}_{\bb} (g_2) (\xx) w(\xx) d F(\xx) \\
    =&\, \int \ttt_{g_1}^\top \bb(\xx) \bb(\xx)^{\top} \widehat{\mathbf{\Omega}}^{-1}\mathbf{\Omega}\ttt_{g_2} w(\xx) dF(\xx) = \ttt_{g_1}^{\top}\mathbf{\Omega} \widehat{\mathbf{\Omega}}^{-1}\mathbf{\Omega}\ttt_{g_2}.
    \end{aligned}
\end{equation*}
Hence 
\begin{equation*}
    \begin{aligned}
    &\, \int g_1(\xx) (\Pi_{\bb} -\hat{\Pi}_{\bb}) (g_2) (\xx) w(\xx) dF(\xx)\\
    =&\, \ttt_{g_1}^{\top}\mathbf{\Omega} \ttt_{g_2} - \ttt_{g_1}^{\top}\mathbf{\Omega} \widehat{\mathbf{\Omega}}^{-1}\mathbf{\Omega}\ttt_{g_2} \\
    =&\, (\mathbf{\Omega}^{1/2}\ttt_{g_1})^{\top} \mathbf{\Omega}^{1/2}(\mathbf{\Omega}^{-1} - \widehat{\mathbf{\Omega}}^{-1}) \mathbf{\Omega}^{1/2}(\mathbf{\Omega}^{1/2}\ttt_{g_2}).
    \end{aligned}
\end{equation*}
By Cauchy-Schwarz inequality and property of operator norm we have
\[
\left|\int g_1(\xx) (\Pi_{\bb} -\hat{\Pi}_{\bb}) (g_2) (\xx) w(\xx) dF(\xx) \right| \leq \|\mathbf{\Omega}^{1/2}\ttt_{g_1}\| \|\mathbf{\Omega}\| \|\mathbf{\Omega}^{-1} - \widehat{\mathbf{\Omega}}^{-1}\| \|\mathbf{\Omega}^{1/2}\ttt_{g_2}\|.
\]
Further note $\|\mathbf{\Omega}^{1/2}\ttt_{g}\| = \|\Pi_{\bb}(g)\|_w \leq \|g\|_w$ and then we obtain the bound in (1).

(2) Let $\bar{\bb}(\xx) = \mathbf{\Omega}^{-1 / 2} \bb(\xx)$ be the orthonormal basis and $\MM = \mathbf{\Omega}^{1/2} \widehat{\mathbf{\Omega}}^{-1}\mathbf{\Omega}^{1/2}$.
\begin{equation*}
    \begin{aligned}
    &\,\int \left( \bb(\xx_1)^{\top} \widehat{\mathbf{\Omega}}^{-1} \bb(\xx_2) \right)^2 w(\xx_1) w(\xx_2) dF(\xx_1) dF(\xx_2) \\
    =&\, \int \left( \bar{\bb}(\xx_1)^{\top} \MM \bar{\bb}(\xx_2) \right)^2 w(\xx_1) w(\xx_2) dF(\xx_1) dF(\xx_2) \\
    = & \, \int\left\{\sum_{j, \ell} M_{j \ell} \bar{\bb}_j\left(\xx_1\right) \bar{\bb}_{\ell}\left(\xx_2\right)\right\}^2  w(\xx_1) w(\xx_2)d F\left(\xx_1\right) d F\left(\xx_2\right) \\
    =&\, \sum_{j, \ell} \sum_{j^{\prime}, \ell^{\prime}} M_{j \ell} M_{j^{\prime} \ell^{\prime}} \int \bar{\bb}_j\left(\xx_1\right) \bar{\bb}_{j^{\prime}}\left(\xx_1\right)w(\xx_1) d F\left(\xx_1\right) \int \bar{\bb}_{\ell}\left(\xx_2\right) \bar{\bb}_{\ell^{\prime}}\left(\xx_2\right) w(\xx_2)d F\left(\xx_2\right) \\
    =&\, \sum_{j, \ell} M_{j \ell}^2 \\
    =&\, \|\MM\|_F^2 \\
    \leq&\, 2(\|\II_k\|_F^2 + \|\II-\mathbf{\Omega}^{1/2} \widehat{\mathbf{\Omega}}^{-1}\mathbf{\Omega}^{1/2}\|_F^2) \\
    =&\,  2(k + \|\mathbf{\Omega}^{1/2}(\mathbf{\Omega}^{-1}- \widehat{\mathbf{\Omega}}^{-1})\mathbf{\Omega}^{1/2}\|_F^2).
    \end{aligned}
\end{equation*}
Now use the property $\|\AA \BB\|_F \leq \|\AA\|\|\BB\|_F$ and $\|\AA\|_F^2 \leq k\|\AA\|^2$, we can further bound the above expression as
\[
2(k+ \|\mathbf{\Omega}\|^2 \|\widehat{\mathbf{\Omega}}^{-1} -\mathbf{\Omega}^{-1}\|_F^2) \leq 2k(1+\|\mathbf{\Omega}\|^2\|\widehat{\mathbf{\Omega}}^{-1} - \mathbf{\Omega}^{-1}\|^2).
\]
(3) 
\[
\|\hat{\Pi}_{\bb} (g)\|_w^2 =\int \ttt_{g}^{\top}\mathbf{\Omega} \widehat{\mathbf{\Omega}}^{-1} \bb(\xx) \bb(\xx)^{\top}\widehat{\mathbf{\Omega}}^{-1}\mathbf{\Omega}\ttt_{g} w(\xx) dF(\xx) = \ttt_{g}^{\top}\mathbf{\Omega} \widehat{\mathbf{\Omega}}^{-1}\mathbf{\Omega} \widehat{\mathbf{\Omega}}^{-1}\mathbf{\Omega}\ttt_{g}.
\]
By Cauchy Schwarz's inequality this is further bounded by
\[
\|\mathbf{\Omega}^{1/2}\ttt_g\| \|\mathbf{\Omega}\|\|\widehat{\mathbf{\Omega}}^{-1}\mathbf{\Omega} \widehat{\mathbf{\Omega}}^{-1}\|\|\mathbf{\Omega}^{1/2}\ttt_g\|.
\]
Again by noting $\|\mathbf{\Omega}^{1/2}\ttt_g\| \leq \|g\|_w$ the proof is completed.

We then prove the theorem. The conditional bias of the $\mathbb{P}_n\left[\widehat{\phi}_{a, 1}^{\text {ge }}(\ZZ)\right]$ is
\[
\ME[\widehat{\phi}_{a, 1}^{\text {ge }}(\ZZ)] - \psi_a = -\int (\hat{\mu}_a(\xx) -\mu_a(\xx)) \left(\frac{1}{\hat{\pi}_a(\xx)\hat{\rho}(\xx)} - \frac{1}{\pi_a(\xx)\rho(\xx)} \right) \pi_a(\xx)\rho(\xx) dF(\xx).
\]
The conditional expectation of the estimated second-order influence function is 
\begin{equation*}
\begin{aligned}
    \ME[\widehat{\phi}_{a, 2}^{\text {ge }}(\ZZ_1,\ZZ_2)] =&\,  -\int \pi_a(\xx)\rho(\xx) (\mu_a(\xx)-\hat{\mu}_a(\xx))\bb(\xx)^{\top} \widehat{\mathbf{\Omega}}^{-1} \int \bb(\uu) \left(\frac{1}{\hat{\pi}_a(\uu)\hat{\rho}(\uu)} - \frac{1}{\pi_a(\uu)\rho(\uu)} \right)\pi_a(\uu)\rho(\uu) d \uu \\
    =&\, -\int \rho(\xx)\pi_a(\xx)(\mu_a(\xx)-\hat{\mu}_a(\xx)) \hat{\Pi}_{\bb} (\frac{1}{\hat{\pi}_a\hat{\rho}} -\frac{1}{\pi_a\rho}) (\xx) dF(\xx).
\end{aligned}
\end{equation*}
Thus the conditional bias of quadratic estimator is
\begin{equation*}
    \begin{aligned}
    &\, \ME[\widehat{\phi}_{a, 1}^{\text {ge }}(\ZZ)]+ \ME[\widehat{\phi}_{a, 2}^{\text {ge }}(\ZZ_1,\ZZ_2)] - \psi_a \\
    =&\, -\int \rho(\xx)\pi_a(\xx)(\hat{\mu}_a(\xx)-\mu_a(\xx)) (\II-\hat{\Pi}_{\bb}) (\frac{1}{\hat{\pi}_a\hat{\rho}} -\frac{1}{\pi_a\rho}) (\xx) dF(\xx) \\
    =&\, -\int \rho(\xx)\pi_a(\xx)(\hat{\mu}_a(\xx)-\mu_a(\xx)) (\II-\Pi_{\bb}) (\frac{1}{\hat{\pi}_a\hat{\rho}} -\frac{1}{\pi_a\rho}) (\xx) dF(\xx) \\
    &\, + \int \rho(\xx)\pi_a(\xx)(\hat{\mu}_a(\xx)-\mu_a(\xx)) (\hat{\Pi}_{\bb}-\Pi_{\bb}) (\frac{1}{\hat{\pi}_a\hat{\rho}} -\frac{1}{\pi_a\rho}) (\xx) dF(\xx) \\
    =&\, -\int (\II-\Pi_{\bb})(\hat{\mu}_a-\mu_a)(\xx) (\II-\Pi_{\bb}) (\frac{1}{\hat{\pi}_a\hat{\rho}} -\frac{1}{\pi_a\rho}) (\xx) \rho(\xx)\pi_a(\xx) dF(\xx) \\
    &\,+ \int \rho(\xx)\pi_a(\xx)(\hat{\mu}_a(\xx)-\mu_a(\xx)) (\hat{\Pi}_{\bb}-\Pi_{\bb}) (\frac{1}{\hat{\pi}_a\hat{\rho}} -\frac{1}{\pi_a\rho}) (\xx) dF(\xx). \\
    \end{aligned}
\end{equation*}
By Cauchy-Schwarz inequality, the first term can be bounded as
\[
\left\|\left(\II-\Pi_{\bb}\right)\left(\frac{1}{\widehat{\rho} \widehat{\pi}_a}-\frac{1}{\rho \pi_a}\right)\right\|_w\left\|\left(\II-\Pi_{\bb}\right)\left(\widehat{\mu}_a-\mu_a\right)\right\|_w.
\]
By (1) in Lemma \ref{lemma-basis}, the second term can be bounded as 
\[
\left\| \frac{1}{\hat{\rho}\hat{\pi}_a} - \frac{1}{\rho\pi_a} \right\|_w \left\|  \hat{\mu}_a - \mu_a \right\|_w \|\widehat{\mathbf{\Omega}}^{-1} - \mathbf{\Omega}^{-1}\|\|\mathbf{\Omega}\|.
\]
For the bound on the variance, we will apply results on variance bounds of U-Statistics.
\[
\hat{\phi}_{a,2}^{\text{ge}}(\ZZ_1) =  \ME[\hat{\phi}_{a,2}^{\text{ge}}(\ZZ_1,\ZZ_2)|\ZZ_1] = -I(S=1,A_1=a)(Y_1-\hat{\mu}_a(\XX_1)) \hat{\Pi}_{\bb} \left( \frac{1}{\hat{\rho}\hat{\pi}_a} - \frac{1}{\rho\pi_a} \right) (\XX_1).
\]
By (3) in lemma \ref{lemma-basis} we have
\[
\ME[(\hat{\phi}_{a,2}^{\text{ge}}(\ZZ_1))^2] \lesssim \left\|\hat{\Pi}_{\bb} \left( \frac{1}{\hat{\rho}\hat{\pi}_a} - \frac{1}{\rho\pi_a} \right) \right\|_w^2 \lesssim \left\| \left( \frac{1}{\hat{\rho}\hat{\pi}_a} - \frac{1}{\rho\pi_a} \right)\right\|_w^2 \lesssim 1.
\]
Furthermore
\begin{equation*}
    \begin{aligned}
    \ME[(\hat{\phi}_{a,2}^{\text{ge}}(\ZZ_1,\ZZ_2))^2] \lesssim&\, \int \left( \bb(\xx_1)^{\top} \widehat{\mathbf{\Omega}}^{-1} \bb(\xx_2) \right)^2 \pi_a(\xx_1)\rho(\xx_1) \pi_a(\xx_2)\rho(\xx_2) dF(\xx_1) dF(\xx_2) \\
    \leq &\, 2k(1+\|\mathbf{\Omega}\|^2\|\widehat{\mathbf{\Omega}}^{-1} - \mathbf{\Omega}^{-1}\|^2),
    \end{aligned}
\end{equation*}
where the last inequality follows from (2) in Lemma \ref{lemma-basis}. Since all the terms are bounded, the above bounds still hold if we consider the term $\hat{\phi}_{a,1}^{\text{ge}}(\ZZ_1)+ \hat{\phi}_{a,2}^{\text{ge}}(\ZZ_1,\ZZ_2)$ (the constant may change). Finally by Lemma 6 in \cite{robins2009quadratic} we have the conditional variance of the quadratic estimator is bounded by $O(1/n+k/n^2)$.

\end{proof}

\section{Proof of Theorem \ref{thm-qr-transportation}}
The conditional bias of the $\mathbb{P}_n\left[\widehat{\phi}_{a, 1}^{\text {tr }}(\ZZ)\right]$ is
\[
\ME[\widehat{\phi}_{a, 1}^{\text {tr }}(\ZZ)] - \psi_a = -\int (\hat{\mu}_a(\xx) -\mu_a(\xx)) \left(\frac{1-\hat{\rho}(\xx)}{\hat{\pi}_a(\xx)\hat{\rho}(\xx)} - \frac{1-\rho(\xx)}{\pi_a(\xx)\rho(\xx)} \right) \pi_a(\xx)\rho(\xx) dF(\xx).
\]
The conditional expectation of the second-order term is 
\begin{equation*}
    \begin{aligned}
    \ME[\widehat{\phi}_{a, 2}^{\text {tr }}(\ZZ_1,\ZZ_2)]=&\, \int \rho(\xx)\pi_a(\xx) (\hat{\mu}_a(\xx)-\mu_a(\xx)) \bb(\xx)^{\top}\widehat{\mathbf{\Omega}}^{-1}\int \bb(\uu) \left(\frac{1-\hat{\rho}(\uu)}{\hat{\pi}_a(\uu)\hat{\rho}(\uu)} - \frac{1-\rho(\uu)}{\pi_a(\uu)\rho(\uu)} \right) \pi_a(\uu)\rho(\uu) dF(\uu) \\
    =&\, \int \rho(\xx)\pi_a(\xx) (\hat{\mu}_a(\xx)-\mu_a(\xx)) \hat{\Pi}_{\bb} \left( \frac{1-\hat{\rho}}{\hat{\pi}_a \hat{\rho}} -\frac{1-\rho}{\pi_a\rho} \right) (\xx) dF(\xx).
    \end{aligned}
\end{equation*}
Thus the conditional bias of quadratic estimator is
\begin{equation*}
    \begin{aligned}
    &\, \ME[\widehat{\phi}_{a, 1}^{\text {tr }}(\ZZ)]+ \ME[\widehat{\phi}_{a, 2}^{\text {tr }}(\ZZ_1,\ZZ_2)] - \theta_a \\
    =&\, -\int \rho(\xx)\pi_a(\xx)(\hat{\mu}_a(\xx)-\mu_a(\xx)) (\II-\hat{\Pi}_{\bb})\left( \frac{1-\hat{\rho}}{\hat{\pi}_a \hat{\rho}} -\frac{1-\rho}{\pi_a\rho} \right) (\xx) dF(\xx) \\
    =&\, -\int \rho(\xx)\pi_a(\xx)(\hat{\mu}_a(\xx)-\mu_a(\xx)) (\II-\Pi_{\bb}) \left( \frac{1-\hat{\rho}}{\hat{\pi}_a \hat{\rho}} -\frac{1-\rho}{\pi_a\rho} \right) (\xx) dF(\xx) \\
    &\, + \int \rho(\xx)\pi_a(\xx)(\hat{\mu}_a(\xx)-\mu_a(\xx)) (\hat{\Pi}_{\bb}-\Pi_{\bb}) \left( \frac{1-\hat{\rho}}{\hat{\pi}_a \hat{\rho}} -\frac{1-\rho}{\pi_a\rho} \right) (\xx) dF(\xx) \\
    =&\, -\int (\II-\Pi_{\bb})(\hat{\mu}_a-\mu_a)(\xx) (\II-\Pi_{\bb}) \left( \frac{1-\hat{\rho}}{\hat{\pi}_a \hat{\rho}} -\frac{1-\rho}{\pi_a\rho} \right) (\xx) \rho(\xx)\pi_a(\xx) dF(\xx) \\
    &\,+ \int \rho(\xx)\pi_a(\xx)(\hat{\mu}_a(\xx)-\mu_a(\xx)) (\hat{\Pi}_{\bb}-\Pi_{\bb}) \left( \frac{1-\hat{\rho}}{\hat{\pi}_a \hat{\rho}} -\frac{1-\rho}{\pi_a\rho} \right) (\xx) dF(\xx). \\
    \end{aligned}
\end{equation*}
By Cauchy-Schwarz inequality, the first term can be bounded as
\[
\left\|\left(\II-\Pi_{\bb}\right)\left( \frac{1-\hat{\rho}}{\hat{\pi}_a \hat{\rho}} -\frac{1-\rho}{\pi_a\rho} \right)\right\|_w\left\|\left(\II-\Pi_{\bb}\right)\left(\widehat{\mu}_a-\mu_a\right)\right\|_w.
\]
By (1) in Lemma \ref{lemma-basis}, the second term can be bounded as 
\[
\left\| \frac{1-\hat{\rho}}{\hat{\pi}_a \hat{\rho}} -\frac{1-\rho}{\pi_a\rho} \right\|_w \left\|  \hat{\mu}_a - \mu_a \right\|_w \|\widehat{\mathbf{\Omega}}^{-1} - \mathbf{\Omega}^{-1}\|\|\mathbf{\Omega}\|.
\]
For the bound on the variance, 
\[
\hat{\phi}_{a,2}^{\text{tr}}(\ZZ_1) =  \ME[\hat{\phi}_{a,2}^{\text{tr}}(\ZZ_1,\ZZ_2)|\ZZ_1] = -I(S=1,A_1=a)(Y_1-\hat{\mu}_a(\XX_1)) \hat{\Pi}_{\bb} \left( \frac{1-\hat{\rho}}{\hat{\pi}_a \hat{\rho}} -\frac{1-\rho}{\pi_a\rho} \right) (\XX_1).
\]
By (3) in Lemma \ref{lemma-basis} we have
\[
\ME[(\hat{\phi}_{a,2}^{\text{tr}}(\ZZ_1))^2] \lesssim \left\|\hat{\Pi}_{\bb} \left( \frac{1-\hat{\rho}}{\hat{\pi}_a \hat{\rho}} -\frac{1-\rho}{\pi_a\rho} \right) \right\|_w^2 \lesssim \left\| \left( \frac{1-\hat{\rho}}{\hat{\pi}_a \hat{\rho}} -\frac{1-\rho}{\pi_a\rho} \right)\right\|_w^2 \lesssim 1.
\]
Furthermore
\begin{equation*}
    \begin{aligned}
    \ME[(\hat{\phi}_{a,2}^{\text{tr}}(\ZZ_1,\ZZ_2))^2] \lesssim&\, \int \left( \bb(\xx_1)^{\top} \widehat{\mathbf{\Omega}}^{-1} \bb(\xx_2) \right)^2 \pi_a(\xx_1)\rho(\xx_1) \pi_a(\xx_2)\rho(\xx_2) dF(\xx_1) dF(\xx_2) \\
    \leq &\, 2k(1+\|\mathbf{\Omega}\|^2\|\widehat{\mathbf{\Omega}}^{-1} - \mathbf{\Omega}^{-1}\|^2),
    \end{aligned}
\end{equation*}
where the last inequality follows from (2) in Lemma \ref{lemma-basis}. Since all the terms are bounded, the above bounds still hold if we consider the term $\hat{\phi}_{a,1}^{\text{tr}}(\ZZ_1)+ \hat{\phi}_{a,2}^{\text{tr}}(\ZZ_1,\ZZ_2)$ (the constant may change). Finally by Lemma 6 in \cite{robins2009quadratic} we have the conditional variance of the quadratic estimator is bounded by $O(1/n+k/n^2)$.

\section{Details in Real Data Analysis}

\subsection{Difference in Covariates Distributions and Possible Effect Modifiers}\label{sec:Real-data-difference}

We first examine the difference in covariate distributions. Since most variables are discrete we will present the bar plots. All the missing values are omitted. The bar plots of education level and race are shown in Figure \ref{fig:education} and Figure \ref{fig:race}. One can perform Pearson's Chi-square tests to test whether there exists a significant difference in the distributions of education level and race from different datasets. The p-values are both smaller than $2.2 \times 10^{-16}$ and we conclude the difference in distribution is significant. For brevity, we only present two variables. From our analysis, similar results applied to most variables, indicating the distributions of most variables from two populations are different. 
\begin{figure}[H]
	\centering
	\subfigure[Numom]{
		\begin{minipage}[t]{0.45\linewidth}
			\centering
			\includegraphics[width=2.5in]{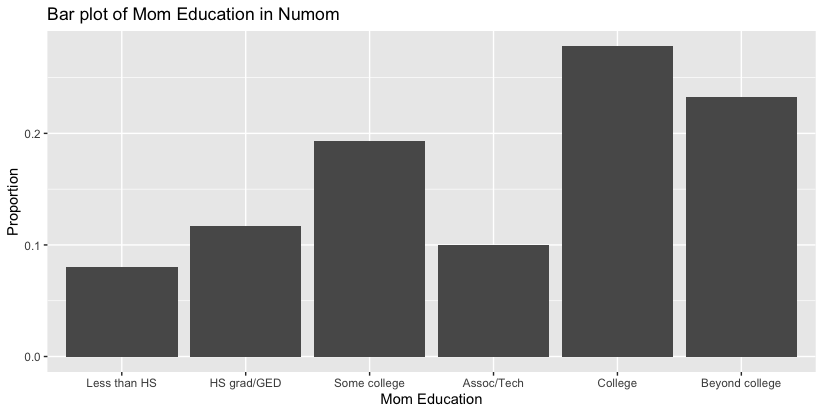}
	\end{minipage}}
	\subfigure[NSFG]{
		\begin{minipage}[t]{0.45\linewidth}
			\centering
			\includegraphics[width=2.5in]{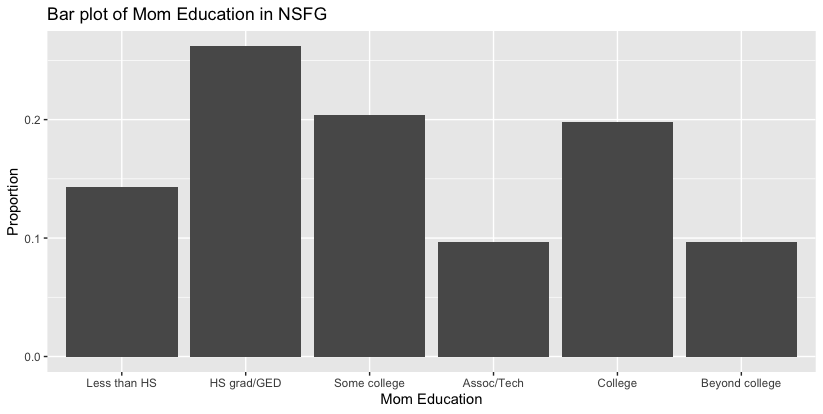}
	\end{minipage}}
	\centering
	\caption{Bar plots of education level in the two datasets}
	\label{fig:education}
\end{figure}

\begin{figure}[H]
	\centering
	\subfigure[Numom]{
		\begin{minipage}[t]{0.45\linewidth}
			\centering
			\includegraphics[width=2.5in]{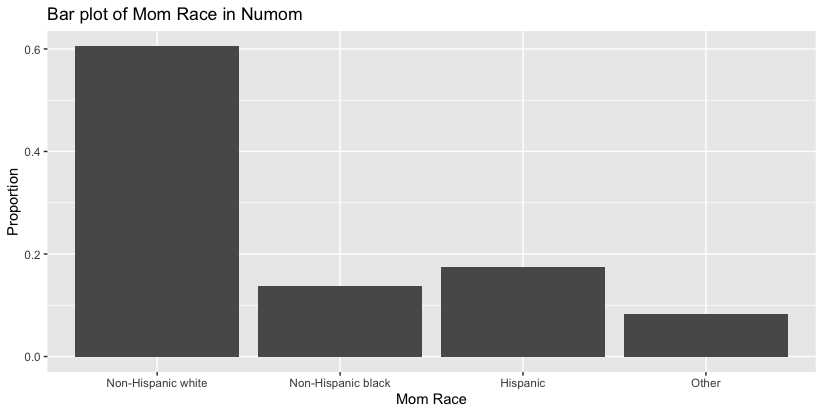}
	\end{minipage}}
	\subfigure[NSFG]{
		\begin{minipage}[t]{0.45\linewidth}
			\centering
			\includegraphics[width=2.5in]{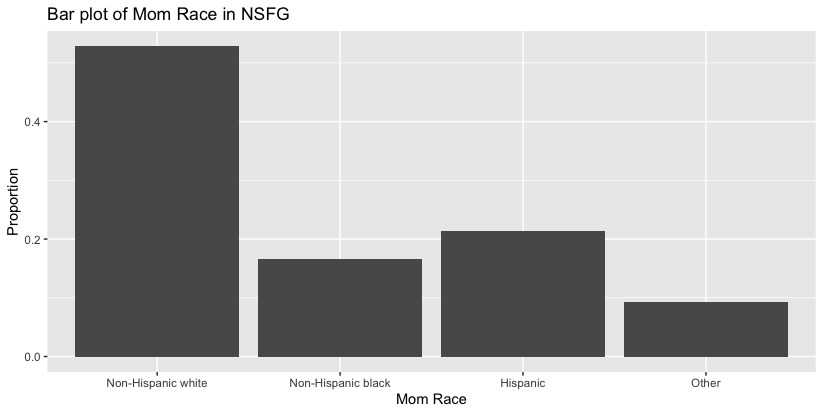}
	\end{minipage}}
	\centering
	\caption{Bar plots of race in the two datasets}
	\label{fig:race}
\end{figure}

Then we will focus on the effect of high fruit intake on preterm birth as an example to show the existence of effect modifiers. The covariate we consider is education level. For each subject in the nuMoM2b dataset, we estimate the non-centered influence function for ATE of fruit on preterm birth. We then divide subjects into groups according to their education level and take the average of influence functions among subjects in the same group to estimate the ATE of high fruit intake on preterm birth in each group. The estimates together with 95\% confidence intervals are shown in Figure \ref{fig:effect_modifier}.

\begin{figure}
    \centering
    \includegraphics[width=4in]{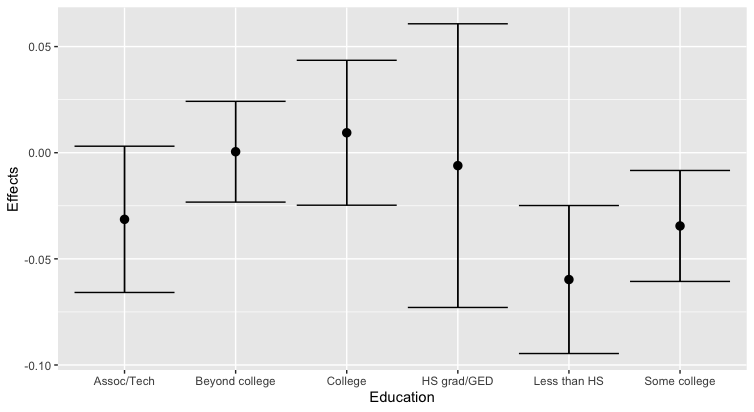}
    \caption{Hetergeneous effects of high fruit intake on preterm birth}
    \label{fig:effect_modifier}
\end{figure}

We can see that the confidence intervals for effects among subjects with education level less than high school and beyond college do not overlap. Similar results can be found if we use some other covariate, treatment and outcome. These results imply heterogeneous effects exist in different levels of covariates and some covariates may modify the effects.

Combining these two observations we found in this part: (1) the distributions of most variables in the two populations are different; and (2) some covariates may modify the effects, so we conclude that the results in \cite{bodnar2020machine} based on nuMoM2b dataset may not generalize to the whole U.S. pregnant female population directly. It is therefore necessary to apply transportation methods to obtain reliable results on the target population.

\subsection{Adjustments by Stratified Cluster Sampling}
The process of stratified cluster sampling is summarized below:

\begin{itemize}
    \item Form $h=1, \dots, H$ strata.
    \item Sample $\alpha = 1,\dots, a_h$ clusters independently from each stratum.
    \item Query feature of interests from all individuals in each cluster $y_{h\alpha i}, i=1,\dots,n_{h \alpha}$.
    \item Weight $w_{h\alpha i}$: this data point can represent $w_{h\alpha i}$ individuals in the population.
\end{itemize}
The estimator for the population mean of $y$ is given by
\begin{equation}\label{weighted-est}
\bar{y}_w=\frac{\sum_{h=1}^{H} \sum_{\alpha=1}^{a_{h}} \sum_{i=1}^{n_{h \alpha}} w_{h \alpha i} y_{h \alpha i}}{\sum_{h=1}^{H} \sum_{\alpha=1}^{a_{h}} \sum_{i=1}^{n _{h \alpha}} w_{h \alpha i}} = \frac{\hat{Y}}{\hat{N}} .    
\end{equation}
The estimated variance is 
\[
\text{Var}\left(\bar{y}_w\right) \doteq \frac{\text{Var}(\hat{Y})+\bar{y}^{2} \times \text{Var}(\hat{N})-2 \times \bar{y} \times \text{Cov}(\hat{Y}, \hat{N})}{\hat{N}^{2}},
\]
where $\text{Var}(\hat{Y}), \text{Var}(\hat{N}), \text{Cov}(\hat{Y}, \hat{N})$ can be estimated from the samples. More detailed discussions on stratified cluster sampling mechanisms can be found in \cite{heeringa2017applied}.

To motivate how to modify the doubly robust estimator, we note that the doubly robust estimator can be written as
\[
\hat{\theta}_a^{dr} = \MP_n \left[ \frac{\hat{\varphi}_a^{\text{tr}}(\ZZ)}{\hat{\MP}(S=0)} \right]  = \frac{n_1}{n} \MP_{n_1} \left[ \frac{\hat{\varphi}_a^{\text{tr}}(\ZZ)}{\hat{\MP}(S=0)} \right]  + \frac{n_2}{n} \MP_{n_2} \left[ \frac{\hat{\varphi}_a^{\text{tr}}(\ZZ)}{\hat{\MP}(S=0)} \right].
\]
Replace $y$ with the non-centered influence function $\hat{\varphi}_a^{\text{tr}}(\ZZ)/\hat{\MP}(S=0)$ in \eqref{weighted-est} and take weighted average in the target dataset (NSFG), we can estimate the mean of potential outcomes among target population as $\hat{\theta}_{a,2} $ with variance $\sigma_{a,2}^2$. In the source population we can use the simple average $\hat{\theta}_{a,1} = \frac{1}{n_1}\sum_{i=1}^{n_1} \hat{\varphi}_a^{\text{tr}}(\ZZ_i)/\hat{\MP}(S=0)$ with variance $\sigma_{a,1}^2$, then we can combine them as
\begin{equation}
    \hat{\theta}_a^{dr} = \frac{n_1}{n} \hat{\theta}_{a,1} + \frac{n_2}{n} \hat{\theta}_{a,2}.
\end{equation}
The variance is $\text{Var}(\hat{\theta}_a) = (n_1^2\sigma_{a,1}^2 + n_2^2 \sigma_{a,2}^2)/n^2$. In practice, both variances can be estimated either by sample variance (for $\sigma_{a,1}^2$) or estimators in stratified cluster sampling (for $\sigma_{a,2}^2$).

\subsection{Model Selection}

We use cross-validation to evaluate the performance of different models. We will use vegetable intake as the treatment and gestational diabetes as the outcome in this part. Consider the following models in the superlearner:
\begin{itemize}
    \item SL.ranger: (mtry, min.node.size, replacement) $\in \{3, 5, 7 \} \times \{5, 10, 20\} \times $ \{with/without replacement\}.
    \item SL.xgboost: (max\_depth, shrinkage) $\in \{4, 6, 8\} \times \{ 0.1, 0.01, 0.001\}$.
    \item SL.glmnet: Use CV to find lambda.min.
    \item SL.ksvm: nu $ \in \{0.1, 0.3, 0.5, 0.7\}$.
\end{itemize}
Five-fold cross-fitting is used to obtain the estimated values of nuisance functions for each subject. We then use the square loss to measure the performance of different models. The average loss together with the 95\% confidence interval for each model is summarized in Figure \ref{fig:cv}. The x-axis indicates the models we are using. For instance, ``glm+ranger'' is the model we used to plot Figure \ref{fig:effects} (where in ``ranger'' we just use the default tuning parameters). ``glm+rf'' is the superlearner including all random forest models with different tuning parameters specified above and ``glm+rf+svm+xg'' is the superlearner including all models with different tuning parameters. From the cross-validation results we conclude including more complicated models does not significantly improve the performance. Under such circumstances, people usually use the simple model for generalization purpose. This justifies the use of the simplest model ``glm+ranger'' in our previous analysis.

\begin{figure}[H]
	\centering
	\subfigure[$\mu_0(\xx)$]{
		\begin{minipage}[t]{0.45\linewidth}
			\centering
			\includegraphics[width=3in]{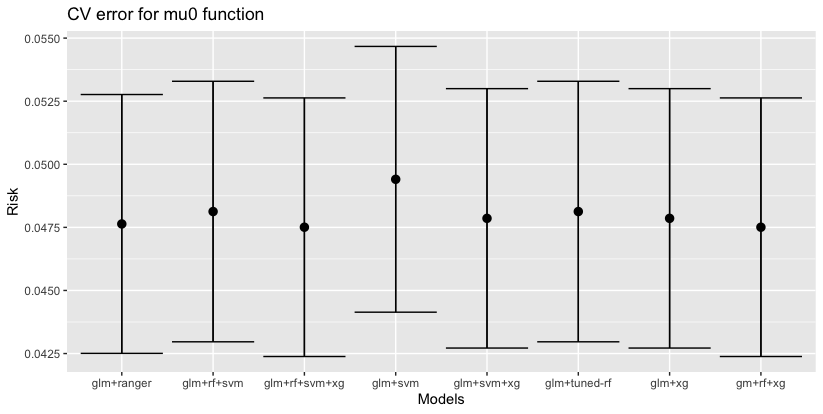}
	\end{minipage}}
	\subfigure[$\mu_1(\xx)$]{
		\begin{minipage}[t]{0.45\linewidth}
			\centering
			\includegraphics[width=3in]{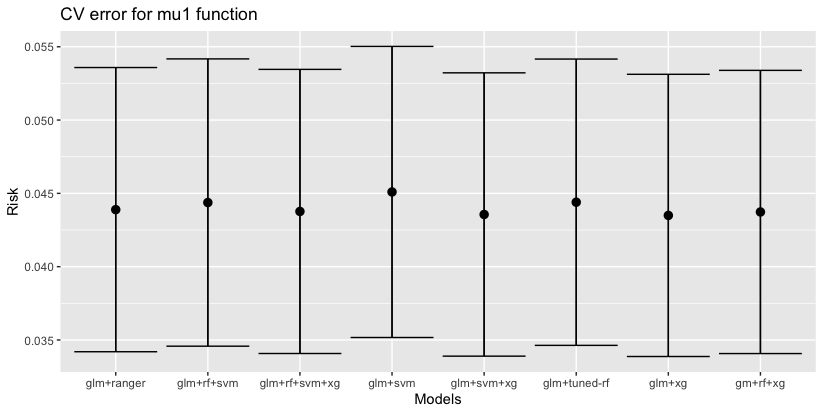}
	\end{minipage}}\\
	\subfigure[$\tau_0(\vv)$]{
		\begin{minipage}[t]{0.45\linewidth}
			\centering
			\includegraphics[width=3in]{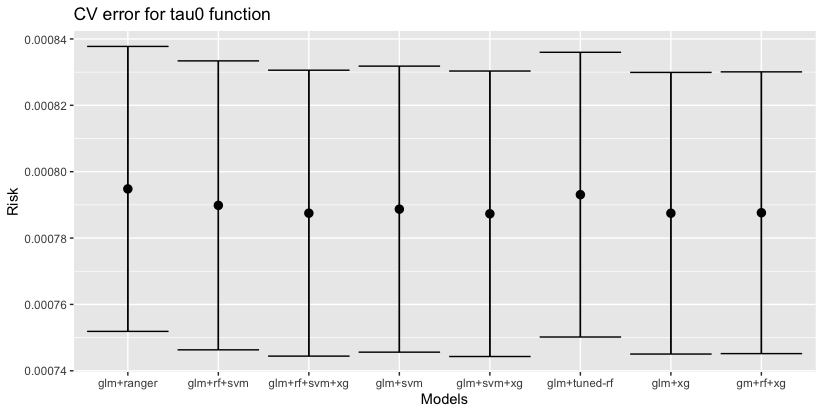}
	\end{minipage}}
	\subfigure[$\tau_1(\vv)$]{
		\begin{minipage}[t]{0.45\linewidth}
			\centering
			\includegraphics[width=3in]{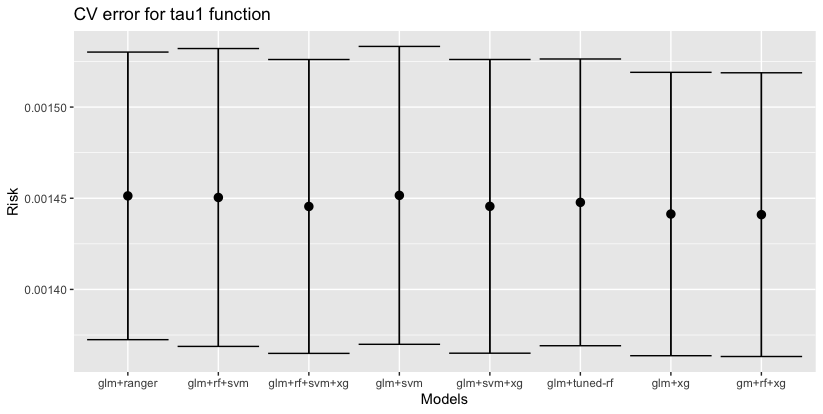}
	\end{minipage}}\\
	\subfigure[$\pi(\xx)$]{
		\begin{minipage}[t]{0.45\linewidth}
			\centering
			\includegraphics[width=3in]{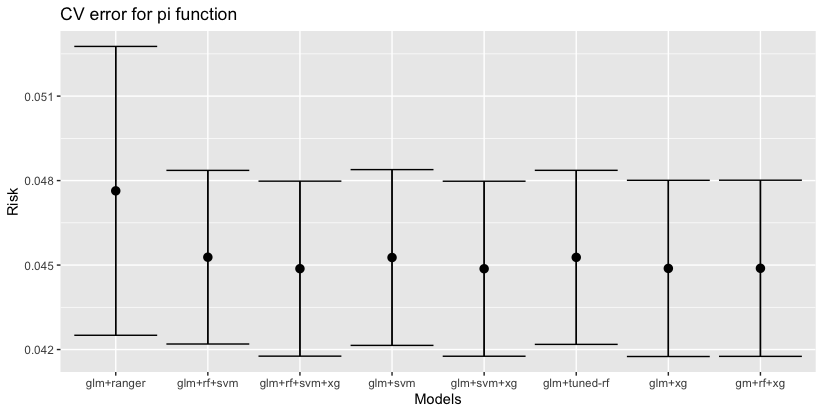}
	\end{minipage}}
	\subfigure[$\rho(\vv)$]{
		\begin{minipage}[t]{0.45\linewidth}
			\centering
			\includegraphics[width=3in]{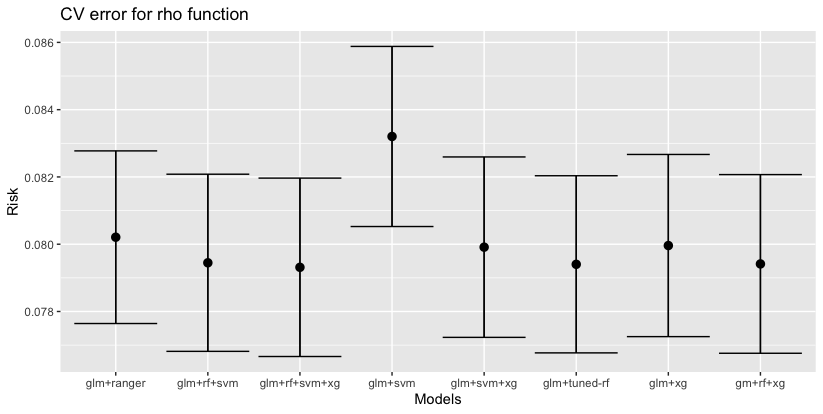}
	\end{minipage}}\\
	\centering
	\caption{Cross-validation results. Treatment is vegetable, outcome is gestational diabetes.}
	\label{fig:cv}
\end{figure}

\subsection{Choice of Positivity Constant}

We then try different choices of positivity parameter $\epsilon$. We set $\epsilon$ to be a sequence from 0.01 to 0.15 with an increment of 0.01. The 15 point estimates of the effects in the target population together with the confidence intervals are presented in Figure  \ref{fig:positive}. We see there is no significant difference in the results with different choices of $\epsilon$. To obtain more stable results, one may want to use a large $\epsilon$ so that the denominators in the doubly robust estimators are not too small. However, a small $\epsilon$ respects the fitted model (it will modify fewer estimated probabilities) and hence the results will be closer to the truth in some sense. So there is a trade-off in the choice of $\epsilon$. In our analysis, the results are fairly stable when using different $\epsilon$'s. So we will stick to the original choice $\epsilon = 0.01$.
\begin{figure}
    \centering
    \includegraphics[width=4in]{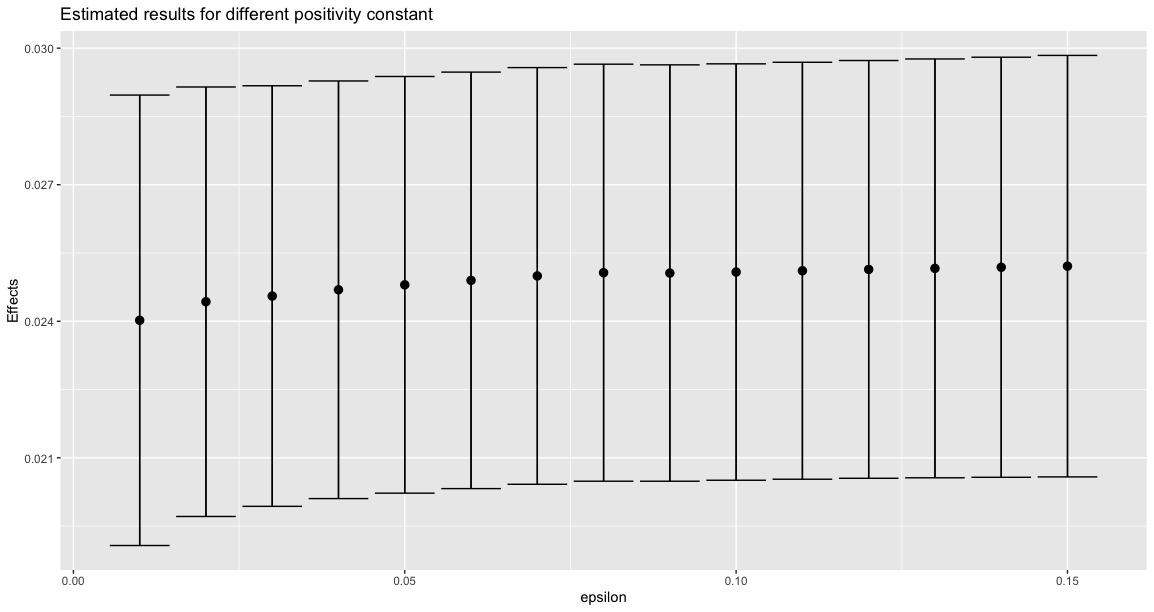}
    \caption{Results with different choice of $\epsilon$. Treatment is vegetable, outcome is gestational
diabetes.}
    \label{fig:positive}
\end{figure}

\subsection{Potential Outliers}

Finally we evaluate the impacts of potential outliers in our analysis. We first estimate the non-centered influence functions for each data point. For each $q$ in the sequence from 0.004 to 0.014 with increment 0.002, we only include samples with influence functions greater than $q$-quantile and smaller than $(1-q)$-quantile of all the empirical influence functions in the subsequent analysis (i.e. we only take an average over these empirical influence function values which are not very extreme). The estimated results are presented in Figure \ref{fig:outlier}. From Figure \ref{fig:outlier} it seems that all the 95\% confidence intervals overlap and our results are quite robust against outliers in the empirical influence functions.
\begin{figure}
    \centering
    \includegraphics[width=4in]{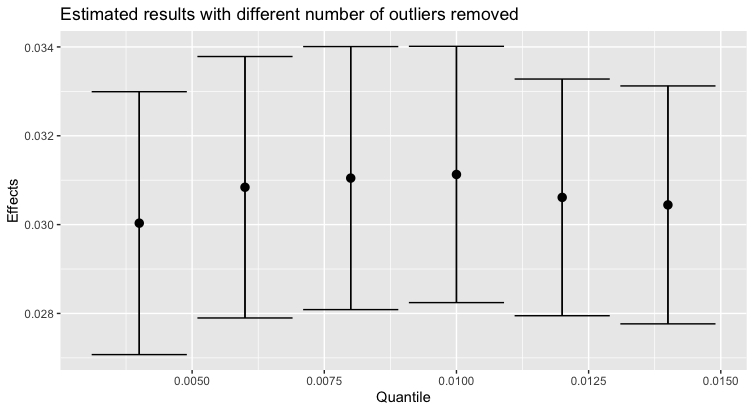}
    \caption{Assessing the effects of outliers in empirical influence functions}
    \label{fig:outlier}
\end{figure}

\end{document}